\documentclass{dis}
\usepackage{latexsym,ifsym,tocbibind}
\usepackage{url}%,
\usepackage{cite}%sort}
\usepackage{deleq,color}
\usepackage{inputenc}
\usepackage{relsize, exscale}
\usepackage{amsmath,amssymb,amsfonts,amscd,graphicx}

\usepackage{psfrag}
\usepackage{epstopdf}

\newtheorem{Theorem}{Theorem}[chapter]
\newtheorem{Proposition}[Theorem]{Proposition}
\newtheorem{Lemma}[Theorem]{Lemma}
\newtheorem{Corollary}[Theorem]{Corollary}

\theoremstyle{Definition}

\newtheorem{Remark}[Theorem]{Remark}

\newcommand{\newg}{G} % this used to be g, but this is in conflict with the metric
\newcommand{\og}{{\overline{g}}}

\newcommand{\coneg}{\tilde g}

\newcommand{\wmcN}{\,\,\,\widehat{\!\!\!\mcN}}

\newcommand{\rmnote}[1]{}%{\mnote{#1}}

\newcommand \al {\alpha}

\newcommand \mcU {{\mycal U}}
\newcommand \mcV {{\mycal V}}

\newcommand \mcO {{\mycal O}}

\newcommand{\myqed} {\hfill $\Box$}

\newcommand{\mcE}{{\mycal E}}
\newcommand{\mcM}{{\mycal M}}
\newcommand{\mcP}{{\mycal P}}

\newcommand{\Id}{\mbox{\rm Id}} %identity matrix
\newcommand{\const}{\mbox{\rm const}} %constants

\newcommand{\hchit}{{}^t{}\hchi}
\newcommand{\hchibt}{{}^t{}\hchib}

\newcommand{\bean}{\begin{eqnarray}\nonumber}
\newcommand{\beal}[1]{\begin{eqnarray}\label{#1}}
\newcommand{\eeal}[1]{\label{#1}\end{eqnarray}}
\newcommand{\bel}[1]{\begin{equation}\label{#1}}

\newcommand{\tmcM}{\,\,\,\,\widetilde{\!\!\!\!\mcM}}

\newcommand{\hJ}{{\hat J}}
\newcommand{\sthd}{{}^{\star_g}}
\newcommand{\ts}{{\cal S}} %pseudo-tangent space

\newcommand{\hD}{\widehat D}
\newcommand{\hxi}{\hat \xi}
\newcommand{\hxib}{\hat{\xib}}
\newcommand{\heta}{\hat \eta}
\newcommand{\hetab}{\hat{\etab}}
\newcommand{\home}{\hat{\omega}}
\newcommand{\homb}{\hat{\underline\omega}}

\newcommand{\obe}{\mathring{\beta}}
\newcommand{\obeta}{\obe}
\newcommand{\obetab}{\oubeta}
\newcommand{\oube}{\,\,\mathring{\!\!\ubeta}}
\newcommand{\oubeta}{\oube}
\newcommand{\osi}{\mathring{\sigma}}
\newcommand{\osigma}{\osi}
\newcommand{\orho}{\mathring{\rho}}
\newcommand{\oro}{\orho}

\newcommand{\hups}{\hat{\ups}}
\newcommand{\hupsb}{\hat{\upsb}}

\newcommand{\KlNhat}[1]{\overline{#1}}
\newcommand{\KlNwidehat}[1]{\overline{#1}}

{\catcode `\@=11 \global\let\AddToReset=\@addtoreset}
\AddToReset{figure}{chapter}
\renewcommand{\thefigure}{\thechapter.\arabic{figure}}

\newcommand{\tr}{\mbox{tr}}

\newcommand{\be}{\begin{equation}}
\newcommand{\bea}{\begin{eqnarray}}
\newcommand{\eea}{\end{eqnarray}}
\newcommand{\beaa}{\begin{eqnarray*}}
\newcommand{\eeaa}{\end{eqnarray*}}
\newcommand{\bseq}{\begin{subeq}}
\newcommand{\eseq}{\end{subeq}}
\newcommand{\eql}[1]{\arrlabel{#1}}
\newcommand{\alp}{\alpha}

\def \rectangle#1#2{\hbox{\vrule\vbox to #2
{\hrule\hbox to
#1{\hfil}\vfil\hrule}\vrule}}
\newcommand{\edd}{\end{document}}

\newcommand{\dd}{\mbox{${\hD}$}}

\newcommand{\nabb}{\mbox{$\nabla \mkern-13mu /$\,}}
\newcommand{\hnabb}{\mbox{$\hnabla \mkern-13mu /$\,}}

\newcommand{\hot}{\KlNwidehat{\otimes}_{\mathit{s}}}

\newcommand{\nn}{\nonumber}
\newcommand{\nnn}{}
\newcommand{\wdd}{d} % macro for a candidate weyl tensor

\newcommand{\chib}{\underline{\chi}{}}
\newcommand{\hchi}{\hat{{\chi}}}
\newcommand{\hchib}{\hat{\underline{\chi}}{}}

\newcommand{\hzeb}{\hat{\zetab}}

\newcommand{\Vol}{{\mathrm{Vol}}_g}     %volume
\newcommand{\xib}{\underline{\xi}}

\newcommand{\D}{{\cal D}}
\newcommand{\M}{{\mycal M}}

\newcommand{\ua}{{\underline{\alpha}}{}}
\newcommand{\ualpha}{{\ua}}
\newcommand{\ualp}{{\ua}}
\newcommand{\bb}{{\underline{\beta}}{}}
\newcommand{\ub}{{\bb}}
\newcommand{\ubeta}{{\bb}}
\newcommand{\dual}{\mbox{}^{\star}\!}

\newcommand{\si}{\sigma}

\newcommand{\ro}{\rho}

\newcommand{\ze}{\zeta}
\newcommand{\hze}{\hat\zeta}
\newcommand{\divv}{\mbox{div}\mkern-19mu /\,\,\,\,}
\newcommand{\hdivv}{\widehat{\mathrm{div}}\mkern-19mu /\,\,\,\,}

\newcommand{\vp}{\varphi}

\newcommand{\ddd}{\nabb}

\newcommand{\hdddd}{{\bf \hat D} \mkern-13mu /\,}

\def\frac#1#2{{{#1}\over{#2}}}
\renewcommand{\qed}{\hfill $\Box$\bigskip}
\newcommand{\eeq}{\end{equation}}
\newcommand{\ee}{\end{equation}}
\newcommand{\beqa}{\begin{eqnarray}}
\newcommand{\beqar}{\begin{deqarr}}
\newcommand{\beqarn}{\begin{deqarr}\nonumber}
\newcommand{\beqarl}[1]{\begin{deqarr}\label{#1}}
\newcommand{\eeqa}{\end{eqnarray}}
\newcommand{\eeqar}{\end{deqarr}}
\newcommand{\eeqarl}[2]{\label{#1}\arrlabel{#2}\end{deqarr}}
\newcommand{\beqan}{\begin{eqnarray*}}
\newcommand{\eeqan}{\end{eqnarray*}}
\newcommand{\ba}{\begin{array}}
\newcommand{\ea}{\end{array}}

\newcommand{\mcN}{{\mycal N}}

\newcommand{\DD}{D}

\DeclareFontFamily{OT1}{rsfs}{} \DeclareFontShape{OT1}{rsfs}{m}{n}{
 <-7> rsfs5 <7-10> rsfs7 <10-> rsfs10}{}
\DeclareMathAlphabet{\mycal}{OT1}{rsfs}{m}{n}

\newcounter{mnotecount}[section]

\newcommand{\R}{\mathbb R}
\newcommand{\N}{\mathbb N}

\newcommand{\eq}[1]{(\ref{#1})}

\newcommand{\Eq}[1]{Equation~\eq{#1}}
\newcommand{\Eqs}[2]{Equations~\eq{#1}-\eq{#2}}

\newcommand{\fij}{f_{i_j}}
\newcommand{\Cdiv}{C_{\mathrm{div}}}
\newcommand{\hCdiv}{C_{\mathrm{div,\psi\psi}}}
\newcommand{\hhCdiv}{C_{\mathrm{div,\varphi\varphi}}}
\newcommand{\tCvarphi}{\tilde C{}_\varphi}
\newcommand{\tCpsi}{\tilde C{}_\psi}

\newcommand{\zC}{\mathring{C}}
\newcommand{\zA}{\mathring{A}}
\newcommand{\zGamma}{\mathring{\Gamma}}

\newcommand{\znabla}{\mathring{\nabla}}
\newcommand{\ywi}[1]{\|#1\|_{W^{1,\infty}(Y)}}
\newcommand{\ywti}[1]{\|#1\|_{W^{2,\infty}(Y)}}
\newcommand{\ylin}[1]{\|#1\|_{L^\infty(\mcN^+\cup \mcN^-)}}
\newcommand{\yli}[1]{\|#1\|_{L^\infty(Y)}}
\newcommand{\yltwo}[1]{\|#1\|_{L^2(Y)}}
\newcommand{\ykn}[1]{\|#1\|_{H^k(Y)}}
\newcommand{\ykmtn}[1]{\|#1\|_{H^{k-2}(Y)}}
\newcommand{\ykmthn}[1]{\|#1\|_{H^{k-3}(Y)}}

\newcommand{\ykns}[1]{\|#1\|^2_{H^k(Y)}}
\newcommand{\ykmns}[1]{\|#1\|^2_{H^{k-1}(Y)}}

   \newcommand{\mcK}{{\mycal K}}
   \newcommand{\mcC}{{\mycal C}}
   
\begin{document}

\mathclass{Primary 35L45; Secondary 35L52, 58J45, 83C05.}

\thanks{
Partly based upon work supported by the National Science Foundation
under Grant No. 0932078 000, while two of the authors (PTC, RTW)
were in residence at the Mathematical Science Research Institute in
Berkeley, California, during the fall semester  of 2014. Part of
work on this paper has been carried-out at the Erwin Schr\"odinger
Institute, Vienna. Partially supported by Narodowe Centrum Nauki
(Poland) under the grant DEC-2011/03/B/ST1/02625 and the Austrian
Science Fund (FWF) under project P 23719-N16.  RTW acknowledges
financial support from the Berlin Mathematical School (BMS).
PTC acknowledges useful discussions with Helmut Friedrich.  RTW is grateful to  the group ``Geometry in Potsdam" and its members for hospitality during part of work of this paper.
}

\abbrevauthors{A. Cabet, P.T. Chru\'sciel and R.T. Wafo}
\abbrevtitle{Characteristic initial value problem}

\title{On the characteristic initial value problem for nonlinear symmetric hyperbolic
systems, including Einstein equations}

\author{Aurore Cabet}
\address{D\'epartement d'Informatique
\\
Universit\'e de Tours, France
\\
E-mail: aurore.cabet@univ-tours.fr}

\author{Piotr T. Chru\'sciel}
\address{Faculty of Physics and Erwin Schr\"odinger Institute, Universit\"at Wien, Austria\\
homepage.univie.ac.at/piotr.chrusciel\\
E-mail: piotr.chrusciel@univie.ac.at}

\author{Roger Tagne Wafo}
\address{University of Douala, Cameroon\\
Faculty of Science\\
Department of Mathematics and Computer Science
\\
E-mail:
rtagnewafo@yahoo.com}

\maketitledis

\tableofcontents
\begin{abstract}
We consider a characteristic initial value problem for a class of
symmetric hyperbolic systems with initial data given on two smooth
null intersecting characteristic surfaces. We prove existence of
solutions on a future neighborhood of the initial surfaces. The
result is applied to general semilinear wave equations, as
well as the Einstein equations with or without sources, and
conformal variations thereof.
\end{abstract}

\makeabstract

\chapter{Introduction}

\label{sec:intro}

There are several reasons why a characteristic Cauchy problem is of
interest in general relativity. First, the general relativistic
constraint equations on characteristic surfaces are trivial to solve
(see, e.g., \cite{ChPaetz,CCM2,RendallCIVP}), while they are not on
spacelike ones. Thus, a good understanding of the characteristic
Cauchy problem is likely to provide more flexibility in constructing
space-times with interesting properties. Next, an observer can in
principle measure the initial data on her past light cone, and use
those to determine the physical fields throughout her past by
solving the field equations backwards in time; on the other hand,
initial data on a spacelike surface near the observer can not be
measured instantaneously. Finally, Friedrich's conformal field
equations may be used to construct space-times using initial data
prescribed on past null
infinity~\cite{ChPaetz2,FriedrichNullData,FriedrichCMP86} which, at
least in some situations, is a null cone emerging from a single
point representing past timelike infinity.

 The characteristic initial value
problem for the vacuum Einstein equations with initial data given on
two smooth null intersecting hypersurfaces has been studied by
several authors
\cite{CagnacEinsteinCRAS1,CagnacEinsteinCRAS2,SachsCIVP,Dautcourt,%
PenroseCIVP,ChristodoulouMzH,F1,F2,DamourSchmidt,%
MzHSeifertCIVP}; compare, in different settings,
\cite{CaciottaNicoloI,CaciottaNicoloII,HaywardNullSurfaceEquations}.
The most satisfactory treatment of the local evolution problem, for
a large class of quasi-linear wave equations and symmetric
hyperbolic systems, has been given by Rendall~\cite{RendallCIVP},
who proved existence of a solution in a neighborhood of the
intersection of the initial data hypersurfaces. A similar result for
a neighborhood of the tip of a light-cone has been established by
Dossa~\cite{DossaAHP}. The region of existence has been extended by
Cabet~\cite{Cabet1,Cabet2} for a class of non-linear wave-equations
satisfying certain structure conditions. In these last papers
existence of the solution in a whole neighborhood of the initial
data hypersurfaces, rather than of their intersection, is
established. We will refer to this kind of results  as ``the
neighborhood theorem''.   Similar results have been established by
Dossa and
collaborators~\cite{DossaBah,DossaTouadera,DossaWafo,DossaHoupa1,DossaHoupa2}
for various families of semilinear wave equations. Finally,
Luk~\cite{Luk} established the neighborhood theorem for the vacuum
Einstein equations in four space-time dimensions, in an argument
which makes use of the specific structure of the nonlinearities
occurring in those equations.

 The aim of this work is to show that no conditions on the non-linearity are necessary for existence near
an (optimal) maximal subset of the initial data hypersurfaces for
the large class of non-linear wave equations which can be written in
a doubly-null form.

We further show that our result applies to Einstein equations in
four space-time dimensions, as well as to a version, due to
Paetz~\cite{TimConformal}, of the conformal field equations of
Friedrich.

As a result we obtain that vacuum general relativistic characteristic initial data with suitable asymptotic behavior (as analyzed in detail in~\cite{ChPaetz3,PaetzScri}) lead to space-times with a piece of smooth Scri, without any smallness conditions on the data.%
\footnote{Once this work was completed we have been made aware of a
similar result in~\cite{LiZhuScri}.}
  Moreover, a  \emph{global-to-the-future} Scri is obtained if the data are sufficiently close to Minkowskian ones.

Higher-dimensional Einstein equations can be handled by a variation
of our techniques, this will be discussed elsewhere.

Our analysis is tailored to a setting where the initial data are
given on two transversely intersecting smooth characteristic
surfaces. The characteristic initial value problem with initial data
on a light cone issued from a point is readily reduced to the one
considered here, by first solving locally near the tip
(see~\cite{DossaAHP,ChConeExistence} and references therein), and
then using the results proved here to obtain a solution near the
maximal domain, within the light-cone, of existence of solutions of
the transport equations.

\chapter{The basic energy identity} \label{Sec:enid}
 Let $Y$ be a $(n-1)$-dimensional compact manifold
 without
 boundary. We are interested in quasi-linear first order
 symmetric hyperbolic systems of the
 form
\bel{fos} Lf = \newg\,,
\ee
on subsets of
\begin{eqnarray}\label{manif}
  \tmcM:=\{u\in[0,\infty),v\in
  [0,\infty),y\in Y\}\,.
\end{eqnarray}
In \eq{fos}, $f$ is assumed to be a section of a real vector bundle
over $\tmcM$, equipped with a scalar product, similarly for $\newg$.
We will use the same symbol $\nabla$, respectively $\langle\cdot
,\cdot \rangle$, to denote connections, respectively scalar
products, on all relevant vector bundles. Both the scalar product
and the connection coefficients are allowed to depend upon $f$, and
we assume that $\nabla$ is compatible with $\langle \cdot,\cdot
\rangle$. Similarly $\tmcM$ will be assumed to be equipped with a
measure $d\mu$, possibly dependent upon $f$. $L$ is a first order
operator of the form
$$L=A^\mu\nabla_\mu\,,$$
where the $A^\mu$'s are self-adjoint, and are smooth functions of
$f$ and of the space-time coordinates. The summation convention is
used throughout.

Let $q_r$, $r=1,\ldots,m$, denote a collection of smooth vector
fields on $Y$ such that for each $y\in Y$ the vectors $q_r(y)$ span
$T_yY$; clearly $m\ge \dim Y$. We will often write $\znabla_r$ for $\znabla_{q_r}$.

For $k\in \N$ let $\mcP^k$ denote the
collection  of differential operators of the form
\bel{Pform} \znabla_{q_{r_1}}\ldots  \znabla_{q_{r_\ell}}\,,\quad
    0\le \ell\le k\,.
\ee
Here $\znabla$ is a fixed, arbitrarily chosen,
smooth connection which is $f$, $u$, and $v$--independent. We number
the operators \eq{Pform} in an arbitrary way and call them $P_r$,
thus
$$\mcP^k=\{P_r, {r=1},\ldots,{N(k)}\}\,,$$ for a certain
$N(k)$, with $P_1=1$, the identity map.

Let $w_r$ be any smooth functions on $\tmcM$, we set
\bel{bid0}X{^\mu}(k):= \sum_{r=1}^{N(k)}w_r \langle P_rf,A^\mu
P_rf\rangle\,,\ee
 so that
\bel{bid1}
 \nabla_\mu (X{^\mu}(k))
 =  \sum_r\Big\{\underbrace{\langle P_rf,A^\mu
P_rf\rangle\partial_\mu w_r}_{I_r}
%\\ &&
 +w_r \big(\underbrace{\langle P_rf,(\nabla_\mu A^\mu) P_rf\rangle}_{II_r}+\underbrace{2\langle P_rf,L
P_rf\rangle}_{III_r}\big)\Big\}
 \,.
\ee
Let
$$
 \Omega_{a,b}=\underbrace{[0,a]}_{\ni u}\times \underbrace{[0,b]}_{\ni v}\times\underbrace{\; Y \;}_{\ni x^B }
 \,,
$$
and let $d\mu = du \, dv \, d\mu_Y$ be any measure, absolutely
continuous with respect to the coordinate Lebesgue measure, on
$\Omega_{ab}$, with smooth density function. From Stokes' theorem we have
$$\int_{\partial \Omega_{a,b}}X^\alpha(k) dS_\alpha = \int_{\Omega_{a,b}}\nabla_\mu(X^\mu(k)) d\mu\,,$$
so that \bean \int_{u=a}X^\alpha(k) dS_\alpha +
\int_{v=b}X^\alpha(k) dS_\alpha & = & \int_{u=0}X^\alpha(k)
dS_\alpha + \int_{v=0}X^\alpha(k) dS_\alpha\\
&& + \int_{\Omega_{a,b}}\nabla_\mu(X^\mu(k)) d\mu\,. \eeal{bid2}

From now on we specialise to $f$'s which are of the form
\bel{bid2.5}f= \left(%
 \begin{array}{c}
  \varphi \\
  \psi \\
 \end{array}%
\right)\,, \ee
with $A^v  $ and $A^u $ satisfying
\bel{bid3} A^u=\left(%
\begin{array}{cc}
  A^u_{\varphi\varphi} & 0 \\
  0 & 0 \\
\end{array}%
\right)\,,\
 A^v=\left(%
\begin{array}{cc}
  0 & 0 \\
  0 & A^v_{\psi\psi} \\
\end{array}%
\right)\,, \
 \mbox{
and  $A^v_{\psi\psi} > 0$, $A^u_{\varphi\varphi}> 0$.} \ee

  It is further assumed that the connections $\nabla $ and
$\znabla$ preserve the splitting \eq{bid2.5}. We will write
\bel{bid2.5a}\newg= \left(%
\begin{array}{c}
  \newg_\varphi \\
  \newg_\psi \\
\end{array}%
\right)\,.
\ee

From \eq{bid2.5}-\eq{bid3} we obtain, for fields supported in a
compact set $K$ \bean \int_{u=a}X^\alpha(k) dS_\alpha & \ge  & c(K)
\sum_r \int_{u=a} w_r \langle P_r\varphi, P_r\varphi\rangle \,dv\,
d\mu_Y\,, \nonumber \\ \label{bid4} & = :& c(K)\,
E_{k,\{w_r\}}[\varphi,a]\,,
\\ \int_{v=b}X^\alpha(k) dS_\alpha & \ge  & c(K) \sum_r \int_{v=b}
w_r \langle P_r\psi, P_r\psi\rangle \,du\, d\mu_Y\nonumber \\
& = :& c(K)\, \mcE_{k,\{w_r\}}[\psi,b] \,,
\eeal{bid6}
for some constant $c(K)$.
Equations \eq{bid1}-\eq{bid2} give thus
\bean E_{k,\{w_r\}}[\varphi,a] + \mcE_{k,\{w_r\}}[\psi,b] & \le &
C_1(K)\Big\{ E_{k,\{w_r\}}[\varphi,0] + \mcE_{k,\{w_r\}}[\psi,0]
\\ & & + \int_{\Omega_{a,b}}\sum_r (I_r + w_r (II_r +  III_r))\Big\}
 \,,
  \phantom{xxxx}
\eeal{bid5}
for some constant $C_1(K)$.

Let $
\lambda\ge 0$, we choose the weights to be independent of $r$:
\bel{bid6.5} w_r= e^{-\lambda(u+v)}\,,\ee and we will write
$E_{k,\lambda}$ for $E_{k,\{w_r\}}$ with this choice of weights,
similarly for $\mcE_{k,\lambda}$. From \eq{bid4} we find
%\bean E_{k,\alpha,\lambda}[\varphi,u]&=& \sum_{0\le \ell \le k}
%u^{-2\alpha-1+2\ell} \times \\ \nonumber &&\sum_{0\le i+j\le
%k-\ell}
% \int_{[0,v]\times Y} |(\nabla_{\partial_v})^i
% \nabla_{q_{a_1}}\ldots \nabla_{q_{{a_j}}} (\nabla_{\partial_u})^\ell\varphi(u,\cdot)|^2e^{-\lambda(u+v)} dv
% d\mu_Y
% \\ & =:& \sum_{0\le \ell \le k}
%u^{-2\alpha-1+2\ell} e^{-\lambda u} \|\nabla^\ell_u
%\varphi(u)\|^2_{k-\ell,\lambda} \,, \eeal{bid7old} where one
%recognises a sum of squares of weighted Sobolev norms $H^{k-\ell}$
%of the transverse derivatives $\nabla_u^\ell \varphi$, with an
%exponential factor $e^{-\lambda v}$ thrown in. \ptcold
\bean
E_{k,%\alpha
\lambda}[\varphi,a]&=& %\sum_{0\le \ell \le k}
%u^{-2\alpha-1+2\ell} \times \\ \nonumber &&
\sum_{0\le %i+
j\le
k%-\ell
}
 \int_{[0,b]\times Y} |%(\nabla_{\partial_v})^i
 \znabla_{q_{r_1}}\ldots \znabla_{q_{{r_j}}} %(\nabla_{\partial_u})^\ell
 \varphi(a,v,\cdot)|^2e^{-\lambda(a+v)} dv
 d\mu_Y
 \\ & =:& %\sum_{0\le \ell \le k}
%u^{-2\alpha-1+2\ell}
\int_{0}^b e^{-\lambda( a+v)} \|%\nabla^\ell_u
\varphi(a,v)\|^2_{H^{k}(Y)} dv\,, \eeal{bid7}
 where one recognises
the usual Sobolev norms $H^{k}(Y)$ on $Y$. One similarly has
\bean \mcE_{k,%\alpha,
\lambda}[\psi,b]&=&
\sum_{0\le %i+
j\le k}%-\ell
%} \\ \nonumber &&
 \int_{[0,a]\times Y} |%(\nabla_{\partial_u})^i
 \znabla_{q_{r_1}}\ldots \znabla_{q_{{r_j}}} %(\nabla_{\partial_v})^\ell
 \psi(u,b,\cdot)|^2 %u^{-2\alpha-1+2\ell}
 e^{-\lambda(u+b)} du\, d\mu_Y
\\
  & =:&
  \int_0^a e^{-\lambda (u+b)} \|
 \psi(u,b)\|^2_{H^k(Y)} du\,.
\eeal{bid8}

We recall some  general inequalities, which will be used repeatedly.
Recall  that $Y$ is a compact manifold without boundary (compare,
however, Remark~\ref{R10V2014.1}, p.~\pageref{R10V2014.1}). First, we
have the  Moser product inequality:
 \bel{Moser1}
 \|fg\|_{H^k(Y)} \le C_M(Y,k) \Big(\|f\|_{L^\infty(Y)} \| g\|_{H^k(Y)}+ \|f\|_{H^k(Y)}
 \| g\|_{L^\infty(Y)}\Big)\,.
 \ee
Next, we have the
 Moser commutation inequality, for $0\le r\le k$:
 \bel{Moser2}
 \|P_r(fg)-P_r(f)g\|_{L^2(Y)} \le C_M(Y,k)
 \Big(\|f\|_{L^{\infty}(Y)} \| g\|_{H^k(Y)}+ \|f\|_{H^{k-1}(Y)} \| g\|_{W^{1,\infty}(Y)}\Big)\,.
 \ee
We shall also need the Moser composition inequality:
 \bel{Moser3}
 \|F(f,\cdot)\|_{H^k(Y)} \le \hat C_M\Big(Y,k,F,\|f\|_{L^\infty(Y)}%,\|g|_{f=0}\|_{L^\infty(Y)}
 \Big) \Big(\|F(f=0,\cdot)\|_{H^k(Y)}+
 \|f\|_{H^k(Y)}\Big)\,.
 \ee

The constants $C_M$ and $\hat C_M$ also depend upon the
connection $\znabla$.

We return to the energy identity on a set $\mcU\times Y$, with $\mcU$ coordinatised by $u$ and $v$.
If $X(k)$ is given by \eq{bid0}, with $w_r=e^{-\lambda(u+v)}$, then,
writing $LP_rf$ as $P_r Lf + [L,P_r]f$, and assuming
\bel{strict} \langle \varphi , A^u_{\varphi\varphi}\varphi\rangle
\ge c |\varphi|^2\,,\qquad \langle \psi , A^v_{\psi\psi}\psi\rangle
\ge c |\psi|^2
 \,,
\ee
with $c>0$, one obtains for $k>\frac{n-1}{2}$
\bean
 \lefteqn{\int_{\mcU\times Y} \nabla_\alpha(X^\alpha(k))d\mu  \le}&&
\\
 &&
 \!\!\!\!\!\!
    \int_\mcU e^{-\lambda(u+v)}\Big\{\left(\|\nabla_\mu
 A^\mu\|_{L^\infty(Y)} % +2C_M(Y,k)\|B\|_{L^\infty(Y)}
 -c\lambda\right)\|f\|^2_{H^{k}(Y)}
 +C(Y,k)\|f\|_{H^{k}(Y)} %\Big[C_M(Y,k) \|B\|_{H^k(Y)} \|f\|_{L^\infty(Y)} +
 \|\newg\|_{H^k(Y)} %\Big]
 \nonumber \\&&%\label{volest1b}\\&&%\phantom{xxx}
 \qquad \qquad \qquad+ 2\int_{\mcU\times Y} \langle P_r f,[L,P_r]f\rangle e^{-\lambda(u+v)}
 \,d\mu
  \Big\}\,.
 \label{volest1}
\eea

 Some special cases are worth pointing out:

 \begin{enumerate}
 \item The case of ODE's in $u$ with a parameter $v$, or vice-versa,
 corresponds to $Y$ being a single point and $k=0$.
 \item The usual
 energy inequality for symmetric hyperbolic systems is obtained
 when $\mcU=I$ is an interval in $\R$.
 \end{enumerate}

 To control the
 commutators we will assume \eq{bid3}. We identify $(\varphi,0)$ with $\varphi$,
similarly for $(0,\psi)$ and $\psi$, and write \bean [A^\mu
\nabla_\mu,P_r]f &= &  [A^u \nabla_u,P_r]f+[A^v \nabla_v,P_r]f+[A^B
\nabla_B,P_r]f
\\&= & [A^u \nabla_u,P_r]\varphi+[A^v \nabla_v,P_r]\psi+[A^B
\nabla_B,P_r]f\,.\eeal{bid15a}
 Thus, it suffices to estimate
 $[A^u_{\varphi\varphi}\nabla_u,P_r]\varphi$, $[A^v_{\psi\psi}\nabla_v,P_r]\psi$,
 and $[A^B\nabla_B,P_r]f$. We define the relative connection coefficients $\Gamma_\mu$ by the formula
 \bel{relconcoef}
 \Gamma_\mu f:= \nabla_\mu f - \znabla_\mu f\,.
 \ee
By hypothesis the connections preserve the $(\varphi,\psi)$
decomposition, so that $\Gamma_\mu$ can be written as
 \bel{Gammdecom}
 \Gamma_\mu=\left(%
\begin{array}{cc}
  \Gamma_{\varphi\varphi,\mu} & 0 \\
  0 & \Gamma_{\psi\psi,\mu} \\
\end{array}%
\right)\,.
 \ee
This leads to the following form of $[A^B\nabla_B,P_r ]f$: \beaa
\yltwo{[A^B\nabla_B,P_r]f }& = &
\yltwo{[A^B\znabla_B,P_r]f+[A^B\Gamma_B,P_r]f }\,. \eeaa
 Using \eq{Moser2}-\eq{Moser3}, the first term is estimated as  \beaa %\lefteqn
 {C'_M
\Big(\|A \|_{W^{1,\infty}(Y)} \|f \|_{H^{k}(Y)}+ \|A \|_{H^{k}(Y)}
\| f\|_{W^{1,\infty}(Y)}\Big)}
%&& \\ &&\le
%C\Big(Y,k,\| f\|_{W^{1,\infty}(Y)},\| A^b\|_{W^{1,\infty}(Y)}\Big)
%%\times \nonumber \\ &&
%\Big( \|f \|_{H^{k}(Y)}+ \|A^b\|_{H^{k}(Y)}\Big)\,.\nonumber \\&&
\,,\eeaa %l{alp0est}
and the second as
\beaa %\lefteqn
 {C''_M
\Big(\|A^B\Gamma_B\|_{W^{1,\infty}(Y)} \|f \|_{H^{k-1}(Y)}+
\|A^B\Gamma_B\|_{H^{k}(Y)} \| f\|_{L^{\infty}(Y)}\Big)}
%&& \\ &&\le
%C\Big(Y,k,\| f\|_{W^{1,\infty}(Y)},\| A^b\|_{W^{1,\infty}(Y)}\Big)
%%\times \nonumber \\ &&
%\Big( \|f \|_{H^{k}(Y)}+ \|A^b\|_{H^{k}(Y)}\Big)\,.\nonumber \\&&
\,,\eeaa %l{alp0est}
leading, by \eq{Moser1}, to an overall estimation \bean
\yltwo{[A^B\nabla_B,P_r]f }& \le  & C\Big(Y,k,\|
f\|_{W^{1,\infty}(Y)},\| A \|_{W^{1,\infty}(Y)},\| \Gamma
\|_{W^{1,\infty}(Y)}\Big) \times \nonumber \\ && \Big( \|f
\|_{H^{k}(Y)}+ \|A\|_{H^{k}(Y)}+ \|\Gamma
\|_{H^{k}(Y)}\Big)\,.\eeal{alp0est}
Here we have written
\bel{10VI0.5}
 \|A\|_{H^{k}(Y)}= \sum_\mu\|A^\mu\|_{H^{k}(Y)}
 \,,
 \quad
 \|\Gamma\|_{H^{k}(Y)}= \sum_\mu\|\Gamma_\mu\|_{H^{k}(Y)}
 \,.
\ee

Writing $\nabla_\mu\varphi$ as
$\partial_\mu\varphi+\gamma_{\varphi\varphi,\mu}\varphi$ we have
\bean \lefteqn{ \underbrace{
    [A^u_{\varphi\varphi}(\partial_u+\gamma_{\varphi\varphi,u}),P_r
    ]\varphi}_{\alpha} = [A^u_{\varphi\varphi},P_r
    ]\underbrace{\partial_u\varphi}_{\nabla_u\varphi-\gamma_{\varphi\varphi,u}\varphi}+
    [A^u_{\varphi\varphi}\gamma_{\varphi\varphi,u},P_r ]\varphi
 }
 &&
\\\nonumber
&=& \underbrace{[A^u_{\varphi\varphi},P_r
]\Big\{(A^u_{\varphi\varphi})^{-1}\Big[\underbrace{-A^B_{\varphi\varphi}\nabla_B\varphi
-A^B_{\varphi\psi}\nabla_B\psi}_{-A^B_{\varphi\varphi}(\znabla_B+\Gamma_{\varphi\varphi,B})\varphi
-A^B_{\varphi\psi}(\znabla_B+\Gamma_{\psi\psi,B})\psi}+%\Big(Bf +b\Big)
\newg_\varphi \Big]}_{\alpha_1} -
\gamma_{\varphi\varphi,\mu}\varphi\Big\}
\\ &&
+\underbrace{[A^u_{\varphi\varphi}\gamma_{\varphi\varphi,u},P_r
]\varphi}_{\alpha_3}=:\alpha_1+\alpha_2+\alpha_3\,,\eeal{alp0est1}
with  $\alpha_2$ defined by the last equality.
%where addition of a subscript $\varphi$ or $\psi$ denotes the
%relevant component of a field.
%
Set:
$$
\tilde{A}^B_{\varphi\varphi,u}:=\left(A^u_{\varphi\varphi}\right)^{-1}A^B_{\varphi\varphi},\quad
\tilde{A}^B_{\varphi\psi,u}:=\left(A^u_{\varphi\varphi}\right)^{-1}A^B_{\varphi\psi},\quad
\tilde{\newg}_{\varphi}=
\left(A^u_{\varphi\varphi}\right)^{-1}\newg_\varphi
$$
and
$$
\tilde{A}^B_{\psi\varphi,v}:=\left(A^v_{\psi\psi}\right)^{-1}A^B_{\psi\varphi},\quad
\tilde{A}^B_{\psi\psi,v}:=\left(A^u_{\psi\psi}\right)^{-1}A^B_{\psi\psi},\quad
\tilde{\newg}_{\psi}=\left(A^v_{\psi\psi}\right)^{-1}\newg_\psi\,.
$$
By \eq{Moser2}-\eq{Moser3} we have the estimate \bean
\yltwo{\alpha_1} &\le& C_M \Big(\|A^u\|_{W^{1,\infty}(Y)}
\|\nabla_u\varphi \|_{H^{k-1}(Y)}+ \|A^u\|_{H^{k}(Y)} \|
\nabla_u\varphi\|_{L^{\infty}(Y)}\Big) \\ &\le& C\Big(Y,k,\|
f\|_{W^{1,\infty}(Y)},\| A^u \|_{W^{1,\infty}(Y)},\| \tilde{A}
\|_{L^{\infty}(Y)},\| \Gamma \|_{L^{\infty}(Y)},\|
\tilde{\newg}_{\varphi}\|_{L^{\infty}(Y)}\Big)\times \nonumber \\ &&
\Big( \|f \|_{H^{k}(Y)}+ \|A^u\|_{H^{k}(Y)}+ \|\tilde{A}
\|_{H^{k-1}(Y)}+ \|\Gamma \|_{H^{k-1}(Y)}+
\|\tilde{\newg}_\varphi\|_{H^{k-1}(Y)}\Big)\,.\nonumber
\\&&
 \eeal{alp1est}
Similarly,
\bean
 \yltwo{\alpha_2} &\le& C\Big(Y,k,\|
\varphi\|_{L^{\infty}(Y)},\| A^u \|_{W^{1,\infty}(Y)},\|
\gamma_{\varphi\varphi,u} \|_{L^{\infty}(Y)}\Big)\times
 \nonumber \\ && \Big( \|\varphi \|_{H^{k-1}(Y)}+ \|A^u\|_{H^{k}(Y)}+
 \|\gamma_{\varphi\varphi,u} \|_{H^{k-1}(Y)} \Big)\,,
 \label{alpha2}
\eea
and
\bean
 \yltwo{\alpha_3} &\le& C\Big(Y,k,\|
 \varphi\|_{L^{\infty}(Y)},\| A^u \|_{W^{1,\infty}(Y)},\|
 \gamma_{\varphi\varphi,u} \|_{W^{1,\infty}(Y)}\Big)\times
 \nonumber
\\
 && \Big( \|\varphi \|_{H^{k-1}(Y)}+ \|A^u\|_{H^{k}(Y)}+
 \|\gamma_{\varphi\varphi,u} \|_{H^{k}(Y)} \Big)\,.
 \label{alpha3}
\eea

By symmetry we have a similar contribution from
$[A^v_{\psi\psi}\nabla_v,P_r]\psi$. It follows that there exists a
constant
\bean \hat C_1 &= &  C\Big(Y,k,\| f\|_{W^{1,\infty}(Y)},\|
A\|_{W^{1,\infty}(Y)},\| \tilde{A}\|_{L^{\infty}(Y)},\|
\gamma\|_{W^{1,\infty}},
\\
 && \phantom{C\Big( }\| \Gamma\|_{W^{1,\infty}},\|
    \newg\|_{W^{1,\infty}},\|\tilde{\newg}\|_{L^{\infty}}\Big)
\eeal{5X13.11}
so that \eq{volest1} can be rewritten as
\bean
 \int_{\mcU\times Y} \nabla_\mu(X^\mu(k))d\mu
 &\leq&
 \int_\mcU e^{-\lambda(u+v)}\Big\{\left(\|\nabla_\mu
A^\mu\|_{L^\infty(Y)} % +2C_M(Y,k)\|B\|_{L^\infty(Y)}
-c\lambda\right)\|f\|^2_{H^{k}(Y)}
 \nonumber
\\
 && +
%\nonumber \\&&\label{volest1a}\\&&%\phantom{xxx}
%+2\|f\|_{H^{k}(Y)} %\Big[C_M(Y,k) \|B\|_{H^k(Y)} \|f\|_{L^\infty(Y)} +
%\|\newg\|_{H^k(Y)} %\Big]
%\nonumber \\&&
%\label{volest2b}\\&&%\phantom{xxx}
 \hat C_1 \|f\|_{H^{k}(Y)}\times\Big( \|f \|_{H^{k}(Y)}
  + \|A\|_{H^{k}(Y)}+
    \|\tilde{A}\|_{H^{k-1}(Y)}
\nonumber
\\ &&+
    \|\Gamma\|_{H^{k}(Y)}+\|\gamma\|_{H^{k}(Y)}+ \|\newg\|_{H^{k}(Y)}
    + \|\tilde{\newg}\|_{H^{k-1}(Y)}\Big)\Big\}
    \, du \, dv\,.
    \nonumber
\\
 %\label{volest2c}
 \label{volest2}
\eea

\chapter{The iterative scheme} \label{sec:iterati-scheme}

\section{Outline of the iteration argument}\label{sub:Initial-data}

For the purpose of the arguments in this section, we let
$$
 \mcN^-:=\{u=0\,, v\in [0,b_0]\}\times Y
  \,,
   \quad
   \mcN^+:=\{u\in [0,a_0],v=0\}\times Y
   \,;
$$
we will see later how to handle general initial characteristic
hypersurfaces for systems arising from wave equations. The initial
data $\overline  f\equiv f|_{\mcN}$ will be given on
$$
 \mcN:=\mcN^-\cup \mcN^+
 \,,
$$
and will belong to a suitable Sobolev class.  More precisely, we are
free to prescribe $\overline \varphi (v) \equiv \varphi (0,v)$ on
$\mcN^-$ and $\overline \psi (u) \equiv \psi (u,0)$ on $\mcN^+$, and
then the fields $  \psi (0,v)$ on $\mcN^-$ and $ \varphi (u,0)$ on
$\mcN^+$ can be calculated by solving transport equations.  In this
section we assume that these equations have global solutions on
$\mcN^\pm$, this hypothesis will be relaxed later.

Throughout we use the convention that overlining a field denotes
restriction to $\mcN$ (consistently with the last paragraph).

Our hypotheses will be symmetric with respect to the variables $u$
and $v$, and therefore the result will also be symmetric. We will
construct solutions on a neighborhood of $\mcN^-$ in
$$  \Omega_{a_0,b_0}:=\{u\in [0,a_0]\,, v\in [0,b_0]\}\times Y
   \,,
$$
and a neighborhood of $\mcN^+$ can then be obtained by
applying the result to the system in which $u$ is interchanged with
$v$.

The method is to use a sequence $\overline  f_i$ of smooth initial
data approaching   $\overline  f$, and to solve a sequence of
linear problems: We let $f_0$ be any smooth extension of $\overline
f_0$ to $\Omega_{a_0,b_0}$. Then, given $f_i$, the field $f_{i+1}$
is defined as the solution of the linear system \bel{itpro} L_i
f_{i+1}=\newg_{i}\,, \ee where \bel{itpro1}
L_i=A^\mu(f_i,\cdot)\nabla(i)_\mu\,,\quad
\newg_{i}=\newg(f_{i},\cdot)\,, \ee
and where we have used the symbol $\nabla(i)$ to denote $\nabla$, as
determined by $f_i$. (The reader may wonder why we do not replace
$\nabla$ by an $f$-independent connection, putting all the
dependence of $\nabla$ upon $f$ into
 the right-hand side of the equation. However,  in some situations the new connection might not be compatible with the scalar product, which has been assumed in our calculations.)
 For smooth initial data and
$f_i$, \eq{itpro} always has a global smooth solution on
$\Omega_{a_0,b_0}$ by~\cite{RendallCIVP}.

By continuity, the  $f_i$'s will satisfy a certain set of
inequalities, to be introduced shortly, on a subset
$$  \Omega_i:=\{u\in [0,a_i]\,, v\in [0,b_0]\}\times Y
   \,.
$$
We will show that there exists $a_*>0$ such that $a_i\ge a_*$, so
that there will be a common domain
$$
 \Omega_*:=\{u\in [0,a_*]\,, v\in [0,b_0]\}\times Y
%   \,.
$$
on which the desired inequalities will be satisfied by all the
$f_i$'s. This will allow us to show convergence to a solution of the
original problem defined on $\Omega_*$.

We note that our system implies a system of non-linear constraint
equations on $\overline  f$, sometimes called \emph{transport
equations}. The solutions of these constraints might blow up in
finite time, see e.g.~\cite{Cabet1} for an example arising from a
semilinear wave equation. It is part of our hypotheses that the
constraints are satisfied throughout $\mcN$; in some situations this
might require choosing $a_0$ and $b_0$ small enough so that a smooth
solution of the constraint equations exists.
%
%For systems satisfying our hypotheses both in the original
%variables, and with the variables $u$ and $v$ interchanged, one
%obtains a solution in a full neighborhood of $\mcN$.

\section{Bounds for the iterative
scheme}\label{sub:Bounds-iterati-scheme}

In order to apply the energy identity of Section~\ref{Sec:enid} we
need to estimate the volume integrals appearing in  \eq{bid2}. We
could appeal to \eq{volest2}, but it is instructive to analyse
\eq{bid5} directly. All terms arising from $I_r$ in \eq{bid1} give a
negative contribution, bounded above by
\bel{bid11} -\lambda c(K)%\sum_{0\le \ell \le k}
\int_0^a \int_0^b e^{-\lambda (u+v)}\|f_{i+1}(u,v)\|^2_{H^k(Y)} du
\,
dv\,. \ee %\ptco
%\bel{bid11old} -\lambda c(K)\sum_{0\le \ell \le k}\int_0^a e^{-\lambda
%u}\|f(u)\|^2_{k-\ell,\alpha-\ell,\lambda} du\,. \ee
The terms arising from $II_r$ give a contribution which, using
obvious notation, is estimated by \bel{bid13}  \|(\nabla_\mu
A^\mu)_i\|_{L^\infty} \int_0^a \int_0^b e^{-\lambda
(u+v)}\|f_{i+1}(u,v)\|^2_{H^k(Y)} du \, dv\,. \ee
% \ptco
%\bel{bid13old} C_2(k,n) \|\nabla_\mu A^\mu\|_{L^\infty} \sum_{0\le \ell \le k}\int_0^a e^{-\lambda
%u}\|f(u)\|^2_{k-\ell,\alpha-\ell,\lambda} du\,. \ee
%
The estimation of the terms arising from $III_r$ requires care, as
we need to control $\lambda$-dependence of the constants. One can
proceed as follows:
\beaa
 \lefteqn{2 \sum_r \int_0^a\int_0^b \langle P_rf_{i+1},L
 P_rf_{i+1}\rangle e^{-\lambda(u+v)}du\,dv\,d\mu_Y}
\\ && = 2
\sum_r \int_0^a\int_0^b \langle P_rf_{i+1},P_r\newg_{i} +[L
_{i},P_r]f_{i+1}\rangle e^{-\lambda(u+v)}du\,dv \,d\mu_Y
\\
 &&
    \le 2 \int_0^a\int_0^b
    e^{-\lambda(u+v)}\|f_{i+1}(u,v)\|_{H^k(Y)}\Big(\underbrace{\|\newg_{i}(u,v)\|_{H^k(Y)}}_{III_{1}}
\\
&& \phantom{2 \int_0^a\int_0^b xx}
    +\underbrace{\sum_r\|[L_{i}
    ,P_r]f_{i+1}(u,v)\|_{L^2(Y)}}_{III_{2 }}\Big) du\,dv\,d\mu_Y
    \,.
\eeaa
The term $III_{1}$ can be estimated by the usual Moser inequality on
$Y$,
$$
\|\newg_{i}(u,v)\|_{H^k(Y)} \le
C\Big(k,Y,\|f_i(u,v)\|_{L^\infty(Y)}\Big)\Big(\|\mathring{\newg}
(u,v)\|_{H^k(Y)}+\|f_i(u,v)\|_{H^k(Y)}\Big)\,,$$
 where $\mathring{\newg} =\newg(f=0)$.
 Let $0<\epsilon\le 1$ be a constant which will be
determined later, the inequality $ab\le  a^2/(4\epsilon)+ \epsilon
b^2$ leads then  to a contribution of $III_{1}$ in \eq{bid5} which
is estimated by
\bean\lefteqn{
 C_3\Big(k,Y,\sup_{u,v}\|f_i(u,v)\|_{L^\infty(Y)}\Big)  \int_0^a
 \int_0^b e^{-\lambda (u+v)}\Big( \|\mathring{\newg}  (u,v)\|^2_{H^k(Y)}}&&
\\
 &&\phantom{xxxxx}+ \epsilon \|f_i(u,v)\|^2_{H^k(Y)}+
 c_1(\epsilon)\|f_{i+1}(u,v)\|^2_{H^k(Y)}\Big) du \, dv\,,
\eeal{bid15}
 with $c_1(\epsilon)\to \infty$ as $\epsilon \to 0$.
The analysis of $III_{2}$ proceeds as in \eq{bid15a}. Since the
$P_r$'s are $u$--independent we have
$$
 [A^u \partial_u,P_r]\varphi=[A^u ,P_r]\partial_u\varphi
  \,,
$$
and, calculating as in \eq{alp0est1}, we can use the equation
satisfied by $\varphi_{i+1}$ to replace $\partial_u \varphi_{i+1}$
by a first order differential operator in $f_{i+1}$ tangential to
$Y$ (with coefficients that perhaps depend upon $f_i$), similarly
for $A^v\partial_v \psi_{i+1}$. The Moser commutation inequality
\eq{Moser2} on $Y$  can then be used to obtain the following
estimation for the corresponding contribution to \eq{bid5}:
\bean
C_3\Big(k,Y,\sup_{u,v}\|f_{i}(u,v)\|_{W^{1,\infty}(Y)},\sup_{u,v}\|f_{i+1}(u,v)\|_{W^{1,\infty}(Y)}\Big)
\times\phantom{xxxxx}
\\  \int_0^a \int_0^b e^{-\lambda (u+v)}
\Big(\epsilon\|f_{i}(u,v)\|^2_{H^k(Y)}+c_2(\epsilon)\|f_{i+1}(u,v)\|^2_{H^k(Y)}
+\|\zA \|^2_{H^{k}(Y)} \nonumber
\\+ \|\mathring{\tilde{A}} \|^2_{H^{k-1}(Y)}+
\|\zGamma \|^2_{H^{k}(Y)}+\|\mathring{\gamma} \|^2_{H^{k}(Y)}+
\|\mathring{\newg} \|^2_{H^{k}(Y)}+
\|\mathring{\tilde{\newg}}\|^2_{H^{k-1}(Y)} \Big)du \, dv\,,
\eeal{bid17}
where
$$
\zA^\mu=A^\mu|_{f=0}\,,\quad \zGamma_\mu=\Gamma_\mu|_{f=0}\,, \quad
\mbox{etc.,}
$$
with norms defined as in \eqref{10VI0.5}.

Define
\bean
 C_0
  &=&
  1
  +
 \sup_{i\in \N,(u,v)\in ([0,a_0]\times
 \{0\})\cup(\{0\}\times[0,b_0])} \|\overline
 {f}_i(u,v)\|_{W^{1,\infty}(Y)}%
%\\
% &
% &
% +\|( \overline
% {\partial_u\psi}_i,\overline{\partial_v\varphi}_i ,
%  \overline{\partial_u^2\psi}_i,\overline{\partial_v \partial_u\varphi}_i)(u,v)\|_{W^{1,\infty}(Y)}  \bigg)
%%
 \,.
  \phantom{xxxxx}
\eeal{bid9def}
We assume that $C_0$ is finite.

%\bea
% C_0
% &=& %\max\Big(\|(\nabla_\mu A^\mu)_i\|_{L^\infty(\mcN^+\cup \mcN^-)} ,
% \sup_{i\in \N,(u,v)\in ([0,a_0]\times
% \{0\})\cup(\{0\}\times[0,b_0])}\|\overline
% {f}_i(u,v)\|_{W^{1,\infty}(Y)}
% + 1
% \,.
%\eeal{bid9def}
%
%\bea C_0 &=& \max\Big(\|(\nabla_\mu
%A^\mu)_i\|_{L^\infty(\mcN^+\cup \mcN^-)} ,
%\sup_{(u,v)\in\partial\Omega_{a_0,b_0}}\|\mathring
%f(u,v)\|_{W^{1,\infty}(Y)}\Big)\,,\phantom{xxxx}\eeal{bid9def} %and let

 Let $\mcK$ be a compact neighborhood of the image
of the initial data map $\overline {f}$. We will assume that the
sequence $\overline {f}_i$ converges to $\overline {f}$ in $L^\infty
(\mcN)$, similarly for first and second-order derivatives, in
particular we can assume that the image of $\overline {f}_i$  lies
in $\mcK$.

Let
\bel{Cdiv} \Cdiv:=\sup |\nabla_\mu A^\mu| +1\,, \ee
where the supremum is taken over all points in $\mcN^+\cup\mcN^-$
and over all $(\varphi,\psi,\nabla \varphi, \nabla \psi)$ satisfying
\bean
 &
 (\varphi,\psi)\in \mcK
 \,,
  \quad| \znabla_B f(u,v)|\le 2 C_0
  \,,
%  &
%\\
% &
 \quad |\partial_u \psi|\le 2 \sup_i\ylin{{\overline{\partial  {\psi}}_i
\over\partial
  u}} +1\,,
  &
   \nonumber
  \\
  &
    |\partial_v
\varphi|\le 2 \sup_i\ylin{\overline{{\partial {\varphi}}_i
\over\partial
  v}} +1
  \,.
\eeal{6IX13.1}
(The suprema over $i$ will be finite in view of our hypotheses on the
sequence $ \overline {f}_i$.) We note that
\bel{9X13.1}
 \nabla_\mu A^\mu = \partial_\psi A^\mu \partial_\mu \psi
  + \partial_\varphi A^\mu \partial_\mu \varphi  + \mbox{terms independent of derivatives of $f$}
  \,,
\ee
so that to control $\Cdiv$ one needs to control those derivatives of
$f$ which appear in $\partial_\psi A^\mu \partial_\mu \psi
  + \partial_\varphi A^\mu \partial_\mu \varphi $.
Now, in the right-hand side of \eq{9X13.1} the values $\partial_u
\varphi$ and $\partial_v\psi$ can be algebraically
  determined in terms of  other fields involved using the field
  equations: Indeed, using \eq{fos} we can view
  $\partial_u\varphi$ as a function, say $F$, of $f$ and
  $\znabla_B f$. Then, when calculating $\Cdiv$, we consider all
  values of $F$ with $f\in \mcK$ and $|\znabla_B f|\le 2C_0$;
  similarly for $\partial_v\psi$.

\begin{Remark}{\rm
  \label{R9X13.1}
It should be clear from \eq{9X13.1} that the condition on
$\partial_v \varphi$ in \eq{6IX13.1} is irrelevant if $A^v$ does not
depend upon $\varphi$. Similarly the condition  on $\partial_u \psi$
in \eq{6IX13.1} is irrelevant if $A^u$ does not depend upon $\psi$.
}\qed\end{Remark}

Let $a_i$ be the largest number in $(0,a_0]$ such that
\beqar
 & \|(\nabla_\mu
 A^\mu)_i\|_{L^\infty(\Omega_{a_i,b_0})} \le \Cdiv \,,
 \label{bid9f}&
\\
 & \sup_{(u,v)\in
 [0,a_i]\times[0,b_0]}\|f_i(u,v)\|_{W^{1,\infty}(Y)},
 \label{bid9f-1}
 \le 4 C_0
  \,.
 &
 \arrlabel{bid9m}
\eeqar

For any $\epsilon>0$ we can  choose $\lambda$ large enough,
independent of $i$, so that the sum of \eq{bid11}, \eq{bid13}, and
of the $f_{i+1}$ contribution to  \eq{bid15} and \eq{bid17}, is
negative on
$$
 \Omega_{\hat a_i, b_0}\,, \ \mbox{where}\ \hat a_i=\min(a_1,\ldots,a_{i+1})
 \,.
$$
If we let $M_k(u,v)$ be any function satisfying
\bean
    M_k(u,v)
    &\ge &
    \|\mathring{\newg}  (u,v)\|^2_{H^k(Y)} +\|\mathring{\tilde{\newg}}  (u,v)\|^2_{H^{k-1}(Y)} + \|\zA \|^2_{H^{k}(Y)}\nonumber
\\
    &&
     + \|\mathring{\tilde{A}} (u,v)\|^2_{H^{k-1}(Y)} +\|\mathring{\gamma}(u,v) \|^2_{H^{k}(Y)} +
        \|\zGamma(u,v) \|^2_{H^{k}(Y)}
        \,,
        \phantom{xxxxx}
\eeal{5X13.21}
we conclude that:

\begin{Lemma}
\label{Lenergy} Let   $0\le b\le b_0<\infty$, suppose that $Y$ is
compact and that \eq{bid9m} holds. Then for every $0<\epsilon\le 1$
 there exist constants
$\lambda_0(k,C_0,\Cdiv,Y,\epsilon)$ and $C_4(a_0,b_0,Y,k,C_0,\Cdiv)$
such that for all $\lambda \ge \lambda_0$ and $0\le a\le \hat a_i\le
a_0$
 we have
\bean E_{k,\lambda}[\varphi_{i+1},a] +
\mcE_{k,\lambda}[\psi_{i+1},b] \le  C_4\Big\{
E_{k,\lambda}[\overline {\varphi}_{i+1}] +
\mcE_{k,\lambda}[\overline {\psi}_{i+1}]\phantom{xxxxxx}
\\+ \int_0^a \int_0^b e^{-\lambda (u+v)}\Big( M_k(u,v)
+ \epsilon\|f_{i}(u,v)\|^2_{H^k(Y)} \Big)du \, dv\Big\}\,.
%\nonumber \\&&
\eeal{bid5a}%\ptco
\myqed
\end{Lemma}

We need, next, to get rid of the $i$-dependent terms in the
integrals at the  right-hand-side of \eq{bid5a}, this can be done as
follows: Set
\bea
% \lefteqn{
\hat C(a,b):=
  C_4\Big\{ \sup_{i\in\N} (E_{k,\lambda}[\overline {\varphi}_{i+1}]
+ \mcE_{k,\lambda}[\overline {\psi}_{i+1}]) + \int_0^a \int_0^b
e^{-\lambda (u+v)} M_k(u,v) du \, dv
 \Big\}
 \!\,;
  \phantom{xxxxx}
%}
&& \eeal{hatCdef}
note that this depends only upon the initial data and the structure
of the equations. Suppose that
\bel{fiint}
 \int_0^{a} \int_0^{b} e^{-\lambda
 (u+v)}\|f_{i}(u,v)\|^2_{H^k(Y)} du \, dv \le 2\hat C(a_0+b_0)
 \,.
\ee
We then have, using \eq{bid5a}, \beaa \lefteqn{ \int_0^a \int_0^b
e^{-\lambda (u+v)}\|\varphi_{i+1}(u,v)\|^2_{H^k(Y)} du \, dv   =
\int_0^a E_{k,\lambda}[\varphi_{i+1},u]du }&&
\\&&\le  \int_0^a (\hat C +2 \epsilon C_4 \hat C(a_0+b_0) ) du \le(\hat C + 2 \epsilon C_4 \hat C(a_0+b_0)  )a_0 \le 2
 \hat C a_0
 \,,
\eeaa
if $\epsilon$ is chosen small enough. Similarly, \beaa \lefteqn{
\int_0^a \int_0^be^{-\lambda (u+v)}\|\psi_{i+1}(u,v)\|^2_{H^k(Y)} du
\, dv }&&
\\&&  = \int_0^b \mcE_{k,\lambda}[\psi_{i+1},v]dv \le (\hat C
 +2 \epsilon C_4 \hat C(a_0+b_0) )b_0 \le 2 \hat C b_0\,.\eeaa Adding one
obtains \eq{fiint} with $i$ replaced by $i+1$. Decreasing $\epsilon$
if necessary we obtain:
\begin{Lemma}
\label{Lnext} Let $\hat C$ be defined by \eq{hatCdef}. Under the
hypotheses of Lemma~\ref{Lenergy}, one can choose
$\epsilon_0(a_0,b_0,Y,k,C_0,\Cdiv)$ so that \eq{fiint} is preserved
under the iteration for all
 $0\le b\le
b_0$, provided that  $0\le a\leq \hat a_i\le a_0$, with the
right-hand-side of \eq{bid5a} being less than $2\hat C(a_0,b_0)$
 for all $\lambda\ge \lambda_0(k,C_0,\Cdiv,Y,\epsilon_0)$.
\myqed
\end{Lemma}

To continue, let us write $$ A^\mu =\left(%
\begin{array}{cc}
  A^\mu_{\varphi \varphi} & A^\mu_{\varphi \psi} \\
  A^\mu_{\psi \varphi} & A^\mu_{\psi \psi} \\
\end{array}%
\right)%\quad g= \left(%
%\begin{array}{c}
%  \newg_\varphi \\
%  \newg_\psi \\
%\end{array}%
%\right)
\,. $$
Since, by hypothesis, the only non-vanishing component of $A^u$ is
$A^u_{\varphi\varphi}$, on any level set of $u$ the field
$\psi_{i+1}$ is a solution of the symmetric hyperbolic system
\bel{bid21} (A^\mu_{\psi\psi}\nabla_\mu)_i\psi_{i+1}\equiv
A^\mu_{\psi\psi}(f_{i},\cdot)\nabla_\mu(i)\psi_{i+1}= (\hat
\newg_\psi)_{i}\,,\ee
where $$(\hat \newg_\psi)_i \equiv
(\newg_\psi)_i-(A^\mu_{\psi\varphi}\nabla_\mu)_i\varphi_{i+1}:=
\newg_\psi(f_i,\cdot) -
A^\mu_{\psi\varphi}(f_{i},\cdot)\nabla_\mu(i)\varphi_{i+1}\,. $$

Set
 \bel{hCdiv}\hCdiv= \sup |\nabla _\mu A^\mu_{\psi\psi}|
 \le \Cdiv
 \,,
\ee
where the sup is taken as in \eq{Cdiv}. A calculation similar to the
one leading to the proof of Lemma~\ref{Lnext} shows that for any
$0<\delta\le 1$ there exists
$\lambda_1(k,C_0,\hCdiv,Y,\delta)<\infty$ such that for
$\lambda\ge\lambda_1$ we obtain  (recall that
$\psi_{i+1}(u,0)=\overline {\psi}_{i+1}(u)$)
% \ptcr{$\delta$ looks at a funny place, but in fact this is ok; but apparently this is not used later; also $k-1$
% here instead of $k$, but this is actually important to have $k-1$ because there is a loss of derivatives in the estimate}
%
\bean \lefteqn{ \|\psi_{i+1}(u,v)\|^2_{H^{k-1}(Y)}\le
C_5(Y,k,C_0,\hCdiv)e^{\lambda v}\Big\{ \ykmns{\overline
{\psi}_{i+1}(u)}+} &&\\ &&
 \delta\int_0^v e^{-\lambda s} \Big(\|f_{i}(u,s)\|^2_{H^{k-1}(Y)}+M_k(u,s)
 %\|\zA \|^2_{H^{k-1}(Y)}+\|\mathring{\tilde{A}} \|^2_{H^{k-1}(Y)}+
%\nonumber
%\\
% && \phantom{xxxxx} +
%\|\mathring{\gamma} \|^2_{H^{k-1}(Y)}+ \|\zGamma \|^2_{H^{k-1}(Y)}
+ \underbrace{\ykmns{(\hat
 \newg_\psi)_i(u,s)}}_{*}\Big)ds\Big\}
 \,.
 \phantom{xxxxx}
\eeal{bid23}
 The contribution of $*$ can be estimated as follows:
\beaa
\int_0^v * \,dv  & = & \int_0^v e^{-\lambda s}\ykmns{( (\newg_\psi)_i- (A^B_{\psi\varphi}\nabla_B)_i
\varphi_{i+1} )(u,s)}ds
\\
&\le & 2 \int_0^v e^{-\lambda s}\left(\ykmns{
(\newg_\psi)_i(u,s)}+\ykmns{(A^B_{\psi\varphi}\nabla_B)_i
\varphi_{i+1} )(u,s)}\right)ds \nonumber
 \\ &\le &
C_6(Y,k,C_0)\int_0^v e^{-\lambda s} \Big(\ykmns{\mathring
\newg_\psi(u,s)}+\|\zA\|^2_{H^{k-1}(Y)}+ \nonumber
\\
&& \phantom{xxxxx} + \|\zGamma\|^2_{H^{k-1}(Y)}
+\|f_{i}(u,s)\|^2_{H^{k-1}(Y)}+\ykns{\varphi_{i+1}(u,s)}\Big)ds
\nonumber
\\ &=& C_6(Y,k,C_0)\Big\{\int_0^v e^{-\lambda s}\Big(
\ykmns{\mathring \newg_\psi(u,s)}+\|\zA\|^2_{H^{k-1}(Y)}+ \nonumber
\\
&& \phantom{xxxxx} +
\|\zGamma\|^2_{H^{k-1}(Y)}+\|f_{i}(u,s)\|^2_{H^{k-1}(Y)}\Big)ds+
e^{\lambda u} E_{k,\lambda}[\varphi_{i+1},u]\Big\}\,,\eeaa %l{bel:bid24}
with $\mathring \newg_\psi(\cdot)=\newg_\psi(f=0,\cdot)$.
 It follows that,
for $0\le u\le a\le \hat{a}_i\le a_0$, \bea \label{dzi0} \nonumber
\lefteqn{ e^{-\lambda v}\|\psi_{i+1}(u,v)\|^2_{H^{k-1}(Y)}\le
C_7(Y,k,C_0,\hCdiv)\Big\{ \ykmns{\overline {\psi}_{i+1}(u)}}&&\\
\nonumber
 && \phantom{xxxxxxxxx}+
\int_0^v e^{-\lambda s} \Big(M_k(u,s)
%\ykmns{\mathring
%\newg_\psi(u,s)}+\|\zA\|^2_{H^{k-1}(Y)}+\|\mathring{\tilde{A}}\|^2_{H^{k-1}(Y)}+\|\mathring{\gamma}\|^2_{H^{k-1}(Y)}\nonumber
%\\
%%
%&&  +
%\|\zGamma\|^2_{H^{k-1}(Y)}
+\delta
\|f_{i}(u,s)\|^2_{H^{k-1}(Y)}\Big)ds + %\sup_{u\in[0,a]}
e^{\lambda u}E_{k,\lambda}[\varphi_{i+1},u]\Big\}\,.\qquad\qquad
\eea
 By
Lemma~\ref{Lnext} the $\varphi_i$ part of the $f_i$ contribution can
be estimated by $e^{\lambda u} E_{k-1,\lambda}[\varphi_i,u]\le
2e^{\lambda u}\hat C(u,v)\le 2e^{\lambda u}\hat C(u,b)$, so that
\bean\lefteqn{ e^{-\lambda v}\|\psi_{i+1}(u,v)\|^2_{H^{k-1}(Y)}\le
C_7(Y,k,C_0,\hCdiv)\Big\{ \ykmns{\overline {\psi}_{i+1}(u)}+2e^{\lambda u}\hat C(u,b)}&&\\
&& + \int_0^v e^{-\lambda s} \Big(M_k(u,s)
%\ykmns{\mathring
%\newg_\psi(u,s)}+\|\zA\|^2_{H^{k-1}(Y)}+\|\mathring{\tilde{A}}\|^2_{H^{k-1}(Y)}+\|\mathring{\gamma}\|^2_{H^{k-1}(Y)} \nonumber
%\\
%%
%&& \phantom{xxxxx} +
%\|\zGamma\|^2_{H^{k-1}(Y)}
+\delta \|\psi_{i}(u,s)\|^2_{H^{k-1}(Y)}\Big)ds
+%\sup_{u\in[0,a]}
e^{\lambda u}E_{k,\lambda}[\varphi_{i+1},u]\Big\}\,. \nonumber \\
&& \eeal{dzi1}
%Multiplying by $e^{-\lambda u}$ and i
%
Integrating in $v$ one obtains, for $0\le b\le b_0$,
\bean\lefteqn{ \int_0^b e^{-\lambda
v}\|\psi_{i+1}(u,v)\|^2_{H^{k-1}(Y)}dv \le \tilde C_\psi(u,b) +  }
 &&
\\
 && + C_7(Y,k,C_0,\hCdiv)\delta\int_0^b\int_0^v
  e^{-\lambda s} \|\psi_{i}(u,s)\|^2_{H^{k-1}(Y)}ds\,dv\,,
 \phantom{xxx}
\eeal{dzi1.5}
where
\bean%\lefteqn{
\tCpsi (a,b) &=& C_7(Y,k,C_0,\hCdiv)\sup_{i\in\N, u\in [0,a]}\Big\{
b\ykmns{\overline {\psi}_{i+1}(u)}
  \\\nonumber
  && + \int_0^b\int_0^v e^{-\lambda s}
    M_k(u,s)\, ds \, dv
+2be^{\lambda u}\hat C(u,b)+ b %\sup_{u\in[0,a]}
e^{\lambda u} E_{k,\lambda}[\varphi_{i+1},u]\Big\}\,.
\eeal{Ctildedef}

Suppose that there exists a  constant $\hCdiv$  such that
\bea
 \label{intpsic1}
  &\sup_i|(\nabla _\mu A^\mu_{\psi\psi})_i|\le \hCdiv\,,
   &
\eea
(note that we necessarily have $\hCdiv\le \Cdiv$), and that \bea
 \displaystyle & \forall \ 0\le v\le b\quad \int_0^v e^{-\lambda s}
\|\psi_{i}(u,s)\|^2_{H^{k-1}(Y)}ds\le 2 \tCpsi (u,b)\,.&
 \label{intpsic}
\eea
\Eq{dzi1.5} shows that
\bean
 \lefteqn{
 \int_0^b e^{-\lambda v}\|\psi_{i+1}(u,v)\|^2_{H^{k-1}(Y)}dv
 }
 &&
\\
 \nonumber
  &\le&
    \tCpsi (u,b)
    + C_7\delta \int_0^b\int_0^v e^{-\lambda s}
    \|\psi_{i}(u,s)\|^2_{H^{k-1}(Y)}ds\,dv
\\
 &\le&
    \tCpsi (u,b) + 2\delta b_0 C_7\tCpsi (u,b) \le 2 \tCpsi (u,b)
 \,,
\eeal{dzi1.5a}
if $\delta=\delta (b_0, \tCpsi(a_0,b_0) ,C_7)$ is chosen small
enough. It follows that \eq{intpsic} is preserved under the
iteration scheme if \eq{bid9m} and \eq{intpsic1} hold. With this
choice of $\delta$, \eq{dzi1} gives now
\bean\lefteqn{ e^{-\lambda v}\|\psi_{i+1}(u,v)\|^2_{H^{k-1}(Y)}\le
C_7(Y,k,C_0,\Cdiv)\Big\{ \ykmns{\overline {\psi}_{i+1}(u)}}&&\\
&& +2e^{\lambda u}\hat C(u,b)+ \int_0^v e^{-\lambda s} M_k(u,s)
%\Big(\ykmns{\mathring
%\newg_\psi(u,s)}+\|\zA\|^2_{H^{k-1}(Y)}+\|\mathring{\tilde{A}}\|^2_{H^{k-1}(Y)}+\|\mathring{\gamma}\|^2_{H^{k-1}(Y)}
%  +
%\|\zGamma\|^2_{H^{k-1}(Y)}\Big)
ds
%\nonumber\\
%&&\qquad\qquad \qquad
+2\delta \tCpsi (u,b)
+%\sup_{u\in[0,a]}
e^{\lambda
u}E_{k,\lambda}[\varphi_{i+1},u]\Big\}\,.\phantom{xxxxxxx}
  \eeal{dzi11.6}
By an essentially identical argument using the symmetry of the
equations under the interchange of $u$ and $v$, but still working
with $0\le u \le \hat a_i$, $0\le v\le b_0$, if we let $\hhCdiv$ be
a constant such that
\bea
 \label{intphic1}
 &
  \sup_i|(\nabla _\mu A^\mu_{\varphi\varphi})_i|\le
  \hhCdiv
%  \,,
  &
\eea
(note that $\hhCdiv\le \Cdiv$), then the condition
\bea
 \displaystyle & \forall \ 0\le u\le
 \hat a_i \quad \int_0^u e^{-\lambda s}
 \|\varphi_{i}(s,v)\|^2_{H^{k-1}(Y)}ds\le 2 \tCvarphi
 (a,v)\,,&
 \label{phipresint}
\eea
where
\bean%\lefteqn{
\tCvarphi (a,b) &=& C_7(Y,k,C_0,\hhCdiv)\sup_{i\in\N, v\in
[0,b]}\Big\{ a\ykmns{\overline {\varphi}_{i+1}(v)}
  \\\nonumber && + \int_0^a\int_0^u e^{-\lambda s}
M_k(s,v)ds\,du
%\ykmns{\mathring \newg_\varphi(s,v)}+\|\zA \|^2_{H^{k-1}(Y)}+
% \|\mathring{\tilde{A}}\|^2_{H^{k-1}(Y)} +
%\nonumber
%\\&& \nonumber  \|\mathring{\gamma}\|^2_{H^{k-1}(Y)}
%+
%\|\zGamma \|^2_{H^{k-1}(Y)}\Big)ds\,du
%%
+2ae^{\lambda v}\hat C(a,v)+ a %\sup_{u\in[0,a]}
e^{\lambda v} \mcE_{k,\lambda}[\psi_{i+1},v]\Big\}\,,
\eeal{Ctildedef2}
 is preserved under iteration, and we are led to:

\begin{Lemma}
\label{Lenergy2}  Under the hypotheses of Lemma~\ref{Lenergy},
%assume moreover that \eq{intpsic1} and \eq{intphic1} hold.
the inequalities \eq{intpsic} and \eq{phipresint} are preserved
under iteration, and there exist constants
$$
 \lambda_2=\lambda_2(k,C_0,\Cdiv,Y,\tCpsi(a_0,b_0),\tCvarphi(a_0,b_0))
 \,,
$$
%,
$C_7=C_7(Y,k,C_0,\Cdiv)$ and
$$
 C_8=C_{8}(Y,k,C_0,\Cdiv,a_0,b_0,\tCpsi(a_0,b_0),\tCvarphi(a_0,b_0))
$$
such that for all $\lambda \ge \lambda_2$ we have, for $(u,v)\in
[0,\hat a_i]\times [0,b_0]$,
\bean
 \lefteqn{ \|f_{i+1}(u,v)\|^2_{H^{k-1}(Y)}\le  C_{7}e^{\lambda
 (a_0+b_0)}\Big\{
 \ykmns{\overline {\varphi}_{i+1}(v)}}\\
 &&+\ykmns{\overline {\psi}_{i+1}(u)}
%\\
%%
% &&
 +\int_0^{a_0} e^{-\lambda s}
 M_k(s,v)
 %\Big(\ykmns{\mathring
% \newg_\varphi(s,v)}+\|\zA \|^2_{H^{k-1}(Y)}+
% \|\zGamma \|^2_{H^{k-1}(Y)}\Big)
ds
  %\phantom{xxxxxx}
%  \nonumber
%\\
%%
% &&
%  \nonumber
+
 \int_0^{b_0} e^{-\lambda s}
 M_k(u,s)
 %\Big(\ykmns{\mathring
% \newg_\psi(u,s)}+\|\zA \|^2_{H^{k-1}(Y)}+
% \|\zGamma \|^2_{H^{k-1}(Y)}\Big)
ds \nonumber
\\
 &&
+C_8\Big(\hat C(a_0,b_0)+ \tCpsi(a_0,b_0)+
\tCvarphi(a_0,b_0)\Big)\Big\}\,.
%\nonumber \\
% &&
\eeal{bid31}
\myqed\end{Lemma}

From now on, we will use the inequalities \eq{bid5a},
\eq{bid23} and \eq{bid31}
$$
 \mbox{with  $\lambda$ chosen to be the largest of $\lambda_0, \;\lambda_1$ and
 $\lambda_2$}
$$
regardless of the value of the parameter $\lambda$ that might occur in the equation in which one of these inequalities is being used.

In what follows the letter $C$ will denote a constant which depends
perhaps upon $C_0$, $\hat C$, $\tCpsi$, $\tCvarphi$, $Y$, $a_0$,
$b_0$ and $k$, and which may vary from line to line; similarly the
numbered constants $C_{{n}}$ that follow may depend upon all those
quantities, but not on $i$. We wish to show that we can choose
$0<a_* \le a_0$ small enough so that $ \hat a_i\ge a_*$, hence  for
$0\le u \le a_*$ the inequalities \eq{bid9m}, \eq{intpsic1} and
\eq{intphic1} hold. Let
\bel{10X13.1}
 \mbox{$k_1 $ be the smallest integer such that $k_1>(n-1)/2+3$}
  \,.
\ee
For $k\ge k_1$, from \eq{bid31} with $k=k_1$ we obtain, by Sobolev's
embedding, for $0\le u\le \hat a_i$, $0\le v\le b_0$,
\bel{st1}
  \| f_{i+1}(u,v)\|_{C^{2}(Y)}\le C
 \,.
\ee
It follows from the equations satisfied by $f$ that
\begin{eqnarray}
 &
  \| \partial_u
    \varphi_{i+1}(u,v)\|_{C^{1}(Y)}\le C
    \,,
 &
 \label{st2}
\\
 &
 \|\partial_v
    \psi_{i+1}(u,v)\|_{C^{1}(Y)}\le C
    \,.
    &
     \label{st3}
\end{eqnarray}
Integrating \eq{st2} in $u$ from $(0,v)$ to $(u,v)$ we find
\bel{st4}
 \| \varphi_{i+1}(u,v)\|_{C^{ 1}(Y)}\le C_0+Cu\le
 C_0+Ca \le 2C_0
\ee
for $a$ small enough, namely
\bel{14IX13.1}
 0\le a \le \min(
 \hat a_i,C_0 C^{-1})
 \,.
\ee
(Note that the bound is independent of $k$.) Further,
\bel{14IX13.2}
 \mbox{$\| \varphi_{i+1}(u,v)\|_{C^{ 1}(Y)}$ cannot exceed $2C_0$ for $0\le u \le C_0 C^{-1}$.}
\ee

 Next, we $u$-differentiate the
equation satisfied by $\psi_{i+1}$, \bean
  (A^\mu_{\psi\psi}\nabla_\mu)_i {\partial \psi_{i+1} \over\partial
  u} &=& -\partial_u\Big((A^\mu_{\psi\psi}\nabla_\mu)_i \Big){ \psi_{i+1}}
  - \partial_u\Big((A^\mu_{\psi\varphi}\nabla_\mu)_i \varphi_{i+1}
  -(\newg_\psi)_i\Big)
  \\ &=:& (B_\psi)_i {\partial \psi_{i} \over\partial
  u} +(b_\psi)_i\,, \label{st5}
\eea
where, symbolically,
$$
 (B_\psi)_i %(\cdot)
 := -\partial_\psi
 \Big((A^\mu_{\psi\psi}\nabla_\mu)_i \Big)
 %(\cdot)
 { \psi_{i+1}}
  - \partial_\psi\Big((A^\mu_{\psi\varphi}\nabla_\mu)_i \varphi_{i+1}
  -(\newg_\psi)_i\Big)
  %(\cdot)
  \,.
$$
The system \eq{st5} is again a symmetric hyperbolic system of first
order, linear in $\partial_u \psi_{i}$ and $\partial_u \psi_{i+1}$,
to which we apply \eq{volest2} with $\mcU=\{u\}\times[0,v]$, with
$0\le u \le \hat a_i$. Note that, from the definition of
 $\hat a_i$ (see \eq{bid9m}) the relevant constant $\hat C_1$ there will be bounded from above by a finite constant, say, $\check C_1\ge 1$ which is $i$, $\lambda$, and $f_i$--independent.
% \twr{6XI13: Juste un commentaire pour dire que je suis d'accord qu'il ne faut pas mettre les normes infinies sur les dérivées dans la définition de %$a_i$, mais est ce que c'est pas nécessaire de mettre un commentaire pour dire que contrairement $\hat C_1$, la constante $\check C_1$ ici ne depend %pas des normes infinies de $\partial\psi_{i+1}$ et $\partial_u\psi_i$ pare ce que le systès en considération en linéaire?\\ -- \\ ptc: see previous %mnote}
 Thus, by \eq{volest2} with $k$ there replaced by $m$,
 \bea
 \lefteqn{
e^{-\lambda(u+v)}\|\partial_u\psi_{i+1}(u,v)\|^2_{H^{m}(Y)}\leq \check{C}_1\Big\{  e^{-\lambda u}\|\partial_u\psi_{i+1}(u,0)\|^2_{H^{m}(Y)}}\nonumber\\
&&+ \int_0^v e^{-\lambda(u+s)}\Big(\|(\nabla_\mu
A^\mu_{\psi\psi})_i\|_{L^{\infty}(Y)}-c\lambda\Big)\|\partial_u\psi_{i+1}\|^2_{H^m(Y)}+
 \|\partial_u\psi_{i+1}\|_{H^m(Y)} \times\nonumber
\\
 &&\Big( \|\partial_u\psi_{i+1}\|_{H^m(Y)}
  + \|(A)_i\|_{H^{m}(Y)}+
    \|(\tilde{A})_i\|_{H^{m-1}(Y)}
    +
    \|(\Gamma)_i\|_{H^{m}(Y)}+\|(\gamma)_i\|_{H^{m}(Y)}\nonumber
\\
  &&+ \|(B_\psi)_i {\partial \psi_{i} \over\partial
  u} +(b_\psi)_i\|_{H^{m}(Y)}
    + \|(\tilde{B}_\psi)_i {\partial \psi_{i} \over\partial
  u} +(\tilde{b}_\psi)_i\|_{H^{m-1}(Y)}\Big)ds\Big\}
  \,,
\eeal{t1}
where $
(\tilde{B}_\psi)_i=(A^v_{\psi\psi})^{-1}_i(\tilde{B}_\psi)_i\,, \;
(\tilde{b}_\psi)_i=(A^v_{\psi\psi})^{-1}_i(\tilde{b}_\psi)_i $ and
$(A)_i$,  the value of the matrix $A$ as determined by $f_i$. Again
from \eq{Moser3} we have
\bean
 \lefteqn{
 \|(A)_i\|_{H^{m}(Y)}+\|(\tilde{A})_i\|_{H^{m-1}(Y)}+ \|(\Gamma)_i\|_{H^{m}(Y)}+\|(\gamma)_i\|_{H^{m}(Y)}\leq }
\nonumber
\\
&&
 C(Y,m,\|f_i\|_{L^{\infty}})\Big(\|f_i\|_{H^{m}}+\|\zA\|_{H^{m}(Y)}+\|\mathring{\tilde{A}}\|_{H^{m-1}(Y)}+
 \|\zGamma\|_{H^{m}(Y)}+\|\mathring{\gamma}\|_{H^{m}(Y)}\Big)
 \nonumber
 \\
&&
 \leq
 C(Y,k,C_0)\Big(\|f_i\|_{H^{m}(Y)}+\sqrt{M_k(u,s)}\Big)  \qquad \text{for}\;\; m\leq k\,.
\eeal{8X13.2}
Now, after eliminating $\partial_v \psi_{i+1}$ and
$\nabla_B\partial_u\varphi_{i+1}$ using the equations, we have
$(B_{\psi})_i=B_{\psi}\left(f_i, f_{i+1}, \znabla_Bf_{i+1}\right) $
and
$$
(b_\psi)_i=b_{\psi}\left(f_i, f_{i+1}, \znabla_Bf_{i+1},
\znabla_B\znabla_Cf_{i+1}, \partial_u\varphi_i \right) \,,
$$
which are
affine functions of $ \znabla_Bf_{i+1}$ and
$\znabla_B\znabla_Cf_{i+1}$. Using again \eq{Moser3} we have:
\bean \lefteqn{ \|(B_\psi)_i {\partial \psi_{i} \over\partial
  u} +(b_\psi)_i\|_{H^{m}(Y)}
    + \|(\tilde{B}_\psi)_i {\partial \psi_{i} \over\partial
  u} +(\tilde{b}_\psi)_i\|_{H^{m-1}(Y)}\leq}\\
  &&
  C\Big(Y,m,\|f_i\|_{L^{\infty}},\|f_{i+1}\|_{W^{2,\infty}}, \|\partial_u\varphi_i\|_{L^{\infty}}
  %,\|\partial_u\psi_i\|_{L^{\infty}}
  \Big)\times
 \nonumber
\\
  &&
  \Big( \|\partial_u\psi_i\|_{H^{m}(Y)}+\|f_i\|_{H^{m}(Y)}+\|f_{i+1}\|_{H^{m+2}(Y)} + \|\partial_u\varphi_i\|_{H^{m}(Y)}
 \nonumber
\\
  &&
  \phantom{}+
  \|\mathring{B}_{\psi}\|_{H^{m}(Y)}+\|\mathring{b}_{\psi}\|_{H^{m}(Y)}+
  \|\mathring{\tilde{B}}_{\psi}\|_{H^{m}(Y)}+\|\mathring{\tilde{b}}_{\psi}\|_{H^{m}(Y)}\Big)
  \,.
   \label{8X13.2-a}
\eea

For $k\ge 3$ let $\hat M_k(u,v)$ be any function such that
\bean
 \hat M_k(u,v)&\ge &  M_k(u,v) +
  \|\mathring{B}_{\psi}(u,v)\|_{H^{k-3}(Y)}^2+\|\mathring{b}_{\psi}(u,v)\|_{H^{k-3}(Y)}^2
\\
 &&  +
  \|\mathring{\tilde{B}}_{\psi}(u,v)\|_{H^{k-3}(Y)}^2+\|\mathring{\tilde{b}}_{\psi}(u,v)\|_{H^{k-3}(Y)}^2
  \,.
\eeal{8X13.1}
Then, by using simultaneously inequalities \eq{bid9m}, and
\eq{bid31}-\eq{st2} with $i+1$ there replaced by $i$ (note that
$\hat a_i$ is decreasing by definition),   one obtains for $m+2=k-1$
\bean
 \lefteqn{
\|(B_\psi)_i {\partial \psi_{i} \over\partial
  u} +(b_\psi)_i\|_{H^{k-3}(Y)}  + \|(\tilde{B}_\psi)_i {\partial \psi_{i} \over\partial
  u} +(\tilde{b}_\psi)_i\|_{H^{k-3}(Y)}
  }
  &&
\\
 &&
  \leq  C
  \Big( \|\partial_u\psi_i\|_{H^{k-3}(Y)}
   +
  \sqrt{\hat M_k(u,s)} +C\Big)
   \,.
  \phantom{xxxxxxxxxxxx}
\eeal{8X13.3}
Adding \eq{8X13.2} and \eq{8X13.3} we obtain
\bea
 \lefteqn{
 e^{-\lambda(u+v)}\|\partial_u\psi_{i+1}(u,v)\|^2_{H^{k-3}(Y)}\leq C_9\Big\{  e^{-\lambda u}\|\partial_u\overline {\psi}_{i+1}(u)\|^2_{H^{k-3}(Y)}
}
 \nonumber
\\
 &&+
 \int_0^v e^{-\lambda(u+s)}\Big(\|(\nabla_\mu A^\mu_{\psi\psi})_i\|_{L^{\infty}(Y)}-c\lambda\Big)\|\partial_u\psi_{i+1}\|^2_{H^{k-3}(Y)}
 \nonumber
 \\
&&+
 \|\partial_u\psi_{i+1}\|_{H^{k-3}(Y)}\Big(\|\partial_u\psi_{i+1}\|_{H^{k-3}(Y)} + \|\partial_u\psi_{i}\|_{H^{k-3}(Y)}
 \nonumber
 \\
&& \phantom{xxxxxxxxxxxxxxxxxxxxxxxxxxxxxxx}  +\sqrt{\hat M_k(u,s)}+
C_{10}  \Big)ds\Big\}\,. \nonumber
\\
 &&
 \leq C_9\Big\{  e^{-\lambda u}\|\partial_u\overline {\psi}_{i+1}(u)\|^2_{H^{k-3}(Y)} +
 \int_0^v e^{-\lambda(u+s)}\Big(\epsilon \|\partial_u\psi_{i}\|_{H^{k-3}(Y)}^2
  \nonumber
\\
 &&+ \hat M_k(u,s)+ C^2_{10}   +
 \big(\|(\nabla_\mu A^\mu_{\psi\psi})_i\|_{L^{\infty}(Y)}-c\lambda+C(\epsilon)\big)\|\partial_u\psi_{i+1}\|^2_{H^{k-3}(Y)}
  \Big)
    ds\Big\}\,,
 \nonumber
 \\
\eeal{t2}
where in the last step we have used Cauchy-Schwarz with $\epsilon$.
It then follows from \eq{intpsic1} that there exists a constant
 $\lambda_3=\lambda_3(Y,k,C_0, \Cdiv)$ such that for all $\lambda\geq \lambda_3$
 \bea
 \lefteqn{
e^{-\lambda v}\|\partial_u\psi_{i+1}(u,v)\|^2_{H^{k-3}(Y)} \leq C_9\Big\{ \|\partial_u\overline {\psi}_{i+1}(u)\|^2_{H^{k-3}(Y)}}\nonumber\\
&& \phantom{xxxx}+ \int_0^v e^{-\lambda s}\Big(  \hat M_k(u,s)+
C^2_{10} +
 \epsilon  \|\partial_u\psi_{i}\|_{H^{k-3}(Y)}^2
   \Big)ds\Big\}\,.
\eeal{t3}
Set
$$
\check{C}_{\psi}=C_9\Big\{
\sup_{i\in\N}\sup_{u\in[0,a_0]}\|\partial_u\overline
{\psi}_{i}(u)\|^2_{H^{k-3}(Y)} +\int_0^{b_0} e^{-\lambda s}\Big(
\hat M_k(u,s)+ C^2_{10}
   \Big)ds\Big\}
   +1\,.
$$
By an argument which should be standard by now,   one can
choose $\epsilon$ small enough such that the inequality
\bel{st3} \|\partial_u\psi_{i+1}(u,v)\|_{H^{k-3}(Y)}^2\leq 2
e^{\lambda v}\check C_\psi \ee
is \emph{preserved by iteration} on  $[0,\hat
a_{i}]\times[0,b_0]\times Y$.

(This is not good enough yet for our purposes  when $A^u$ depends
upon $\psi$, as we will then need \eq{st3} with $2 C_0$ at the
right-hand side to be able to make sure that the contribution from
$(\nabla_\mu A^\mu)_i$ can be estimated by $\Cdiv$, so that some
more work will have to be done in the general case).

In any case, let
\bel{21X13.1} \mbox{$k_2$ be the smallest integer larger than or
equal to $ \frac{n+7}{2}$.} \ee
For $k\ge k_2$  we can use \eq{st3} with $k$ replaced by
 $k_2$ there and the Sobolev embedding to obtain
\bel{st8}
  \forall (u,v)\in[0,\hat{a}_i]\times[0,b_0], \quad  \left\|{
\partial_u\psi_{i+1} (u,v)}\right\|_{C^{1}(Y)}\le C
 \,.
\ee
By integration in $u$ we therefore find
\bel{st8n}
 \left\|\psi_{i+1} (u,v)\right\|_{C^{1}(Y)}\le C_0+Cu
  \le 2C_0
 \,,
\ee
again in the range~\eq{14IX13.1} (but note that the constant $C$
there might have to be taken larger now, remaining independent of
$i$ and $k$).

Keeping in mind \eq{14IX13.2}, we conclude that the condition
\bel{14IX13.3}
 \mbox{$\| f_{i+1}(u,v)\|_{C^{ 1}(Y)}$ cannot exceed $4C_0$ for $0\le u \le C_0 C^{-1}$}
\ee
is stable under iteration.

Moreover, replacing the bound $C_0C^{-1}$ by a  smaller
$i$-independent number, say $a_*$, if necessary, integration in $u$
shows that
\bel{14IX13.3}
 \mbox{$f_{i+1}(u,v)$ cannot leave the neighborhood $\mcK$ for $0\le u \le a_*$.}
\ee
with $\mcK$   as  in \eq{6IX13.1}.

If $A^v$ does not depend upon $\varphi$,  and $A^u$ does not depend
upon $\psi$, the conditions on   $\partial_u \psi$ and $\partial_v
\varphi$ in \eq{6IX13.1} are irrelevant for all the estimates so
far, and so
\bel{14IX13.4}
 \mbox{the bounds \eq{bid9f-1}
cannot be violated for $0\le u \le a_*$. } \ee
Hence, using the definition of $\Cdiv$,
\bel{14IX13.4}
 \mbox{the inequalities \eq{bid9m}
cannot be violated for $0\le u \le a_*$. } \ee
Recall that   $\hat a_i$ was defined as either $a_0$  or the first
number at which the inequalities \eq{bid9m} fail for $f_i$ or
$f_{i+1}$. So, if we assume that the inequalities   \eq{bid9m} hold
at the induction step $i$ with $a_i\ge a_*$, we conclude that
$a_{i+1}\ge a_*$ as well. Hence
$$
 \hat a_i \ge a_*
 \,.
$$

The above implies that
 \eq{intpsic1}
and \eq{intphic1} hold  for $0\le u \le a_*$.  We have therefore
obtained:

\begin{Proposition}
\label{Piter} Let $k> (n+7)/2$,  assume that   $A^v$ does not depend
upon $\varphi$, and that $A^u$ does not depend upon $\psi$. Suppose
that there exists a constant $\mcC$ such that
\begin{equation}\label{Itcond}
   \sup_{\mcN^-\cup\mcN^+}\Big\{ \left|\overline{\partial_v
    {f}}_i\right| + \left|\overline{\partial_u   f}_i\right| +
   \ykn{ \overline {f}_i(u,v)}+ M_k(u,v)
   %\ykn{\mathring
%   g(u,v)} + \|\mathring{\gamma} (u,v)\|^2_{H^{k}(Y)}\\
%   &+ \|\zA(u,v)\|^2_{H^{k}(Y)}+\|\mathring{\tilde{A}}(u,v)\|^2_{H^{k}(Y)}
%    + \|\mathring{\tilde{\newg}} (u,v)\|^2_{H^{k-1}(Y)} + \|\zGamma (u,v)\|^2_{H^{k}(Y)}
\Big\}\le \mcC\,.
\end{equation}
There exists a constant $0<a_*=a_*(a_0,b_0,\mcC,Y)\le a_0$ such that
all the fields $f_i$ satisfy the hypotheses of
Lemmata~\ref{Lenergy}--\ref{Lenergy2} on $[0,a_*]\times[0,b_0]\times
Y$, as well as their conclusions with $\hat a_i$ replaced by $a_*$.
\end{Proposition}

It remains to obtain the pointwise bounds \eq{bid9f}, \eq{intpsic1}
and \eq{intphic1} in the general case; this will follow from
pointwise estimates on $\partial _v \varphi$, and on improved
estimates on $\partial_u \psi$.

We start by showing that the inequality
$$
\sup_i\sup_{(u,v)\in[0,a_*]\times[0,b_0]} | \partial_u\psi_i (u,v)
|\leq 2 \sup_i\sup_{ v \in[0,b_0] } | \partial_u\psi_i (0,v)  | +1
$$
is preserved by iteration, after reducing $a_*$ if necessary.

We consider the restriction of the $u$-differentiated equation satisfied by
$\psi_{i+1}$ on $\mcN^-$, that is for $u=0$, which we write as
\bel{st5-09X13}
  \overline{(A^\mu_{\psi\psi}\nabla_\mu)_{i}} \overline{{\partial \psi_{i+1} \over\partial
  u}} = \overline{(B_\psi)_{i}} \overline{{\partial \psi_{i} \over\partial
  u}} +\overline{(b_\psi)_{i}}
   \,.
\ee
Setting $\Psi_i={\partial \psi_{i} \over\partial
  u}-\overline{{\partial \psi_{i} \over\partial u}
   }$ and subtracting \eq{st5-09X13} from \eq{st5}  gives an equation of the form
\bel{9X13-a} (A^\mu_{\psi\psi}\nabla_\mu)_i \Psi_{i+1}
=(B_\psi)_i\Psi_{i} +\mathcal{E}_i
 \,,
\ee
where
\bel{12X13.2} \mathcal{E}_i=\underbrace{
-\left((A^\mu_{\psi\psi}\nabla_\mu)_i
-\overline{(A^\mu_{\psi\psi}\nabla_\mu)_i} \right)
\overline{{\partial  \psi_{i+1} \over\partial
  u}}}_{\Delta_1}+\underbrace{\left((B_\psi)_i-\overline{(B_\psi)_i} \right)\overline{{\partial \psi_{i} \over\partial
  u}}}_{\Delta_2} + \underbrace{(b_\psi)_i-\overline{(b_\psi)_i}}_{\Delta_3}\,.
\ee
It is easy to see that both $(B_{\psi})_i$ and $(b_{\psi})_i$ are
affine in $\znabla_B f_{i+1}$ and $\znabla_B \znabla_C f_{i+1}$ with
coefficients depending upon $f_{i}$ and $f_{i+1}$, thus if $k-1>
(n-1)/2+3$ by \eq{bid31}, \eq{st1}, \eq{Moser1} and \eq{Moser3}, we
have
\bel{12X13.01}
\|(B_{\psi})_i\|_{H^{k-2}(Y)}+\|(b_{\psi})_i\|_{H^{k-3}(Y)}+\|(B_{\psi})_i\|_{W^{1,\infty}(Y)}+\|(b_{\psi})_i\|_{W^{2,\infty}(Y)}<C
 \,.
\ee
Further,
\beaa
 (A^\mu_{\psi\psi}\nabla_\mu)_i \overline{{\partial\psi_{i+1} \over\partial
  u}}
&=&  %
%\Big((A^v_{\psi\psi})_i\partial_v + (A^v_{\psi\psi}\gamma_{\psi\psi,v})_i
% +(A^B_{\psi\psi})_i\znabla_B +(A^B_{\psi\psi}\Gamma_{\psi\psi,B})_i\Big) \overline{{\partial \psi_{i+1} \over\partial
%  u}}
% \\
% &=&
% (A^v_{\psi\psi})_i\overline{\partial_v{\textcolor{red}{\partial \psi_{i+1} \over\partial u}}}
% +
%\Big((A^v_{\psi\psi}\gamma_{\psi\psi,v})_i
% +(A^B_{\psi\psi})_i\znabla_B +(A^B_{\psi\psi}\Gamma_{\psi\psi,B})_i\Big) \overline{{\partial \psi_{i+1} \over\partial
%  u}}
% \\
% &=&
 (A^v_{\psi\psi})_i\overline{\partial_u \partial \psi_{i+1} \over\partial v}
 +
\Big((A^v_{\psi\psi}\gamma_{\psi\psi,v})_i
 +(A^B_{\psi\psi})_i\znabla_B +(A^B_{\psi\psi}\Gamma_{\psi\psi,B})_i\Big) \overline{{\partial \psi_{i+1} \over\partial
  u}}
  \,.
\eeaa
%and
%\beaa
%\overline{(A^\mu_{\psi\psi}\nabla_\mu)_i} \overline{{\partial  \psi_{i+1} \over\partial u}}
%%&=&
%%\overline{\Big((A^v_{\psi\psi})_i\partial_v + (A^v_{\psi\psi}\gamma_{\psi\psi,v})_i
%% +(A^B_{\psi\psi})_i\znabla_B +(A^B_{\psi\psi}\Gamma_{\psi\psi,B})_i\Big) {\partial \psi_{i+1} \over\partial u}}\\
%&=&
%\overline{(A^v_{\psi\psi})_i\partial_v{\partial \psi_{i+1} \over\partial u}} + \overline{\Big((A^v_{\psi\psi}\gamma_{\psi\psi,v})_i
% +(A^B_{\psi\psi})_i\znabla_B +(A^B_{\psi\psi}\Gamma_{\psi\psi,B})_i\Big) {\partial \psi_{i+1} \over\partial u}}\\
% &=&
%\overline{(A^v_{\psi\psi})_i}\overline{\partial_u{\textcolor{red}{\partial \psi_{i+1} \over\partial v}}} + \Big(\overline{(A^v_{\psi\psi}\gamma_{\psi\psi,v})_i}
% +\overline{(A^B_{\psi\psi})_i}\znabla_B +\overline{(A^B_{\psi\psi}\Gamma_{\psi\psi,B})_i}\Big) \overline{{\partial \psi_{i+1} \over\partial u}}\\
%\eeaa
%%
Hence we have the following contribution to the first term in
\eq{12X13.2}:
\beaa \lefteqn{
  \left((A^v_{\psi\psi})_i-\overline{(A^v_{\psi\psi})_i}\right)\overline{\partial_u\partial_v\psi_{i+1}}
  }
\\
&=& \bigg\{u\int_0^1\frac{\partial (A^v_{\psi\psi_i})}{\partial
u}\Big(tu,v, tf_i(u,v)+(1-t)f_i(0,v)\Big)dt + (f_i(u,v)-f_i(0,v))
 \times
\\
&&
%+
%(\varphi_i(u,v)-\varphi_i(0,v)\int_0^1\frac{\partial (A^v_{\psi\psi)_i}}{\partial\varphi}\Big(tu,v, tf_i(u,v)+(1-t)f_i(0,v)\Big)dt
%\\
%&&
\int_0^1\frac{\partial (A^v_{\psi\psi_i})}{\partial f}\Big(tu,v,
tf_i(u,v)+(1-t)f_i(0,v)\Big)dt \bigg\}\overline{\partial_u
\partial_v\psi_{i+1}}
\eeaa Now recall that by hypothesis the quantity $
\sup_{v\in[0,b_0]}\sup_i\|\overline{\partial_v \partial_u
\psi_i}\|_{H^k(Y)} $ is bounded and that the derivative ${\partial
\psi_{i+1} \over\partial v}$ is an affine function of $f_{i+1}$ and
$\znabla f_{i+1}$ with coefficients depending upon $f_i$. Thus there
exists a constant
\bel{25X13.p1}
 C=C(C_0,
 \sup_{v\in[0,b_0]}\sup_i\|\overline{\partial_v \partial_u \psi_i}\|_{H^k(Y)} )
 >0
  \,,
\ee
which is $i-$independent, such that $ \forall (u,v)\in[0, \hat
a_i]\times[0,b_0],$
$$
 \|\left((A^\mu_{\psi\psi}\nabla_\mu)_i -\overline{(A^\mu_{\psi\psi}\nabla_\mu)_i} \right) \overline{{\partial_u  \psi_{i+1} }}\|_{L^2(Y)}\leq C\left(u +\|f_i(u,v)-f_i(0,v)\|_{L^2(Y)} \right)
 \,.
$$
The $L^2-$norm of remaining terms in the first term $\Delta_1$ of
\eq{12X13.2} are estimated in the same way with $
(A^\mu_{\psi\psi})_i $
 replaced successively by
 $
(A^v_{\psi\psi}\gamma_{\psi\psi,v})_i ,\; $ $
(A^B_{\psi\psi})_i\znabla_B,\; $ $
(A^B_{\psi\psi}\Gamma_{\psi\psi,B})_i $ and $ \overline{{\partial_u
\partial_v\psi_{i+1} }} $ replaced by $ \overline{{\partial_u
\psi_{i+1}}} $ leading to the following estimate for $\Delta_1$:
\bel{12X12.3}
 \forall (u,v)\in[0, \hat a_i]\times[0,b_0], \;
 \|\Delta_1(u,v)\|_{L^2(Y)}\leq C\left(u +\|f_i(u,v)-f_i(0,v)\|_{L^2(Y)} \right)
 \,.
\ee

We continue with the analysis of the second term $\Delta_2$ of
\eq{12X13.2}. The explicit expression of $(B_{\psi})_i$ shows that
$(B_{\psi})_i$ is a collection of terms of the form  $\Gamma_i P_r
f_{i+1}, \; 0\leq r\leq  1 $ where the $\Gamma_i$'s are smooth
function depending upon the fields $f_i$. We order these terms in an
arbitrary way and write
$$
 (B_{\psi})_i=\sum_{m=1}^p \Gamma_{i,m} P_{r_m} f_{i+1}
$$
where  $ P_{r_m}$ is either the identity or $\znabla_B$. We have
\beaa
 \Delta_2
  &=&
   \left((B_\psi)_i-\overline{(B_\psi)_i} \right)\overline{ \partial_u \psi_{i} }
\\
 &=&
 \sum_{m=1}^p \left(\Gamma_{i,m} P_{r_m} f_{i+1} -\overline{\Gamma_{i,m} P_{r_m} f_{i+1}}\right)\overline{ \partial_u \psi_{i}}
  \,.
\eeaa
Thus,
\beaa
 \lefteqn{\left(\Gamma_{i} P_{r} f_{i+1} -\overline{\Gamma_{i} P_{r} f_{i+1}}\right)\overline{ \partial_u \psi_{i} }
 }
\\
 &&
  =
 \Gamma_{i}\left( P_{r} f_{i+1}-\overline{ P_{r} f_{i+1}}\right)\overline{ \partial_u \psi_{i} }
 +
 \left(\Gamma_{i} -\overline{\Gamma_{i}} \right)\overline{P_{r} f_{i+1}} \overline{ \partial_u \psi_{i} }
\\
 &=&
 \Gamma_{i}\left( P_{r} f_{i+1}-\overline{ P_{r} f_{i+1}}\right)\overline{ \partial_u \psi_{i} }
\\
 &&
 +
 u\left[\int_0^1\frac{\partial\Gamma_{i}}{\partial u}(tu, tf_i(u,v)+(1-t)f_i(0,v))\, dt\right]\overline{P_{r} f_{i+1}} \overline{ \partial_u \psi_{i} }
 \\
 &&
+(f_i(u,v)-f(0,v))\left[\int_0^1\frac{\partial\Gamma_{i}}{\partial f}(tu,
tf_i(u,v)+(1-t)f_i(0,v))\, dt
 \right]
  \overline{P_{r} f_{i+1}} \overline{
\partial_u \psi_{i} } \,. \eeaa
We then see that \beaa \lefteqn{ \|\left(\Gamma_{i} P_{r} f_{i+1}
-\overline{\Gamma_{i} P_{r} f_{i+1}}\right)\overline{ \partial_u
\psi_{i} }\|_{L^2(Y)} \leq}
\\
 &&
 C(C_0)\left(u +\|f_i(u,v)-f_i(0,v)\|_{L^2(Y)}+\|f_{i+1}(u,v)-f_{i+1}(0,v)\|_{H^1(Y)} \right)
 \,,
\eeaa
which gives
\bel{12X13.4}
 \|\Delta_2\|_{L^2(Y)}
 \leq
 C(C_0)\left(u +\|f_i(u,v)-f_i(0,v)\|_{L^2(Y)}+\|f_{i+1}(u,v)-f_{i+1}(0,v)\|_{H^1(Y)} \right)
 \,.
\ee
As far as the last term $\Delta_3$ of \eq{12X13.2} is concerned, we
note that $(b_{\psi})_i$ is a sum of terms of the form
$$
\tilde{\Gamma}_i\znabla_{r_{1}}\ldots \znabla_{r_{j}}f_{i+1} \,,
$$
with $0\leq j\leq 2$ and $\tilde{\Gamma}_i$ depending upon $f_i$ and
$\partial_u\varphi_i$. Thus as in the previous case, we see that
\bea
 \label{12X13.5}
 \|\Delta_3\|_{L^2(Y)}
&\leq&
 C(C_0)\Big(u +\|f_i(u,v)-f_i(0,v)\|_{L^2(Y)}+\|\partial_u \varphi_i(u,v)-\partial_u \varphi_i(0,v)\|_{L^2(Y)}
 \nonumber
\\
 &&\qquad\qquad\qquad+\|f_{i+1}(u,v)-f_{i+1}(0,v)\|_{H^2(Y)} \Big)
 \,.
\eea
Now from \eq{st2} and \eq{st8} we find that (note that from these
inequalities and the equation satisfied by $\varphi_{i+1}$, the
$L^\infty$ norm of $\partial^2_u\varphi_{i}$ is uniformly bounded):
\bel{12X13.6} \forall (u,v)\in[0, \hat a_i]\times[0,b_0]\times Y,
\;\;\|\mathcal{E}_i(u,v)\|_{L^2(Y)}\leq Cu \,. \ee
By \eq{12X13.01}, the $L^2-$norm of the right-hand-side of
\eq{9X13-a} is estimated as:
\bea \label{12X13.7} \|(B_\psi)_i\Psi_{i}
+\mathcal{E}_i\|_{L^2(Y)}&\le& \|(B_\psi)_i\Psi_{i}\|_{L^2(Y)}
+\|\mathcal{E}_i\|_{L^2(Y)} \nonumber
\\
&\leq&
 Cu+\|(B_\psi)_i\|_{L^{\infty}(Y)}\|\Psi_{i}\|_{L^2(Y)}
\nonumber
\\
 &\leq&
 C\left( u+\|\Psi_{i}\|_{L^2(Y)}\right)
 \,.
\eea
Next, we write the energy estimate for the  system \eq{9X13-a}.
Consider the vector field (recall $w_r=e^{-\lambda(u+v)}$)
\bel{9X13-c}
 Z{^\mu}:= w_r \langle\Psi_{i+1},(A^\mu_{\psi\psi})_i
 \Psi_{i+1}\rangle
 \,,
\ee
 so that
 \beaa \nabla_\mu (Z{^\mu})
  & = &
  \Big\{-2\lambda  \langle \Psi_{i+1},(A^v_{\psi\psi})_i \Psi_{i+1}\rangle
 \\
 &&
 +  \langle  \Psi_{i+1},(\nabla_\mu A^\mu_{\psi\psi})_i \Psi_{i+1}\rangle + 2\langle  \Psi_{i+1},(A^\mu_{\psi\psi}\nabla_\mu)_i\Psi_{i+1}\rangle  \Big\} w_r
 \,.
\eeaa
We apply Stokes' theorem on the set $\{u\}\times[0,v]\times Y$ and
obtain
\beaa e^{-\lambda v}\|\Psi_{i+1}(u,v)\|^2_{L^{2}(Y)} &\leq&
C\Big\{ \|\Psi_{i+1}(u,0)\|^2_{L^2(Y)}  +e^{\lambda
u}\int_0^v\nabla_\mu (Z{^\mu})(s,v)dsd\mu_Y \Big\}
\\
&\leq&
C\Big\{\|\Psi_{i+1}(u,0)\|^2_{L^2(Y)} -2\lambda\int_0^ve^{-\lambda
s}  \Big( \langle  \Psi_{i+1},(A^v_{\psi\psi})_i
 \Psi_{i+1}\rangle
\\
&& + \langle  \Psi_{i+1},(\nabla_\mu A^\mu_{\psi\psi})_i
\Psi_{i+1}\rangle + 2\langle
\Psi_{i+1},(A^\mu_{\psi\psi}\nabla_\mu)_i
 \Psi_{i+1}\rangle\Big)dsd\mu_Y
 \Big\}
\\
&\leq&
 C\Big\{\|\Psi_{i+1}(u,0)\|^2_{L^2(Y)}\\
&& + \int_0^ve^{-\lambda s}\Big\{\left(\|(\nabla_\mu
A^\mu_{\psi\psi})_i\|_{L^{\infty}(Y)}-2c\lambda\right)\|\Psi_{i+1}(u,s)\|^2_{L^2(Y)}
\\
&& +2\|\Psi_{i+1}(u,s)\|_{L^2(Y)} \|((B_\psi)_i\Psi_{i}
+\mathcal{E}_i)(u,s)\|_{L^2(Y)}\Big\}ds \Big\} \,.
 \eeaa
Now from \eq{12X13.7}, we have:
\beaa
 e^{-\lambda v}\|\Psi_{i+1}(u,v)\|^2_{L^{2}(Y)}
 &\leq&
 C\Big\{\|\Psi_{i+1}(u,0)\|^2_{L^2(Y)}
 \\
&&
 +\int_0^ve^{-\lambda s}\Big\{\left(\|
  (\nabla_\mu A^\mu_{\psi\psi})_i\|_{L^{\infty}(Y)}-2c\lambda\right)
   \|\Psi_{i+1}(u,s)\|^2_{L^2(Y)}
\\
&&
 + \|\Psi_{i+1}(u,s)\|_{L^2(Y)} \left(u+\|\Psi_i(u,s))\|_{L^2(Y)}\right)\Big\}ds\Big\}
\\
&\leq&
 C\Big\{\|\Psi_{i+1}(u,0)\|^2_{L^2(Y)}
\\
&&
 +\int_0^ve^{-\lambda s}\Big\{\left(\|(\nabla_\mu A^\mu_{\psi\psi})_i\|_{L^{\infty}(Y)}-2c\lambda+C(\epsilon)\right)\|\Psi_{i+1}(u,s)\|^2_{L^2(Y)}
\\
&& + \left( u^2+\epsilon\|\Psi_i(u,s)\|^2_{L^2(Y)}\right)\Big\}ds\Big\}
\\
&\leq& C\Big\{\|\Psi_{i+1}(u,0)\|^2_{L^2(Y)}+\int_0^ve^{-\lambda s}
 \left(
    u^2+\epsilon\|\Psi_i(u,s)\|^2_{L^2(Y)}
     \right)
     \,ds\Big\} \,, \eeaa
for $\lambda$ large enough. Thus there exists $\lambda_\epsilon>0$
such that for all $(u,v)\in \Omega_{\hat a_i}$,
$$
e^{-\lambda_\epsilon v}\|\Psi_{i+1}(u,v)\|^2_{L^{2}(Y)}
 \leq
C\Big\{\|\Psi_{i+1}(u,0)\|^2_{L^2(Y)}+ \int_0^ve^{-\lambda_\epsilon
s}( u^2+\epsilon\|\Psi_i(u,s)\|^2_{L^2(Y)})ds\Big\}
\,.
$$
Recall $ \Psi_{i+1}(u,0)=\partial_u\psi(u,0)-\partial_u\psi(0,0) $,
thus,
$$
 |\Psi_{i+1}(u,0)|\leq u\cdot \underbrace{ \sup_{i\in\N}\sup_{u\in[0,a_0]} \|\partial^2_u\psi_i(u,0)\|_{L^{\infty}(Y)}}_{=:\hat c}
 \,,
$$
leading to
\beaa
 \lefteqn{
 e^{-\lambda_\epsilon v}\|\Psi_{i+1}(u,v)\|^2_{L^{2}(Y)}
 }
 &&
\\
 &&
 \leq
 u^2\left(\hat c^2\mu_Y(Y)+ C\int_0^ve^{-\lambda_\epsilon s}ds\right) +\epsilon C\int_0^ve^{-\lambda_\epsilon s}\|\Psi_i(u,s)\|^2_{L^2(Y)}ds
    \,.
\eeaa
Suppose now for the purpose of induction that (the constant
$\overline  C_0$ will be chosen shortly, see \eq{21X13.2})
\bel{12X13.8}
 (u,v)\in [0, \hat a_i]\times[0,b_0], \;\|\Psi_{i}(u,v)\|^2_{L^{2}(Y)}\leq \overline  C_0e^{\lambda_\epsilon(v-b_0)}
 \,,
\ee
then by the previous inequality,
$$
 e^{-\lambda_\epsilon (v-b_0)}\|\Psi_{i+1}(u,v)\|^2_{L^{2}(Y)} \leq
 u^2 \underbrace{e^{\lambda_\epsilon b_0}\big(\hat c^2\mu_Y(Y)+  \frac{C }{\lambda_\epsilon}\big)}_{=:C(\lambda_\epsilon)} +\epsilon C \overline  C_0b_0
 \,.
$$
We choose $\epsilon$ small enough such that $\epsilon Cb_0 \leq
1/2$. Once this choice of $\epsilon$ is made (then
$C(\lambda_\epsilon)$ is fixed)
 we see that
 $
 u^2C(\lambda_\epsilon)\leq \overline  C_0/2
 $
 provided that
 \bel{13X13-1}
 0< u\leq  \overline  C_0\left(\sqrt{2C(\lambda_\epsilon)}\right)^{-1}
 \,.
 \ee
We have thus proved that in the range of the $u-$variable given by
\eq{13X13-1} it holds that
$$
\|\Psi_{i+1}(u,v)\|^2_{L^{2}(Y)}
 \leq
\overline  C_0 e^{ \lambda_\epsilon (v-b_0)} \,.
$$
This proves that \eq{12X13.8} is preserved by iteration.

Now, recall that from \eq{st3},
$$
 e^{-\lambda v}\|\partial_u\psi_{i+1}(u,v)\|_{H^{k-3}(Y)}^2\leq 2 \check C_\psi
\,.
$$
Note that the constant $\check C_\psi$ in \eq{12X13.8} can be chosen
independently of $\lambda$, and that the $\lambda$ here is
independent from $\lambda_\epsilon$ in the previous inequalities,
but it is convenient to chose them to be equal, and we shall do so.
Thus we can write
$$
 e^{-\lambda v/2}\|\Psi_{i+1}(u,v)\|_{H^{k-3}(Y)}
 \leq C_0+(2 \check C_\psi)^{1/2} =: C
\,.
$$
By interpolation, there exists a  constant $c_m>0$ such that for all
$m\in(0,k-3)$,
$$
 \|\Psi_{i+1}(u,v)\|_{H^{m}(Y)} \leq c_m\|\Psi_{i+1}(u,v)\|^{\theta}_{H^{k-3}(Y)}\|\Psi_{i+1}(u,v)\|^{1-\theta}_{L^{2}(Y)}
 \,,
$$
with a certain constant $\theta\in(0,1)$. Then, multiplying by $
e^{-\lambda v/2}$ we obtain,
\beaa
 e^{-\lambda v/2} \|\Psi_{i+1}(u,v)\|_{H^{m}(Y)}
 &\leq&
 c_m\|e^{-\lambda v/2}\Psi_{i+1}(u,v)\|^{\theta}_{H^{k-3}(Y)}\|e^{-\lambda v/2}\Psi_{i+1}(u,v)\|^{1-\theta}_{L^{2}(Y)}
 \\
 &\leq &
  c_m C^{\theta}(\overline  C_0 e^{-\lambda b_0})^{1/2(1-\theta)}
 \,,
\eeaa
which can be rewritten as
$$
  e^{-\lambda (v-b_0)/2} \|\Psi_{i+1}(u,v)\|_{H^{m}(Y)}
  \leq
  c_m C^{\theta}(\overline  C_0 )^{1/2(1-\theta)}(e^{\lambda b_0})^{\theta/2}
 \,.
$$
For $m=k-4>\frac{n-1}{2}$ (which is possible if $k >\frac{n+7}{2}$),
from the Sobolev's embedding theorem there exists a constant $C_S>0$
such that
$$
e^{-\lambda (v-b_0)/2} \|\Psi_{i+1}(u,v)\|_{L^{\infty}(Y)}
  \leq
  C_S C^{\theta}(\overline  C_0 )^{1/2(1-\theta)}(e^{\lambda b_0})^{\theta/2}
 \,.
$$
Finally, we choose $\overline  C_0$ small enough so that
\bel{21X13.2}
  C_S C^{\theta}(\overline  C_0 )^{1/2(1-\theta)}(e^{\lambda b_0})^{\theta/2}<\sup_i\sup_{ v \in[0,b_0] } | \partial_u\psi_i (0,v)  |+1
  \,,
\ee
and obtain that
$$
\|\Psi_{i+1}(u,v)\|_{L^{\infty}(Y)}\leq e^{ \lambda
(v-b_0)/2}(\sup_i\sup_{ v \in[0,b_0] } | \partial_u\psi_i (0,v)
|+1) \,,
$$
which leads to
\bel{21X13.3} \|\partial_u\psi_{i+1}(u,v)\|_{L^{\infty}(Y)}\leq
2\sup_i\sup_{ v \in[0,b_0] } | \partial_u\psi_i (0,v)  |+1 \,, \ee
for all $v\in[0,b_0]$ and all $u$ in the range of \eq{13X13-1}, with
$\overline  C_0$ defined in \eq{21X13.2}. Thus we conclude, as after
\eq{14IX13.4}, that after reducing $a_*$ if necessary,
$$
\sup_i\sup_{(u,v)\in[0,a_*]\times[0,b_0]} | \partial_u\psi_i (u,v)
|\leq 2 \sup_i\sup_{ v \in[0,b_0] } | \partial_u\psi_i (0,v)  | +1
 \,.
$$
The estimate $ (|\nabla_\mu A^\mu|)_i\le \Cdiv$ for all $i$ follows
when $A^v$ does not depend upon $\varphi$.

When $A^v$ depends upon $\varphi$ it remains to obtain pointwise
estimate on $\partial_v\varphi$, we start by $v$-differentiating the
equation satisfied by $\varphi$:
\bean
  (A^\mu_{\varphi\varphi}\nabla_\mu)_i {\partial \varphi_{i+1} \over\partial
  v} &=& -\partial_v\Big((A^\mu_{\varphi\varphi}\nabla_\mu)_i \Big){ \varphi_{i+1}}
  - \partial_v\Big((A^\mu_{\varphi\psi}\nabla_\mu)_i \psi_{i+1}
  -(\newg_\varphi)_i\Big)
  \\ &=:& (B_\varphi)_i {\partial \varphi_{i} \over\partial
  v} +(b_\varphi)_i\,, \label{st10}
\eea
where
$$
(B_\varphi)_i %(\cdot)
:= -\partial_\varphi
 \Big((A^\mu_{\varphi\varphi}\nabla_\mu)_i \Big)
 %(\cdot)
 {
 \varphi_{i+1}}
  - \partial_\varphi\Big((A^\mu_{\varphi\psi}\nabla_\mu)_i \psi_{i+1}
  -(\newg_\varphi)_i\Big)
  %(\cdot)
  \,,
$$
  and with $(b_\varphi)_i$ containing all the remaining terms.
  After replacing  $v$-derivatives of $\psi_{i+1}$  using the field equations,
   $(B_\varphi)_i $ and $(b_\varphi)_i$ become affine
   in $\znabla_B
  f_{i+1}$ and $\znabla_B\znabla_C
  f_{i+1}$, with coefficients depending upon $f_i$.

Recall that $k_2$ has been defined   in \eq{21X13.1}; for $k\ge k_2$
by
  \eq{Moser1} and \eq{Moser3} we have the estimate
\bel{st11}
 \ykmtn{(B_\varphi)_i}+\ykmthn{(b_\varphi)_i}
 +\ywti{(B_\varphi)_i}+\ywi{(b_\varphi)_i} \le C_9\;
\ee
Applying \eq{volest2} with $k$ replaced by $k-3$,
 with $\mcU=[0,u]\times
\{v\}$, $f={\partial \varphi_{i+1} \over\partial
  v}$,  etc., to   \eq{st10},   we obtain
\bean
\lefteqn{
 e^{-\lambda(u+v)}\|\frac{\partial\varphi_{i+1}}{\partial v}(u,v)\|^2_{H^{k-3}(Y)}\leq C_{10}
  \Big\{
 e^{-\lambda v}\underbrace{\|\frac{\partial\varphi_{i+1}}{\partial v}(0,v)\|^2_{H^{k-3}(Y)}}_{\|\frac{\partial\overline {\varphi}_{i+1}}{\partial v}(v)\|^2_{H^{k-3}(Y)}}  +
    }
 &&
 \\
 \nonumber
 &&
 \int_0^u e^{-\lambda(s+v)}\Big[\left(\|(\nabla_\mu
 A^\mu_{\varphi\varphi})_i(s,v)\|_{L^\infty(Y)}
 -c\lambda\right)\|{\partial \varphi_{i+1} \over\partial v}(s,v)\|^2_{H^{k-3}(Y)} +
 \nonumber
 \\
 &&
 C_{11}\| {\partial \varphi_{i+1} \over\partial v}(s,v)\|_{H^{k-3}(Y)}\Big( \|{\partial \varphi_{i} \over\partial v}(s,v) \|_{H^{k-3}(Y)}+\|{\partial \varphi_{i+1} \over\partial
  v}(s,v) \|_{H^{k-3}(Y)}+ C_{12}\Big)\Big]ds
   \Big\}
  \,.
 \nonumber
 \\
 &&
 \eeal{alp1estn}
As before, using the inequality $ab\le  a^2/(4\epsilon)+\epsilon
b^2$, one is led to
\bean \lefteqn{
 e^{ -\lambda u}\|\frac{\partial\varphi_{i+1}}{\partial v}(u,v)\|^2_{H^{k-3}(Y)}\leq C_{10}\Big\{
 \|\frac{\partial\overline {\varphi}_{i+1}}{\partial v}(v)\|^2_{H^{k-3}(Y)}  +
 }
 &&
 \\
 \nonumber
 &&
 \int_0^u e^{-\lambda s}\Big\{\left(\|(\nabla_\mu A^\mu_{\varphi\varphi})_i(s,v)\|_{L^\infty(Y)}
 +2C_{11}+\frac{C_{11}}{4\epsilon}-c\lambda\right)\|{\partial \varphi_{i+1} \over\partial v}(s,v)\|^2_{H^{k-3}(Y)}
 \nonumber
 \\
 &&
 \phantom{xxxxxxxxxxxxxxxxxxxxx}
 +\epsilon C_{11}\| {\partial \varphi_{i} \over\partial v}(s,v)\|^2_{H^{k-3}(Y)}+C_{11} C^2_{12}\Big\}
 \,.
 \qquad
 \eeal{alp1estn-a}
Since (see \eq{bid9f})  $$ |(\nabla_\mu
A^\mu_{\varphi\varphi})_i|\leq |(\nabla_\mu A^\mu)_i|\leq \Cdiv,
\quad \forall (u,v)\in [0, a_i]\times[0,b_0]\,,
$$
there exists a constant $\lambda_3=\lambda_3(C_{10},\Cdiv, C_0,k)$
which does not depend on $i$ such that, for all $\lambda\geq
\lambda_3$, the previous inequality implies
\bean
 \lefteqn{ e^{ -\lambda u}\|\frac{\partial\varphi_{i+1}}{\partial v}(u,v)\|^2_{H^{k-3}(Y)}\leq C_{10}
\Bigg\{\|\frac{\partial\overline {\varphi}_{i+1}}{\partial
v}\|^2_{H^{k-3}(Y)}
 \phantom{xxxxx}
 }
 &&
\\
 &&
 \phantom{xxxxx}
 +  C_{11}\int_0^u e^{-\lambda s}\Big\{
 \epsilon \| {\partial \varphi_{i} \over\partial
  v}(s,v) \|^2_{H^{k-3}(Y)}+ C^2_{12}\Big\}ds\Bigg\}\,.
    \label{alp1estn-b}
 \eea
 Integrating in $u$, for $0\leq u\leq \hat{a}_i\leq a_0,$ one obtains
\bea
 \label{alp1estn-c}
 \lefteqn{
  \int_0^ue^{ -\lambda t}\|\frac{\partial\varphi_{i+1}}{\partial v}(t,v)\|^2_{H^{k-3}(Y)}dt\leq C_{10}
  \Bigg\{a_0\|\frac{\partial\overline {\varphi}_{i+1}}{\partial v}(v)\|^2_{H^{k-3}(Y)}
   }
 &&
 \nonumber
 \\
 &&
 +  C_{11}\int_0^u\int_0^t e^{-\lambda s}\Big\{\epsilon \| {\partial \varphi_{i} \over\partial v}\|^2_{H^{k-3}(Y)}+ C^2_{12}\Big\}dsdt\Bigg\}
 \,.
\eea
 Let
\bel{newconst}
 \zC_{\varphi}(u):=C_{10}\Big\{\sup_{i\in\N}\sup_{v\in[0,b_0]}a_0\|\frac{\partial\overline {\varphi}_{i+1}}{\partial v}(v)\|^2_{H^{k-3}(Y)} + \int_0^u\int_0^t C_{11 }C^2_{12}dsdt\Big)\Big\}
 \,.
\ee
Proceeding as before one gets rid of the ${\partial \varphi_{i}
\over\partial v}$ terms in the integral appearing in
\eq{alp1estn-b}, for all $0\le u \le  \hat{a}_i$ as follows: suppose
that
\bel{eeqq1}
 \forall\; 0\leq t\leq u\leq \hat{a}_i\leq a_0, \;\int_0^t e^{-\lambda s}\| {\partial \varphi_{i} \over\partial v}(s,v)\|^2_{H^{k-3}(Y)}ds\leq 2\zC_{\varphi}(t)
 \,,
\ee
then, equation \eq{alp1estn-c} gives
$$
 \int_0^ue^{ -\lambda t}\|\frac{\partial\varphi_{i+1}}{\partial v}(t,v)\|^2_{H^{k-3}(Y)}dt\leq \zC_{\varphi}(u)+  2a_0\epsilon C_{10}C_{11}\zC_{\varphi}(u)
 \,.
$$
Thus, one can choose  $\epsilon=
\epsilon(C_{10},C_{11},C_{12},\Cdiv, C_0,k, \lambda_3)$ small enough
so that
$$
\int_0^ue^{ -\lambda t}\|\frac{\partial\varphi_{i+1}}{\partial
v}(t,v)\|^2_{H^{k-3}(Y)}dt\leq 2\zC_{\varphi}(u) \,,
$$
which shows that \eq{eeqq1} is preserved under iteration.

For any $\lambda \ge \lambda_3|_{k=k_2}$ we obtain from
\eq{alp1estn-b} that:
$$
 \|\frac{\partial\varphi_{i+1}}{\partial v}(u,v)\|^2_{H^{k_2-3}(Y)}\leq C
 \,.
$$
Now, Sobolev's embedding implies
$$
 \ywi{{\partial \varphi_{i+1} \over\partial v}(u,v)}\le C
 \,.
$$
As this holds for all $i$, \eq{st10} proves that
\bel{goody}
 \yli{{\partial^2 \varphi_{i+1} \over\partial u \partial v}(u,v)}\le C
 \,.
\ee
By integration
$$
 \left|{\partial \varphi_{i+1} \over\partial v}(u,v)\right|\leq\left|{\partial \varphi_{i+1} \over\partial v}(0,v)\right|+Cu  \le 2\sup_{i\in\N}\ylin{{\partial \overline {\varphi}_i \over\partial v}}
 \,,
$$
provided that
\bel{third-condition}
 0\leq u \leq C^{-1}
  \big(\sup_{i\in\N}\ylin{{\partial \overline {\varphi}_i \over\partial v}}
   \big)
    \,.
\ee

Now, we choose $a_*$ to be the smallest of $a_0$ and of the four
constants appearing in the right-hand-side of inequalities
\eq{14IX13.2}, \eq{14IX13.3}, \eq{13X13-1} and \eq{third-condition}.
Recall that $ a_i$ was defined as either $a_0$  or the first number
at which the inequalities \eq{bid9m} fail   for $f_i$ or $f_{i+1}$.
So, if we assume that the inequalities   \eq{bid9m} hold at the
induction step $i$ with $a_i\ge a_*$, we conclude that $a_{i+1}\ge
a_*$ as well. Hence $\hat a_i \ge a_*\,, \forall i\in\N$. The above
implies that \eq{intpsic1} and \eq{intphic1} hold for $0\le u \le
a_*$. Since $a_*$ is independent of $k$, we have obtained:
\begin{Proposition}
% \ptc{pour rtw: verifier qu'on n'a pas besoin d'ajouter de derivees dans $C_0$
% \\ ...\\
%6XI13: Je l'ai fait, en fait les équations dérivées étant linéaires, on utlise les estimations déjà obtenues pour éviter de mettre les normes infinies %des nouvelles fonctions dans la constante $C_1$ de l'energie}
% \ptc{les derivees se calculent a partir des donnees initiales et des equations de transport; donc la condition
% sur les derivees premieres n'est pas triviale, faut le dire quelque part, mais attention a ne pas etre repetitif}
\label{Piter2}
 Let $\N\ni k> (n+7)/2$, and suppose that there exists
a constant $\mcC$ such that for $(u,v)\in [0,a_0]\times [0,b_0]$ we
have
\begin{align}
 \label{Itcond23X13}
  \sup_{\mcN^-\cup\mcN^+}\Big\{ \left|\overline {\partial_vf}_i\right|& + \left|\overline {\partial_u f}_i\right| +\ykn{\overline {f}_i(u,v)}+ M_k(u,v)
   %\ykn{\mathring
%   \newg(u,v)} +  \|\mathring{\gamma}  (u,v)\|^2_{H^{k}(Y)} \\
%   &\|\mathring{\tilde{\newg}}  (u,v)\|^2_{H^{k}(Y)}+ \|\zA (u,v)\|^2_{H^{k}(Y)}+ \|\mathring{\tilde{A}} (u,v)\|^2_{H^{k}(Y)}
%    +  \|\zGamma  (u,v)\|^2_{H^{k}(Y)}
 \Big\}
 \le
  \mcC
 \,.
\end{align}
There exists a constant $0<a_*=a_*(a_0,b_0,\mcC,Y)\le a_0$ such that
the fields $f_i$ satisfy the hypotheses of Lemma~\ref{Lenergy} on
$[0,a_*]\times[0,b_0]\times Y$. As a consequence, there exists a
constant $C=C(a_0,b_0,\mcC,Y,k)$
 such that for $(u,v)\in  [0,a_*]\times[0,b_0]$ we have
\bean
 \lefteqn{
 \int_0^{a_*}  \|\psi_i(s,v)\|_{H^{k }(Y)}^2 ds +
  \int_0^{b_0}  \|\varphi_i(u,s)\|_{H^{k }(Y)}^2 ds +  \|f_i(u,v)\|_{H^{k-1}(Y)}
    }
    &&
\\
&& \phantom{xxxx xxxx}
   + \|\partial_v \psi_i(u,v)\|_{H^{k-2}(Y)}  +
  \|\partial_u \varphi_i(u,v)\|_{H^{k-2}(Y)}
   \nonumber
\\
&& \phantom{xxxx xxxx}
   + \|\partial_u \psi_i(u,v)\|_{H^{k-3}(Y)}  +
  \|\partial_v \varphi_i(u,v)\|_{H^{k-3}(Y)}  \le C
 \,.
\phantom{xx xxx} \eeal{23X13.p5}
\end{Proposition}

\begin{Remark}{\rm
  \label{R1III14.1}
  The result remains true for $k\in \R$; this can be established by commuting the equation with an appropriate
  pseudo-differential operator in the $Y$-variables. However, this will be of no concern to us here.
}\qed\end{Remark}

\section{Convergence of the iterative sequence}\label{Convergence
iterative sequence}

% \ptc{section checked 24 X 13}
To prove convergence of the sequence, we set $$\delta f_{i+1}:=
f_{i+1}-f_i\,. $$  We have the equation \bel{cis1} (A^\mu
\nabla_\mu)_i\delta f_{i+1} = \delta \newg_{i}\,, \ee with
$$\delta \newg_{i}:= \newg_{i} -\newg_{i-1}-\left((A^\mu \nabla_\mu)_{i}-(A^\mu
\nabla_\mu)_{i-1}\right)f_i\,. $$
The standard identity
$$
 h(x)-h(y)=
 (x-y)\int_0^1h'(tx+(1-t)y)dt
 \,,
$$
applied both to $\newg_{i}-\newg_{i-1}$, and   $(A^\mu
\nabla_\mu)_{i}-(A^\mu \nabla_\mu)_{i-1}$, leads to the
straightforward estimate, for all $\lambda$ and $0\leq a\leq a_*$,
$$
    \|e^{-\lambda(u+v)}\delta \newg_{i}\|_{L^2([0,a]\times
    [0,b_0]\times Y)}\le C_1 \|e^{-\lambda(u+v)}\delta
    f_{i}\|_{L^2([0,a]\times [0,b_0]\times Y)}
    \,,
$$
with a constant $C_1$ which depends upon
$\sup_i\|f_i\|_{W^{1,\infty}}$, and which is independent of
$\lambda$ and of $i$. Here we reset the numbering of the constants,
so that the constant $C_1$ of this section has nothing to do with
the constant $C_1$ of the previous section, etc.

We apply the energy inequality \eq{bid5} with $k=0$; there are  then
no commutator terms  in \eq{volest1}, leading to
\bean \lefteqn{\|e^{-\lambda (u+v)}\delta
\varphi_{i+1}(u)\|_{L^2([0,b_0]\times Y)} + \|e^{-\lambda
(u+v)}\delta \psi_{i+1}(v)\|_{L^2([0,a_*]\times Y)}} && \\
 && \le C_2\bigg\{\|e^{-\lambda v}\delta
\overline {\varphi}_{i+1}\|_{L^2([0,b_0]\times Y)} + \|e^{-\lambda
u}\delta\overline {\psi}_{i+1}\|_{L^2([0,a_*]\times Y)}\nonumber\\
&&+
\left(\|(A^\mu\nabla_\mu)_i\|_{L^\infty}
-c\lambda\right)\|e^{-\lambda(u+v)}\delta
f_{i+1}\|^2_{L^2([0,a_*]\times [0,b_0]\times Y)} \nonumber
\\
&& %\label{volest1a}\\&&%
\phantom{xxx} +2\|e^{-\lambda(u+v)}\delta
f_{i+1}\|_{L^2([0,a_*]\times [0,b_0]\times Y)}
\|e^{-\lambda(u+v)}\delta \newg_{i}\|_{L^2([0,a_*]\times
[0,b_0]\times Y)} \bigg\}
    \nonumber
\\
 && \le C_2\bigg\{\|e^{-\lambda v}\delta
\overline {\varphi}_{i+1}\|_{L^2([0,b_0]\times Y)} + \|e^{-\lambda
u}\delta\overline {\psi}_{i+1}\|_{L^2([0,a_*]\times Y)} \nonumber\\
&&+
 \left(\|(A^\mu\nabla_\mu)_i\|_{L^\infty}
+C_1 -c\lambda\right)\|e^{-\lambda(u+v)}\delta
f_{i+1}\|^2_{L^2([0,a_*]\times [0,b_0]\times Y)}\bigg\}
\nonumber \\
&& %\label{volest1a}\\&&%
\phantom{xxx} +C_1C_2 \|e^{-\lambda(u+v)}\delta
f_{i}\|^2_{L^2([0,a_*]\times [0,b_0]\times Y)} \,.
%\Big]
\eeal{cis2}
%

%For our further purposes it suffices to consider sequences
Now, for the purpose of the proof of Theorem~\ref{T23X13.2}, the
sequences   $(\overline {\varphi}_{i})_{i\in\N}$ and  $(\overline
{\psi}_{i})_{i\in\N}$ are Cauchy sequences in the spaces
$H^k([0,b_0]\times Y)$ and $H^k([0,a_0]\times Y)$ respectively and
thus in $L^2([0,b_0]\times Y)$ and $L^2([0,a_0]\times Y)$.
  Therefore, without loss of generality they can be replaced by subsequences,
  still denoted as  $(\overline {\varphi}_{i})_{i\in\N}$ and  $(\overline {\psi}_{i})_{i\in\N}$,
%     exists a subsequence $(\overline {\varphi}_{i_j})_{j\in \N}$
%  of   $(\overline {\varphi}_{i })_{i\in \N}$ and  a subsequence $(\overline {\psi}_{i_j})_{j\in \N}$
%  of   $(\overline {\psi}_{i })_{i\in \N}$
such that
\bel{24X13.p5}
 C_2 \|\delta \overline {\varphi}_{i }\|_{L^2([0,b_0]\times Y)}\leq \frac{1}{2^{i+1}}  \quad \text{and} \quad C_2\|\delta \overline {\psi}_{i }\|_{L^2([0,a_0]\times Y)}\leq \frac{1}{2^{i+1}} \,.
\ee
Assuming that \eq{24X13.p5} holds, we have:
\bean \lefteqn{\|e^{-\lambda (u+v)}\delta
\varphi_{i+1}(u)\|_{L^2([0,b_0]\times Y)} + \|e^{-\lambda
(u+v)}\delta \psi_{i+1}(v)\|_{L^2([0,a_*]\times Y)}} && \\
 && \le   \frac{1}{2^{i}}+C_2
 \left(\|(A^\mu\nabla_\mu)_i\|_{L^\infty}
+C_1 -c\lambda\right)\|e^{-\lambda(u+v)}\delta
f_{i+1}\|^2_{L^2([0,a_*]\times [0,b_0]\times Y)}
\nonumber \\
&& %\label{volest1a}\\&&%
\phantom{xxx} +C_1C_2 \|e^{-\lambda(u+v)}\delta
f_{i}\|^2_{L^2([0,a_*]\times [0,b_0]\times
Y)} \,.%\Big]
\eeal{cis2??}
In particular, given any $0<\alpha<1/2$, for all $\lambda $
sufficiently large  and for all $(u,v)\in [0,a_*]\times[0,b_0]$, we
find
\beqar
 \|e^{-\lambda (u+v)}\delta \varphi_{i+1}(u)\|^2_{L^2([0,b_0]\times Y)}
 & \le&
 \frac{1}{2^{i}}+ C_1C_2 \|e^{-\lambda(u+v)}\delta f_{i}\|^2_{L^2([0,a_*]\times[0,b_0]\times Y)}
\,,
 \nonumber
 \\ \label{cis3}
\\
 \|e^{-\lambda (u+v)}\delta \psi_{i+1}(v)\|^2_{L^2([0,a_*]\times Y)} &\le&\frac{1}{2^{i}}+ C_1 C_2\|e^{-\lambda(u+v)}\delta
 f_{i}\|^2_{L^2([0,a_*]\times [0,b_0]\times Y)}
\,,
 \nonumber
 \\
 \label{cis3.0}
\\
 \|e^{-\lambda(u+v)}\delta f_{i+1}\|^2_{L^2([0,a_*]\times
 [0,b_0]\times Y)} &\le& \frac{1}{C_2\cdot2^{i}}+\alpha \|e^{-\lambda(u+v)}\delta
 f_{i}\|^2_{L^2([0,a_*]\times [0,b_0]\times Y)}
  \,;
 \nonumber
 \\
\label{cis4} \eeqar
Here $\lambda $ has to be chosen so that
\bel{cis6}
 0<\frac {C_1}{c\lambda-\|(A^\mu\nabla_\mu)_i\|_{L^\infty} -C_1}< \alpha <\frac 12 \,.
\ee
We can now make use of the elementary fact: If $(U_n)_{n\in \N}$ is
a sequence of positive real numbers satisfying $
 U_{n+1} \leq \alpha U_n+ \frac{\beta}{2^n}
 \,,
$
then
\bel{24X13.p6}
 U_n \leq \al^n U_0 +
 2\beta\left(\frac{(1/2)^n-(\alpha )^n}{1-2\alpha }\right)
 \,.
\ee
\Eqs{cis4}{24X13.p6} show that
$$
 \mbox{$\sum e^{-\lambda(u+v)}\delta f_i$ converges in
    $L^2([0,a_*]\times [0,b_0]\times Y)$.}
$$
This implies that $f_i$ converges in the same space to some function
$f$.  It further follows from \eq{cis3} that for all $0\le u\le a_*$
the sum $\sum_i e^{-\lambda (u+v)}\delta \varphi_{i}(u)$ converges
in ${L^2([0,b_0]\times Y)}$, uniformly in $u$; this implies  uniform
convergence of $\varphi_i(u)$ to some function $\varphi(u)$ in that
topology. Similarly for all $v\in [0,b_0]$ the sequence $\psi_i(v)$
converges, uniformly in $v$, to some function $\psi(v)$ in
 ${L^2([0,a_*]\times Y)}$.

For $k>(n+7)/2$ the estimates of the previous section apply and show
that the sequence of derivatives $\nabla f_i$
 %, $\partial_v
% f_i$ and $\partial_u \varphi_i$ are
is uniformly bounded so that, by Arzela-Ascoli, a subsequence
$f_{i_j}$
%(u,v):=f_{i_j}(u,v,\cdot)$
can be chosen which converges uniformly to some function which is
Lipschitz continuous in all variables
%$x^A$
on $[0,a_*]\times[0,b_0]\times Y$. It follows that $f$ has a
Lipschitz continuous representative, this representative will be
chosen from now on. Similarly, $f_{i_j+1}$ has a subsequence, still
denoted by the same symbol, uniformly converging to some Lipschitz
continuous function $f'$. Since $f_{i_j+1}$ converges to $f$ in
$L^2$ we must have $f' =f$, thus $f_{i_j+1}$ converges uniformly to
$f$.

Now, by Proposition~\ref{Piter2} the sequence $\fij(u,v)$ is
bounded in $H^{k-1}(Y)$, and converges uniformly to the continuous
function $f(u,v)$. By weak compactness
$$
 f(u,v) \equiv \big(\varphi(u,v ),\psi(u,v )\big)\equiv \big(\varphi(u,v,\cdot),\psi(u,v,\cdot)\big) \in H^{k-1}( Y)
 \,.
$$
By interpolation, for every $s<k-1$ we have
\bel{cis8}
 f_{i_j}(u,v)
 \,,
 \
  f_{i_j+1}(u,v)\to f (u,v)\ \mbox{ in } \ H^s( Y)
 \,,
\ee
uniformly in $u$ and $v$.
In particular
\bel{cis8.2}
 f_{i_j}(u,v)\,,
 \
  f_{i_j+1}(u,v)\to f (u,v)\ \mbox{ in } \ C^1( Y)
 \,,
\ee
uniformly in $u$ and $v$. Thus both $\varphi$ and $\psi$ are
differentiable with respect to the $x^A$'s.

Using the notation of Section~\ref{Sec:enid}, \eq{itpro} now shows
that the sequence $\partial_u \varphi_{i_j+1}(u,v)$ converges
uniformly to the Lipschitz-continuous function
$$
 (*):=
 (A^u_{\varphi\varphi})^{-1}\Big[-A^B_{\varphi\varphi}\nabla_B\varphi
 -A^B_{\varphi\psi}\nabla_B\psi +\newg_\varphi \Big] -
 \gamma_{\varphi\varphi,u}\varphi
 \,.
$$
Similarly $\partial_v \psi_{i_j+1}(u,v)$ converges uniformly to a
Lipschitz-continuous function, as determined by the right-hand-side
of the equation involving $\partial_v \psi$. From
\bel{1III14.1}
 \underbrace{\varphi_{i_j+1}(u_2,\cdot)}_{\to
\varphi(u_2,\cdot)}-\underbrace{\varphi_{i_j+1}(u_1,\cdot)}_{\to
\varphi(u_1,\cdot)}=\int_{u_1}^{u_2}\underbrace{\partial_u
\varphi_{i_j+1}(s,\cdot)}_{\to (*)} ds \ee
one finds that $\varphi$ is differentiable in $u$. Similarly $\psi$
is differentiable in $v$, and \eq{fos} holds.

From what has been said we have
\beal{24X13.p10}
 &
f\in L^\infty\Big([0,a_*]\times [0,b_0]; H^{k-1}(Y)\Big) \,,
 &
\\
 &
    \partial_A f\,, \ \partial_u\varphi\,, \ \partial_v \psi \ \in
    L^\infty\Big([0,a_*]\times [0,b_0]; H^{k-2}(Y)\Big)
    \,,
\\
 &
     \partial_v\varphi\,, \ \partial_u \psi \ \in
    L^\infty\Big([0,a_*]\times [0,b_0]; H^{k-3}(Y)\Big)
    \,.
 &
\eeal{24X13.p10x}
Thus
\bel{24X13.p12}
   f\in\cap _{0\le i \le 1}
    W^{i,\infty}\Big([0,a_*]\times [0,b_0]; H^{k-2-i}(Y)\Big)\subset
    C^{0,1}([0,a_*]\times [0,b_0]\times Y)
    \,.
\ee

We note that the new field
\bel{2III14.1}
 f'= \left(
             \begin{array}{c}
               \varphi' \\
               \psi'
             \end{array}
           \right)
           \,,
           \
           \mbox{where}
           \
 \varphi'= \left(
             \begin{array}{c}
               \varphi \\
              \partial_v \varphi \\
              \partial_u  \varphi \\
               \partial_A\varphi
             \end{array}
           \right)
           \
           \mbox{and}
           \
 \psi'= \left(
             \begin{array}{c}
               \psi \\
               \partial_v \psi \\
               \partial_u \psi \\
               \partial_A \psi
             \end{array}
           \right)
\ee
is defined on $[0,a_*]\times [0,b_0]\times Y$ and solves a system
of equations satisfying our structure conditions. By what has been
said the initial data are of $H^{k-3}$ differentiability class. So
if $k-3>(n+7)/2$, the argument leading to \eq{24X13.p12} applies to
$f'$ and gives
\bean
   f&\in&
L^\infty\Big([0,a_*]\times [0,b_0]; H^{k-1}(Y)\Big) \cap _{0< i \le
2}
    W^{i,\infty}\Big([0,a_*]\times [0,b_0]; H^{k-3i}(Y)\Big)
\\
&&
     \subset
     C^{1,1}([0,a_*]\times [0,b_0]\times Y)
    \,.
\eeal{24X13.p12x}
This argument can be applied $k_1$ times, where
\bel{2III14.2}
 \mbox{$k_1$ is the largest number such that $k- 3 k_1>(n+7)/2$.}
\ee
It ensues that
\bean
   f&\in&
L^\infty\Big([0,a_*]\times [0,b_0]; H^{k-1}(Y)\Big) \cap _{0< 3i \le
k - \frac{n+7}2}
    W^{i,\infty}\Big([0,a_*]\times [0,b_0]; H^{k-3i}(Y)\Big)
\\
&&
     \subset
     C^{k_1-1,1}([0,a_*]\times [0,b_0]\times Y)
    \,,
\eeal{24X13.p13}
where the last inclusion holds provided that $k_1\ge 1$.

\begin{Remark}{\rm
 \label{R8V14.1}
For $k>6+(n+7)/2$ the first line of \eq{24X13.p13} can be partly
improved to
\bea
   f&\in &
C\Big([0,a_*]\times [0,b_0]; H^{k-1}(Y)\Big) \cap _{0< 3i \le k -
\frac{n+7}2 -6}
    C^{i }\Big([0,a_*]\times [0,b_0]; H^{k-3i}(Y)\Big)%
     \,.
\nonumber
\\
%&&
%     \subset
%     C^{k_1 }([0,a_*]\times [0,b_0]\times Y)
\eeal{24X13.p13a}
To see this, note first that the map
\bel{1III14.2}
 (u,v)\mapsto \partial_u^i \partial_v^jf(u,v,\cdot) \in H^{k-3(i+j)}(Y)
\ee
is weakly continuous, being the limit of a bounded sequence of
continuous maps. Using the equation satisfied by $f$ and the trivial
identities
\beaa &
 \partial_u^i \partial_v^j\varphi(u,v)=\partial_u^i \partial_v^j\varphi(0,v)+\int_0^u \big( \partial_u \partial_u^i \partial_v^j\varphi(s,0) + \int_0^v \partial_u \partial_v\partial_u^i \partial_v^j\varphi(s,t)dt\big)ds
 \,,
 &
\\
 &
 \partial_u^i \partial_v^j\psi(u,v) = \partial_u^i \partial_v^j\psi(u,0) + \int_0^v \big(\partial_v \partial_u^i \partial_v^j\psi (0,t) + \int_0^u \partial_u \partial_v \partial_u^i \partial_v^j\psi(s,t) ds\big) dt
 \,,
 &
\eeaa
one sees  that the function
$$
 (u,v)\mapsto \|\partial_u^i \partial_v^jf(u,v,\cdot)\|_{  H^{k-3(i+j)}(Y)}
$$
is continuous. This, together with standard arguments, implies that
the map \eq{1III14.2} is continuous, and \eq{24X13.p13a} easily
follows.
}\qed\end{Remark}

\section{Existence and uniqueness}
 \label{sec:initial-data}
 \label{sub:Initial-data}
% \ptc{section checked 24 X 13}

In order to complete the proof of existence of a solution for the
system (\ref{fos}), we need to initialize the iteration  and make
sure that condition \eq{Itcond23X13} is fulfilled. Recall that in
the current setting
$$
 \mcN^-= \{0\}\times[0,b_0]\times Y\,,\quad
 \mcN^+= [0,a_0]\times\{0\}\times Y
 \,.
$$
We have the following:

\begin{Theorem}
 \label{T23X13.2}
 Let $Y$ be a $(n-1)-$dimensional compact manifold without boundary, let $a_0$ and $b_0$ two positive real numbers and set
$$
 \Omega_0=[0,a_0]\times[0,b_0]\times Y
$$
%.
Consider the symmetric hyperbolic system \eq{fos} on $\Omega_0$ with
the splitting \eq{bid2.5} and assume that \eq{bid3} holds. Let
$\overline  \varphi$ and $\overline  \psi$ be  defined respectively
on $\mcN^-$ and $\mcN^+$, providing Cauchy data for \eq{fos}:
\bel{ID}
 \left\{
  \begin{array}{l} \varphi= \overline \varphi \quad \text{on} \quad \mcN^-
  \\
 \psi= \overline \psi \quad \text{on} \quad \mcN^+
 \end{array}
 \right.
\,. \ee
Let $\ell\in\N$, $\ell> \frac{n+9}2$ and suppose that
\bel{23X13.1}
 \overline \varphi\in \cap_{0\leq j\leq \ell}C^j([0,b_0];H^{\ell-j}(Y))
 \quad
 \mbox{and}
 \quad
 \overline \psi\in \cap_{0\leq j\leq \ell}C^j([0,a_0];H^{\ell-j}(Y))
 \,.
\ee
Assume that the  \emph{transport equations}
\begin{eqnarray}
   &
   A^\mu_{\varphi\varphi}|_{v=0}\partial_\mu \varphi|_{v=0}=  \big(-A^{\mu}_{\varphi \psi}\partial_\mu \psi + G_\varphi\big)\big|_{v=0}
   \,,
   &
 \label{23X13.p1}
\\
&
  A^\mu_{\psi \psi}|_{u=0}\partial_\mu \psi|_{u=0}=  \big(-A^{\mu}_{ \psi\varphi}\partial_\mu \varphi + G_\psi\big)\big|_{u=0}
 \,,
 &
 \label{23X13.p2}
\end{eqnarray}
with initial data
$$
 \mbox{$\varphi|_{u=v=0} = \overline \varphi|_{v=0}$ and $\psi|_{u=v=0} = \overline \psi|_{u=0}$,}
$$
%,
have a global solution on $([0,a_0]\times Y)\cup ([0,b_0]\times Y)$.
Then there exists an $\ell$-independent constant $
  a_* \in(0,a_0]
$ such that the Cauchy problem \eq{fos}, \eq{ID} has a
 solution $f$  defined on
$
 [0,a_*]\times[0,b_0]\times Y
$
satisfying \eq{24X13.p10}-\eq{24X13.p10x} with $k=\ell-1$. If $\ell>\frac{n+12}2$ we further have
\bean
   f&\in&
L^\infty\Big([0,a_*]\times [0,b_0]; H^{\ell-2}(Y)\Big) \cap_{0< 3i
\le \ell - \frac{n+9}2}
    W^{i,\infty}\Big([0,a_*]\times [0,b_0]; H^{\ell-1-3i}(Y)\Big)
\\
&&
     \subset
     C^{\ell_1-1,1}([0,a_*]\times [0,b_0]\times Y)
    \,,
\eeal{4III14.t1}
where $\ell_1$ is the largest number such that
$\ell-3\ell_1>\frac{n+9}2$.
The solution $f$ is unique within the class of $C^1$ solutions, and
is smooth if  $\overline \varphi$ and $\overline \psi$ are.
\end{Theorem}

\begin{Remark}{\rm
 \label{R10V2014.1}{\rm
Some remarks about the hypothesis that $Y$ is \emph{compact without
boundary} are in order. First, our analysis applies to compact
manifolds with boundary without further due when suitable boundary
conditions are imposed on the boundary.
For instance, in case of systems obtained by rewriting the wave
equation as in Section~\ref{s14XI13.1}, Dirichlet, Neumann or
maximally dissipative boundary conditions  at $\partial Y$ are
suitable. Next, again for systems of wave equations, the case of
non-compact $Y$'s can be reduced to the compact one as follows: let
$p\in Y$, we replace $Y$ by a small conditionally compact
neighborhood of $p$ with smooth boundary. We solve the equation on
the new $Y$ imposing e.g. Dirichlet conditions on
$[0,a_0]\times[0,b_0]\times\partial Y$. Arguments based on
uniqueness in domains of dependence show that there is a one-sided
space-time neighborhood of the generators of $\mcN_\pm$ through $p$
on which the solution is independent of the boundary conditions
imposed. This provides the desired solution on the neighborhood.
Returning to the original $Y$, the union of such neighborhoods with
the corresponding solutions yields the desired solution.
}
}\qed\end{Remark}

\begin{proof}
Let $(\overline \varphi_i)_{i\in\N}$ and $(\overline
\psi_i)_{i\in\N}$ be any two sequences of smooth initial data which
converge towards $\overline \varphi$ and $\overline \psi$
respectively in the spaces
$$
\cap_{0\leq j\leq
\ell}C^j([0,b_0];H^{\ell-j}(Y))\quad \mbox{ and}\quad
 \cap_{0\leq j\leq \ell}C^j([0,a_0];H^{\ell-j}(Y))
 \,.
 $$
Set $f_{-1} \equiv 0$, and for  $i\in\N$ define
$
 \overline{f}_i=  (
                    \overline \varphi_i
                    ,
                    \overline \psi_i)
.$
Given $f_i$, we let $f_{i+1}$ to be the solution of the linear
system \eq{itpro} with Cauchy data
$$
 \left\{
  \begin{array}{l} \varphi_{i+1}= \overline \varphi_{i+1} \quad \text{on} \quad \mcN^-
  \\
 \psi_{i+1}= \overline \psi_{i+1} \quad \text{on} \quad \mcN^+
 \end{array}
 \right.
\,.
$$
We wish to apply    Proposition \ref{Piter2} with $k= \ell-1$, for
this we need to show that the constant $\mcC$ of \eq{Itcond23X13} is
finite.  We start by noting that the sequence
$
 (\overline \psi_i)_{i\in\N} %,(\partial_u\overline \psi_i)_{i\in\N},(\partial^2_u\overline \psi_i)_{i\in\N}
$ has been chosen to converge in the space $
 \cap_{0\leq j\leq 2}C^j([0,a_0];H^{\ell-j}(Y))
$, and since $ \ell>\frac{n+7}{2} $ the continuous embedding
$$
 \cap_{0\leq j\leq 2}C^j([0,a_0];H^{\ell-j}(Y)) \hookrightarrow \cap_{0\leq j\leq 2}C^j([0,a_0];W^{1,\infty}(Y))
$$
ensures that this convergence also holds in $ \cap_{0\leq j\leq
2}C^j([0,a_0];W^{1,\infty}(Y)) $. Since convergent sequences are
bounded, we obtain that
\bel{24X13.p1}
 \sup_{i\in \N,u\in [0,a_0]}
 \Big(
 \|\overline {\psi}_i(u)\|_{W^{1,\infty}(Y)}+\| \partial_u\overline \psi_i(u)\|_{W^{1,\infty}(Y)}+\|\partial_u^2\overline \psi_i(u)\|_{W^{1,\infty}(Y)}
 \Big)
 < \infty
\,. \ee
Similarly,
\bel{24X13.p2}
 \sup_{i\in \N,v\in [0,b_0]}
 \Big(
 \|\overline {\varphi}_i(v)\|_{W^{1,\infty}(Y)}+\| \partial_v\overline \varphi_i(v)\|_{W^{1,\infty}(Y)}+\|\partial_v^2\overline \varphi_i(v)\|_{W^{1,\infty}(Y)}
 \Big)
 < \infty
\,. \ee

By hypothesis, the transport equations with the initial data
$(\overline  \varphi, \overline  \psi)$ have global solutions on
$\mcN^\pm$. Continuous dependence of solutions of symmetric
hyperbolic systems upon data implies that the transport equations
with $(\overline  \varphi_i, \overline  \psi_i)$ will also have
global solutions on $\mcN^\pm$ for all $i$ large enough, bounded in
$C^1(\mcN)$ uniformly in $i$. We can thus use \eq{dzi11.6} at  $u=0$
to obtain, for all $i\in\N$ and all $\lambda$ sufficiently large,
\beaa
 \lefteqn{
    e^{-\lambda v}\|\psi_{i}(0,v)\|^2_{H^{\ell-1}(Y)}\le
    C_7(Y,\ell,C_0,\Cdiv)\Big\{\| {\psi}_{i}(0,0)\|^2_{H^{\ell-1}(Y)}}
    &&
\\
    && +2\hat C(0,b)+ \int_0^v e^{-\lambda s} M_\ell(0,s)ds
        +2\delta \tCpsi (0,b_0)+
        E_{\ell,\lambda}[\overline \varphi_{i},0]\Big\}
        \,.
        \phantom{xxxxxxx}
\eeaa
The right-hand-side is bounded uniformly in $i$ and $v\in  [0,b_0]$.
Thus there exists a constant, which we denote again by $C$, such
that
$$
\forall i\in\N,\; \forall v\in  [0,b_0],\quad \left\| \psi_{i} (0,v)
\right\|^2_{H^{\ell-1}(Y)}\leq C \,.
$$

We can repeat this process using the transport equation satisfied by
$\partial_u\psi_{i+1}(0,v)$, which is obtained by
$u-$differentiating the equation satisfied by $\psi_{i+1}$. This
leads to the inequality \eq{st3} at $u=0$ for every $i\in \N$ with
$k-3$ there replaced by $\ell-2$; the gain of one derivative here,
as compared to \eq{st3}, is due to the fact that $\varphi|_{u=0}$ is
directly given in terms of initial data, and hence is controlled in
$H^ \ell(Y)$, while in \eq{st3} we only had uniform control in
$H^{k-1}(Y)$. That is, for all $i\in\N$,
\beaa
 \left\|\partial_u\psi_{i} (0,v) \right\|^2_{H^{\ell-2}(Y)}
 \leq
   2e^{\lambda
 v}
  \Big[C_9\Big\{ \sup_{i\in\N}\sup_{u\in[0,a_0]}\|\partial_u\overline {\psi}_{i}(0)\|^2_{H^{\ell-2}(Y)}
 \\
 +\int_0^{b_0} e^{-\lambda s}\Big(  \hat M_{\ell}(0,s)+ C^2_{10}\Big)ds\Big\}
   +1
   \Big]
 \,.
\eeaa

An identical argument using \eq{alp1estn-b} gives the desired
control of $\varphi(u,0)$ and $\partial_v \varphi(u,0)$. This proves
that the left-hand side of \eq{Itcond23X13} is finite.

We can now appeal to  Section \ref{Convergence iterative sequence}
to conclude that the sequence $(f_i)_{i\in\N}$ converges towards a
solution $ f$ of the Cauchy problem  \eq{fos}, \eq{ID} in a space as
stated in the theorem. This requires choosing more carefully the
sequences $(\varphi_i)_{i\in\N}$, $(\psi_i)_{i\in\N}$, see
\eq{24X13.p5}, which is possible because the constant $C_2$
appearing in \eq{24X13.p5} is the same for all suitably bounded
sequences, possibly after taking $i\ge i_0$ for some $i_0$ large enough.

Note that the neighborhood of $\mcN^-$ on which the solution has
been constructed is independent of the Sobolev differentiability
class of the data. This implies that smooth initial data lead to
smooth solutions.

We continue with uniqueness of solutions. Let $f_\ell$, $ \ell=
1,2$, be two solutions of (\ref{fos}) with identical initial data
(\ref{ID}). Setting $\delta f= f_1-f_2$ leads to the  equation
\bel{eqdelta}
 (A^\mu\nabla_\mu)_1 \delta f= (\newg)_1-(\newg)_2-
 \Big((A^\mu\nabla_\mu)_1-(A^\mu \nabla_\mu)_2 \Big)f_2
  \,,
\ee
with  $\delta f$ vanishing on $\mcN$. The calculation now is similar
to that of Section \ref{Convergence iterative sequence}. Equation
(\ref{eqdelta}) can be rewritten as (\ref{cis1}) with $\delta
f_{i+1}$ and $\delta f_{i}$ there replaced by $\delta f \,, \;
 A^\mu(f_i)\nabla_\mu(f_i)$
  replaced by $A^\mu(f_1)\nabla_\mu(f_1)$
and $\delta \newg_{i}$ there replaced here by
$$
 \delta \newg:= (\newg)_1-(\newg)_2-\Big((A^\mu\nabla_\mu)_1-(A^\mu\nabla_\mu)_2\Big)f_2
 \,.
$$
The current equivalent of (\ref{cis2})  with $\overline{\delta
\varphi} = \overline{\delta \psi}\equiv 0$ there    reads
\bean
 \lefteqn{\|e^{-\lambda (u+v)}\delta
 \varphi(u)\|_{L^2([0,b_0]\times Y)} + \|e^{-\lambda (u+b)}\delta
 \psi(b)\|_{L^2([0,a_*]\times Y)}}
\\
 && \le
 \left(\|(A^\mu(f_1)\nabla_\mu(f_1))\|_{L^\infty}
 -c\lambda\right)\|e^{-\lambda(u+v)}\delta f\|^2_{L^2([0,a_*]\times
 [0,b_0]\times Y)}
\nonumber
\\
 \phantom{xxx} &&+C_1 \|e^{-\lambda(u+v)}\delta f\|^2_{L^2([0,a_*]\times[0,b_0]\times Y)}
\,.
 \eeal{cis22}
It then follows (compare with \eqref{cis4}) that there exists a real
number $\al \in (0,1)$ such that
$$
\|e^{-\lambda(u+v)}\delta f\|^2_{L^2([0,a]\times [0,b_0]\times Y)}
\le  \alpha \|e^{-\lambda(u+v)}\delta f\|^2_{L^2([0,a]\times
[0,b_0]\times Y)} \,.
$$
This means that $f_1= f_2$ almost everywhere on
$[0,a]\times[0,b_0]\times Y$, and since $f_1$ and $f_2$ are
continuous,   equality holds everywhere.
\end{proof}

\bigskip

The symmetry of the problem under the interchange of $u$ and $v$
shows that our construction also provides a solution in a
neighborhood of $\mcN^+$:

\begin{Corollary}\label{23X13.4}
  Under the hypotheses of Theorem \ref{T23X13.2}, there exists two constants $0< a_*\leq a_0$
  and $0< b_*\leq b_0$ and a unique solution $f$ of the Cauchy problem \eq{fos}, \eq{ID} defined on the neighborhood
  $$
  \Big([0,a_*]\times[0,b_0]\times Y\Big)\cup\Big([0,a_0]\times[0,b_*]\times Y\Big)
  $$
 of $\mcN = \mcN^+\cup\mcN^-$ such that
$$
 f\in L^\infty\Big([0,a_*]\times [0,b_0]; H^{\ell-2}(Y)\Big) \cap _{0< 3i \le \ell - \frac{n+9}2}
  W^{i,\infty}\Big([0,a_*]\times [0,b_0]; H^{\ell-1-3i}(Y)\Big)
$$
similarly on $[0,a_0]\times [0,b_*]$. %%
%$$
% f\in L^\infty\Big([0,a_0]\times [0,b_*]; H^{\ell-2}(Y)\Big) \cap _{0< 3i \le \ell - \frac{n+9}2}
%  W^{i,\infty}\Big([0,a_0]\times [0,b_*]; H^{\ell-1-3i}(Y)\Big)
%  \,.
%$$%
%
%Moreover, if $\ell>6+(n+9)/2$ then
%%
%$$
% f\in C\Big([0,a_*]\times [0,b_0]; H^{\ell-2}(Y)\Big) \cap _{0< 3i \le \ell - \frac{n+9}2}
%    C^{i }\Big([0,a_*]\times [0,b_0]; H^{\ell-1-3i}(Y)\Big)
%     \,,
%$$
%%
%with   similar properties on $[0,a_0]\times [0,b_*]$.
%%
%$$
% f\in \cap_{0\leq j\leq \ell-2}C^j([0,a_*]\times[0,b_0];H^{\ell-2-j}(Y))
%$$
%and
%%
%$$
% f\in \cap_{0\leq j\leq \ell-2}C^j([0,a_0]\times[0,b_*];H^{\ell-2-j}(Y))
%\,.
%$$
\end{Corollary}

\begin{Remark}{\rm
 \label{R29X13.11}
 Theorem \ref{T23X13.2} can be used to  obtain a solution of \eq{fos}, \eq{ID} when the transport equations can be solved globally on the hypersurfaces
$\wmcN^-=\{0\}\times [0, \infty) \times Y$ and $\wmcN^+=[0,
\infty)\times \{0\}\times Y$ as follows: Let $a_0$ and $b_0$ be two
arbitrary positive real numbers. Corollary \ref{23X13.4} shows that
there exist  two constants $0< a_*\leq a_0$ and $0< b_*\leq b_0$ and
a unique continuous solution $f$ of the Cauchy problem \eq{fos},
\eq{ID} defined on
$$
 \mcU_{a_0,b_0}:= \Big([0,a_*]\times[0,b_0]\times Y\Big)\cup\Big([0,a_0]\times[0,b_*]\times Y\Big)
 \,.
$$
Here $a_*$ and $b_*$ might depend upon $a_0$ and $b_0$. Uniqueness
of solutions on each $\mcU_{a_0,b_0} $ shows that solutions defined
on two such overlapping regions coincide on the overlap. This allows
one to define a solution on
$$
 \mcU= \underset{a_0, b_0\in\R_+}{\cup}\mcU_{a_0,b_0}
% \,,
$$
in an obvious way. We thus obtain a neighborhood of the entire
initial data hypersurface $\wmcN=\wmcN^-\cup\wmcN^+$. Note that the
thickness of the neighborhood might shrink to zero when receding to
infinity along  $\wmcN$ .
}\qed\end{Remark}

\section{Continuous dependence upon data}
 \label{ss26X13.1}

The aim of this section is to prove that the solutions obtained in
Theorem \ref{T23X13.2} are stable under small perturbations of the
Cauchy data. More precisely:

\begin{Theorem}
 \label{T28X13.2}
  Let $f$ be a solution of \eq{fos}  on $[0,a_0]\times[0,b_0]\times Y$, and
let $(f_i)_{i\in\N}$ be a sequence of solutions on
$[0,a_0]\times[0,b_0]\times Y$ such that the sequence of the
associated initial data $(\overline f_i)_{i\in\N}$   converges to
$\overline f$ in the topology determined by \eq{23X13.1} with
$\ell\ge \frac{n+15}2$.
 Then

\begin{enumerate}
 \item There exists $0<a_* \le a_0$ so that
 $$
  \mbox{the sequence $f _i$  is bounded in $C^{1,1}([0,a_*]\times [0,b_0]\times Y)$.
  }
 $$
 \item Suppose that $0<a \le a_0$ is such that the sequence $(f_i)_{i\in\N}$  is bounded in $C^{1,1}([0,a ]\times [0,b]\times Y)$, then  for any $0<s<\ell-(n+9)/2$ the sequence $(f_i)_{i\in\N}$ converges to $f$ in the topology of
\bel{28X13.21}
 L^\infty\Big([0,a_*]\times [0,b_0]; H^{\ell-2}(Y)\Big) \cap _{0< 3i \le s}  W^{i,\infty }\Big([0,a_*]\times [0,b_0]; H^{\ell-1-3i}(Y)\Big)
 \,.
\ee
%
%\bel{28X13.21}
%   \cap_{0\leq j\leq s-2}C^j([0,a ]\times[0,b_0]; H^{s-j-2}(Y))
% \,.
%\ee
%%
\end{enumerate}
\end{Theorem}

\begin{Remark}{\rm
\label{R28X13.2} The  sequence $(f_i)_{i\in\N}$ of point 2.\
converges also in $C^1( [0,a ]\times[0,b_0]\times Y)$.
}\qed\end{Remark}

\begin{proof}
Let us denote by $\|\overline f \|_\ell$ the norm associated to
\eq{23X13.1}, and by $|||  f |||_s$ the norm in the space
\eq{28X13.21}.

Let $(\overline f_{i,j})_{j\in \N}$ be a sequence of smooth initial
data such that
$$  \|\overline f_{i,j}-\overline f_i\|_\ell \le \frac 1 {2^j}
 \,.
$$

Let $f_{i,j}$ be the (smooth) solution of  \eq{fos} with initial
data $\overline f_{i,j}$. By the estimates of
Section~\ref{sub:Bounds-iterati-scheme} for all $i,j$ large enough
we can find $0\le a_{*}\le a_0$ so that all the  $f_{i,j}$'s are
defined on a common set $[0,a_*]\times [0,b_0]\times Y$, with a
common bound in $C^{1,1}([0,a_*]\times [0,b_0]\times Y)$.

By Arzela-Ascoli, when $j$ tends to infinity the $f_{i,j}$'s
converge to a solution of \eq{fos}, say $g_i$, with initial data
$\overline f_i$. By uniqueness $g_i=f_i$. This proves point 1.
%
%Proposition~\ref{Piter2} applied to $(\overline f_{i,j})_{i,j\in \N}$, possibly restricting to $i$ and $j$ large enough, shows existence of $0<a_*\le a$ so that the conclusions of point 1 of the current theorem hold both for all functions $f_{i,j}$ and   $f_i$.
%
%
%
%To prove point 2., note that   for every function $i\mapsto j(i)$ satisfying $j(i)\ge i$ the sequence $(\overline f_{i,j(i)})_{i\in\N}$ converges to $\overline f$ in $\| \cdot \|_k$.

Since $\overline f_{i,j}$ converges to $\overline f_i$ and
$\overline f_i$ converges to $\overline f$, there exists a sequence
$\overline f_{i,j(i)}$ which converges to $\overline f$ as $i$ tends
to infinity. By the argument just given, the associated solutions of
\eq{fos} $f_{i,j(i)}$ converge, as $i$ tends to infinity, to a
solution $g$ of \eq{fos}. By uniqueness, $g=f$. Hence the
$f_{i,j(i)}$'s converge to $f$.

Thus, for every $\epsilon>0$ there exists $i_\epsilon$ so that for
$i\ge i_\epsilon$ and $j\ge j_0(i)$ we have
 $$
  ||| f_{i,j} - f |||_s\le \frac 12 \epsilon
   \,.
 $$
 But for $j$ large enough it holds that
 $$
   |||    f_{i,j}-  f_i |||_s  \le \frac 12 \epsilon
  \,,
 $$
 which implies the claim.
 \end{proof}

\section{A continuation criterion}
 \label{ss26X13.-1}

What has been said so far easily leads to the following
\emph{continuation criterion} for solutions with \emph{smooth}
initial data:

\begin{Theorem}
  \label{T5VI14.1}
 Suppose that $(
\varphi,\psi)$ is a $C^1$
 solution on $[0,a)\times [0,b_0]\times Y$ of the equations considered so far, for some $a< a_0$, with smooth initial data on $\mcN$.
  If  $(
\varphi,\psi)$ is bounded in $C^1$ norm  on $[0,a]\times
[0,b_0]\times Y$, then there exists $\epsilon>0$  such that the
solution can be extended to a smooth solution defined on
$[0,a+\epsilon]\times [0,b_0]\times Y$.
\end{Theorem}

Indeed, for smooth data, if an a priori control of the $C^1$-norm of
the fields is known,  for any $k$ one obtains the estimate for the
$k $'th order energy directly from \eq{bid5}, \eq{volest2} and
Gronwall's inequality, with no need to introduce the iterative
scheme of Section \ref{sec:iterati-scheme}. We emphasise that in the
current case the constant $\hat{C}_1$ of equation \eq{5X13.11} is
controlled directly.

One would like to have a similar continuation criterion for
solutions of finite differentiability class. However, due to the
losses of differentiability occurring in our argument it is not
clear whether such a result can be established. We have not
attempted to investigate this issue any further.

\chapter{Application to semi-linear wave equations}
 \label{s14XI13.1}

\section{Double-null coordinate systems}
 \label{s25X13.1}

Let $(\mcM,g) $ be  a smooth $(n+1)$-dimensional space-time, and let
$\wmcN^\pm$ be two null hypersurfaces in $\mcM$ emanating from a
spacelike manifold $Y$ of codimension two. We will denote by
$\mcN^\pm$ the intersection of $\wmcN^\pm$ with the causal future of
$Y$.

In order to apply our results above to semi-linear wave equations
with initial data on $\mcN^\pm$ we need to construct local
coordinate systems $(u,v,x^A)$, where the $x^A$'s are local
coordinates on $Y$, near
$$ \mcN:= \mcN^+\cup \mcN^-
$$
so that
\bel{22X13.p1}
 \mcN^-:=\{u=0\}
  \,,
 \quad
 \mcN^+:=\{ v=0\}
 \,.
\ee
 We will further need
\bel{22X13.p2}
 g(\nabla u, \nabla u) = 0 = g(\nabla v, \nabla v)
 \,,
\ee
wherever defined. Such coordinates can be constructed in a standard
way, but we give the details as specific parametrisations will be
needed in the problem at hand.

Let $\ell _Y$ and $\omega _Y$ be any smooth null future pointing
vector fields defined along $Y$ and normal to $Y$ such that $\ell
_Y$ is tangent to $\mcN^+$ and $\omega _Y$ is tangent to $\mcN^-$.
Then both $\wmcN^+$ and $\mcN^+ $ are threaded by the null geodesics
issued from $Y$ with initial tangent $\ell _Y$ at $Y$. These
geodesics will be referred to as the \emph{generators} of $\wmcN^+$,
respectively of $\mcN^+$. The  associated field of tangents,
normalised in any convenient way, will be denoted by $\ell^+$. Let
$r_+$  denote the corresponding  parameter along the integral curves
of $\ell^+$,  with $r_+=0$ at $Y$. We emphasise that the
normalisation of $\ell^+$ is arbitrary at this stage, so that $r^+$
could e.g.\ be required to be affine, but we do \emph{not} impose
this condition. Similarly $\wmcN^- $ and $\mcN^- $ are threaded by
their null geodesic generators issued from $Y$, tangent to $\omega
_Y$ at $Y$, with field of tangents $\omega^-$ and  parameter $r_-$.

Let $x^A _Y$ be any local coordinates on an open subset $\mcO$ of
$Y$, they can be propagated to functions $x^A_\pm $ on $\mcN^\pm$ by
requiring the $x^A_\pm$'s to be equal to $x^A_Y$ along the
corresponding null geodesic generators of $\mcN^\pm$. Then
$(r_\pm,x^A_\pm)$ define local coordinates on $\mcN^\pm$ near each
of the relevant generators.

On $\wmcN^+$ we let $\omega^+$ be any smooth field of null vectors
transverse to $\wmcN^+$ and normal to the level-sets of $r_+$ such
that $\omega^+|_Y=\omega_Y $.  The function $u$ is defined by the
requirement that $u$ is constant along   the null geodesics issued
from $\wmcN^+$ with initial tangent $\omega^+$, equal to $r_+$ at
$\wmcN^+$. We denote by $\omega$ the field of tangents to those
geodesics, normalised in any suitable way. Thus
\bel{22X13.p5}
 \omega(u)=0\,,
 \quad
 u|_{\mcN^-}= 0
 \,.
\ee

We claim that the level sets of $u$, say $\mcN_u^-$, are null
hypersurfaces. To see this, consider a  one-parameter family
$\lambda\mapsto x(\lambda,s)$ of generators within $\mcN^-_u$. Then
$X:=\partial_\lambda x$ is tangent to $\mcN^-_u$ and solves  the
Jacobi equation along each of the generators $s\mapsto
x(\lambda,s)$. Further, every vector  tangent to $\mcN^-_u$ belongs
to such a family of vectors. We have
\bean
 \frac{d  (g(X,\omega))}{ds }  &= & g\big(\frac{DX}{ds},\omega\big)
 =g\big(\frac{D  }{\partial s  }\frac{  \partial  x}{  \partial \lambda},\frac{\partial x}{\partial s}\big)
 =g\big(\frac{D  }{\partial \lambda  }\frac{  \partial  x}{  \partial s},\frac{\partial x}{\partial s}\big)
\\
 &= &
  \frac 12 \partial_\lambda \big(g(\omega,\omega)\big)=0
  \,.
\eeal{29X13.1}
Now, on $\mcN^-_u\cap \mcN^+$ the vector $X$  can be decomposed as
$X=X^\parallel + \alpha \omega$, where $X^\parallel$ is tangent to
$\mcN^-_u\cap \mcN^+$ and $\alpha\in \R$. Both $X^\parallel $ and
$\omega$ are orthogonal to $\omega$, hence $g(X,\omega)=0$  at the
intersection. \Eq{29X13.1} gives $g(X,\omega)\equiv 0$.   This shows
that all vectors tangent to $ \mcN_u^-$ are orthogonal to $\omega$,
and since $\omega$ is also tangent to $ \mcN_u^-$ we conclude that
$T\mcN_u^-$ is null. Consequently $\nabla u$ is proportional to the
null vector $\omega$, and  thus
$$
 g(\nabla u, \nabla u) =0
 \,.
$$

Similarly, on $\mcN^-$ we let $\ell^-$ be any smooth field of null
vectors transverse to $\mcN^-$ and normal to the level-sets of $r_-$
such that $\ell^-|_Y=\ell_Y $.  The function $v$ is defined by the
requirement that $v$ is constant along   the null geodesics issued
from $\mcN^-$ with initial tangent $\ell^-$, and with initial value
$r_-$ at $\mcN^-$. We denote by $\ell$ the field of tangents to
those geodesics, normalised in any convenient way. It holds that
\bel{22X13.p6}
 \ell(v)=0
 \,,
 \quad
 v|_{\mcN^+}= 0
 \,,
 \quad
 g(\nabla v,\nabla v)=0
 \,.
\ee

By construction we have
\bel{22X13.p3}
 \ell|_{\mcN^\pm}=\ell^\pm
 \,,
  \qquad
 \omega|_{\mcN^\pm}=\omega^\pm
 \,.
\ee

So far the construction was completely symmetric; this symmetry will
be broken now by defining the functions $x^A$ through the
requirement that the $x^A$'s be constant along the null geodesics
starting from $\mcN^-$ with initial tangent $\ell^-$, and taking the
values $x^A_-$ at the intersection point.

The construction just given breaks down when the geodesics start
intersecting. However, it always provides the desired coordinates in
a neighborhood of $\mcN$. In particular, given two generators of
$\mcN^\pm$ emanating from the same point on $Y$, there exists a
neighborhood of those generators on which $(u,v,x^A)$ form a
coordinate system. We emphasise that
\bel{22X13.p4}
 g(\omega, \omega)= g(\ell, \ell)= 0
% , \quad g(\omega, \ell)= g(\omega, \ell)= -{1\over 2}
 \,,
\ee
and that we also have
\bel{22X13.p5}
  \ell^v=0=\ell^A
  \
  \Longleftrightarrow
  \
    \ell = \ell^u \partial_u
 \,,
 \quad
 \omega^u=0
  \
  \Longleftrightarrow
  \
  \omega= \omega^v \partial_v + \omega^A \partial_A
% , \quad g(\omega, \ell)= g(\omega, \ell)= -{1\over 2}
 \,.
\ee
The first group of equations \eq{22X13.p5} follows from the fact
that both $x^A$ and $v$ are constant along the integral curves of
$\ell$, while the second is a consequence of the fact that $u$ is
constant along the integral curves of $\omega$.

Finally, once the coordinates $u$ and $v$ have been constructed, for
some purposes it might be convenient to rescale $\ell$, or $\omega$,
or both, so that
\bel{22X13.p5b}
 g(\omega, \ell)= -{1\over 2}
 \,.
\ee
Such rescalings do not affect \eq{22X13.p4}-\eq{22X13.p5}, which are
the key properties of $\ell$ and $\omega$ for us. \Eq{22X13.p5b}
determines $\ell$ and $\omega$ up to one multiplicative strictly
positive factor, $\ell\mapsto \alpha \ell$, $\omega\mapsto
\alpha^{-1} \omega$.

\subsection{$\R$-parametrisations}
 \label{ss29X13.1}

Let us finish this section by providing a construction in which the
functions $u$ and $v$ run from zero to infinity on all generators of
$\mcN^+$ and $\mcN^-$.

Let $\mcU^+\subset\wmcN^+\times \R$ be the maximal domain of
definition of the map, which we denote by
$$\Psi^+(p,s):
 \mcU^+
 \to
 \mcM
\,,
 \quad
  p
  \in \wmcN^+\,, \ s\in \R
  \,,
$$
%,
defined by following a null geodesic from $p\in\wmcN^+$  an
affine-parameter $s\in \R$ in the direction  $\ell^+$ at $p$. Let
$\mcV^+\subset \mcU^+$ be the domain of injectivity of $\Psi^+$.
Then $\Psi^+(\mcV^+)$ is an open subset of $\mcM$ containing
$\wmcN^-$.

Let the set $\mcU^-\subset\wmcN^-\times \R$, the map $\Psi^-$, and
the set $\mcV^-\subset \mcU^-$ be the corresponding constructs on
$\wmcN^-$, using the integral curves of $\omega$. Then
$\Psi^-(\mcV^-)$ is an open subset of $\mcM$ containing $\mcN^+$.

Set
$$
 \mcO:= \Psi^+(\mcV^+)\cap \Psi^-(\mcV^-) \supset \mcN^+\cup \mcN^-
 \,.
$$
Let $h$ be any complete smooth Riemannian metric on $\mcO$, rescale
$\ell$ and $\omega$  to new vector fields on $\mcO$, still denoted
by $\ell$ and $\omega$, so that $h(\ell,\ell)=1=h(\omega,\omega)$.
Then the integral curves of $\ell$ and $\omega$ are complete in
$\mcO$. The corresponding parameters $r_\pm$ on $\wmcN^\pm$ run over
$\R$ for all generators of $\wmcN^\pm$, as desired.

It should be pointed out that the above normalisation of $\ell$ and
$\omega$ has only been imposed for the sake of constructing $u$ and
$v$. Once we have the functions $u$ and $v$  on $\mcO$ we can revert
to any other normalisation of the fields $\ell$ and $\omega$, in
particular we can assume that \eq{22X13.p5b} holds.  It might then
not be true anymore that $\ell(u)=1$ on $\mcN^-$ and/or
$\omega(v)=1$ on $\mcN^+$, but these conditions are irrelevant for
our purposes in this section. In fact, the condition \eq{22X13.p5b}
plays no essential role in what follows in this section either.

\subsection{Regularity}
 \label{ss25X13.1}

Now, it is well known that coordinate systems obtained by shooting
geodesics lead to a loss of differentiability of the metric. The aim
of this section is to show that, in our context, the optical
functions $u$, $v$   are of the same
differentiability class as the metric.%
\footnote{The argument here has been suggested to us by Hans
Lindblad. We are grateful to Hans for useful discussions concerning
this point.}
As a result, after passing to a doubly-null coordinate system one
loses one derivative of the metric. While unfortunate, this is not a
serious problem for  semi-linear equations, as considered in this
section. On the other hand, this leads to difficulties when
attempting to apply our techniques to the harmonically-reduced
Einstein equations. This is  why we will restrict ourselves to
dimension four when analyzing the Einstein equations, as then a
doubly-null formulation of Einstein equations is directly available,
without having to pass to harmonic coordinates.

First, to avoid a conflict of notation, we will use the symbol  $x$
for the coordinate $u$ of Section~\ref{sec:iterati-scheme}, and $y$
for the coordinate $v$ used there, \emph{without  assuming that
$x$ or $y$ solve the eikonal equation.} Thus, we let $(x,y,x^A)$ be
any coordinate system such that $\mcN^-=\{x=0\}$, and
$\mcN^+=\{y=0\}$. We assume that these hypersurfaces are
characteristic for the metric $g$. We have just seen how to
construct solutions $u$ and $v$ to the eikonal equation, and we wish
to analyze their differentiability properties.

To obtain the desired estimates, we start by differentiating the
eikonal equation:
\bel{25X13.11}
 g^{\mu\nu} \partial_\mu v \, \partial_\nu v = 0
 \quad
 \Longrightarrow
 \quad
 g^{\mu\nu} \partial_\nu v \, \partial_\mu\partial_\alpha v =
 -\frac 12 \partial_\alpha(g^{\mu\nu}) \partial_\mu v \, \partial_\nu v
 \,.
\ee
Setting $f\equiv \varphi\equiv (\varphi_\alpha  ) :=
(\partial_\alpha v)$, we obtain a symmetric-hyperbolic evolution
system
\bel{25X13.12} g^{\mu\nu}\varphi_\mu \partial_\nu \varphi_\alpha =
-\frac 12\partial_\alpha(g^{\mu\nu}) \varphi_\mu \varphi_\nu
 \quad
 \Longleftrightarrow
  \quad
 A^\mu \partial_\mu \varphi = G
 \,,
\ee
with
\bel{25X13.13}
 A^\mu =(A^\mu{}_\alpha{}^\beta)= (-g^{\mu\nu}\varphi_\nu \delta_\alpha^\beta)
 \,,
 \quad
 G = (G_\alpha) = (\frac 12\partial_\alpha(g^{\mu\nu}) \varphi_\mu \varphi_\nu)
% \,.
\ee
(the negative sign above is related to our convention $(-+\cdots +)$
for the  signature of the metric, together with the requirement that
$\nabla u$ and $\nabla v$ are both past pointing). The function $v$
is required to vanish on $\mcN^+$.

An obvious corresponding equation can be derived for the second null
coordinate $u$, which is required to vanish on $\mcN^-$.

We have:

\begin{Theorem}
\label{T25X13.1} Let $(\mcM,g)$ be a smooth space-time with a metric
$g$ with components
$$
 g_{\mu\nu} \in \cap_{0\leq j\leq \ell}C^j([0,a_0]\times[0,b_0];H^{\ell-j}(Y))
$$
in the coordinate system above, with some  $\ell\in\N$ satisfying
$\ell> \frac{n+6}2$. Let $\overline  u$, $\overline  v$   be
continuous functions  on $\mcN $, with $\overline  u \equiv 0$ on
$\mcN^-$, differentiable on $\mcN^+$ and  $
\partial_x \overline  u$ strictly positive there, and $\overline  v \equiv 0$ on $\mcN^+$, differentiable on $\mcN^-$ and  $
\partial_y \overline  v$ strictly positive there, with
\bel{25X13.11x}
 \overline  u |_{\mcN^+} \in \cap_{0\leq j\leq \ell }C^j([0,a_0];H^{\ell -j}(Y))
 \,,
 \quad
 \overline  v |_{\mcN^-} \in \cap_{0\leq j\leq \ell }C^j([0,b_0];H^{\ell -j}(Y))
 \,.
\ee
There exist $\ell$-independent constants $\mu_*>0$ and $
  a_* \in(0,a_0]
$,
with $b_0-\mu_* a_*>0$, such that the eikonal equations $g(\nabla u,
\nabla u)=0=g(\nabla v, \nabla v)$ have  unique solutions $u$ and
$v$, realising the initial data $\overline  u$ and $\overline  v$,
defined on \bel{28X13.1}
 \Omega_*:= \{x\in [0,a_*], 0\le y \le b_0 -\mu_* x \}\times Y
% \,,
\ee
(see Figure~\ref{figure2}), of differentiability class $C^3(\Omega_*)$, with $\nabla u$ and $\nabla v$ without zeros and linearly independent there, and satisfying %
\bean \lefteqn{
   u  \in
 L^\infty([0,b_0-\mu_* a_*];H^{\ell}([0,a_*]\times Y))
 }
 &&
\\
 &&
 \cap_{1\leq j\leq \ell-1}C^j([0,b_0-\mu_* a_*];H^{\ell-j}([0,a_*]\times Y))
 \,,
  \label{11V14.11}
\\
 \lefteqn{
   v  \in
 L^\infty([0,a_*];H^{\ell}([0,b_0-\mu_* a_*]\times Y))
 }
 &&
  \nonumber
\\
 & &
  \cap_{1\leq j\leq \ell-1}C^j([0,a_*];H^{\ell-j}([0,b_0-\mu_* a_*]\times Y))
 \,.
\eeal{25X13.21}
The solutions $u$ and $v$ are smooth if the metric and the initial
data $\overline  u$ and $\overline v$ are.
\end{Theorem}

\begin{figure}[t]
\begin{center} {
\psfrag{a0}{\huge $a_*$} \psfrag{b0}{\huge $b_0$} \psfrag{aa}{\huge
$b_0-\mu_* a_*$} \psfrag{x}{\huge$x$} \psfrag{y}{\huge$y$}
\psfrag{nm}{\huge $\mcN^-$} \psfrag{np}{\huge $\mcN^+$}
\psfrag{0}{\huge $0$}
%\psfrag{1s}{\huge $II_{\sigma}$}
%\psfrag{2s}{\huge $I_{\sigma}$}
\psfrag{omega}{\huge $\Omega_{*}$}
\resizebox{3in}{!}{\includegraphics{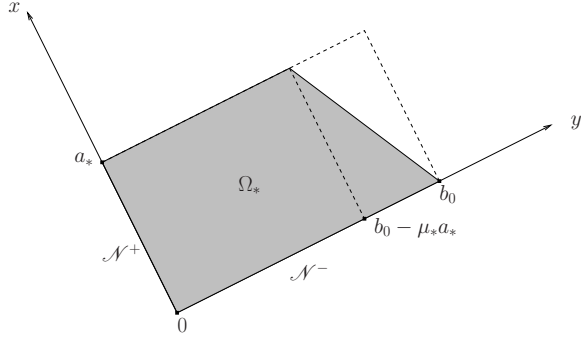}}
 \caption{The set $\Omega_{*}$.\label{figure2}}
 }
\end{center}
\end{figure}

\begin{Remark}{\rm
 \label{R28X13.1}
The constant $\mu_*$
  is only needed for the function $v$, and can be set to zero if $u$
only is considered. The functions $u$ and $v$  possess
differentiability properties similar to those in
\eq{11V14.11}-\eq{25X13.21} on that part of $\Omega_*$ which is not
covered by \eq{11V14.11}-\eq{25X13.21}, which we didn't exhibit as
the result is somewhat cumbersome to write formally.
}\qed\end{Remark}

\begin{proof} Since the metric is $C^2$, existence follows from the
arguments above, and we only need to justify the regularity
properties. The result  is established through a simplified version
of the arguments from Section~\ref{sec:iterati-scheme}. Special care
has to be taken in the proof to make sure that there are no unwanted
contributions to the  energy from some boundaries.

Let us start with the initial data for the function $u$. On $\mcN^+$
the inverse metric takes the form, for any function $\chi$, (see,
e.g., \cite[Appendix~A]{CCM2})
\bel{26X13.1}
 \overline {g(\nabla \chi, \nabla \chi)}|_{\mcN^+} = \overline{g}^{xx}(\partial_ x\overline  \chi )^2
 +
 2\overline{g}^{xy}\partial_x \overline  \chi\, \overline{\partial_y \chi}
 +
 2\overline{g}^{xA}\partial_x \overline  \chi\, \partial_A\overline  \chi
 +
 \overline{g}^{AB}\partial_A \overline  \chi\, \partial_A\overline  \chi
 \,.
\ee
Since neither $\overline{g}^{xy}$ nor $\partial_x \overline u$ has
zeros, it follows from \eq{26X13.1} that the equation $$g(\nabla u,
\nabla u)|_{\mcN^+}=0$$ allows us to calculate
$$
 \overline{\partial_y u}\in \cap_{0\leq j\leq \ell-1}C^j([0,a_0];H^{\ell-j-1}(Y))
$$
on $\mcN^+$ in terms of $\overline g$ and the tangential derivatives
of $\overline  u$, leading to
$$
  \overline \psi|_{\mcN^+} \equiv (\partial_\mu u) |_{\mcN^+} \in \cap_{0\leq j\leq \ell-1}C^j([0,a_0];H^{\ell-j-1}(Y))
  \,.
$$

Next, we will need to control $\nabla u$ on $\mcN^-$, this proceeds
as follows: On $\mcN^-$ we have, for any function $\overline  u$,
\bel{26X13.2}
 \overline {g(\nabla u, \nabla u)} = \overline{g}^{yy}(\partial_ y\overline  u )^2
 +
 2\overline{g}^{xy}\partial_y \overline  u\, \overline{\partial_x u}
 +
 2\overline{g}^{yA}\partial_y \overline  u\, \partial_A\overline  u
 +
 \overline{g}^{AB}\partial_A \overline  u\, \partial_A\overline  u
 \,.
\ee
Since $\overline  u\equiv 0$ in our case, the equation
$\overline{g(\nabla u, \nabla u)}=0$ holds identically. We further
have
\bel{26X13.2a}
 \nabla u|_{\mcN^-} = \overline {g ^{xy} \partial_x u} \partial_y
 \,,
\ee
and we need an equation for $\partial_x u|_{x=0}$. For this we can
use the $u$-equivalent of \eq{25X13.11},
\bel{26X13.4}
 g^{\mu\nu} \partial_\nu u \, \partial_\mu \partial_x u =
 -\frac 12 \partial_x(g^{\mu\nu}) \partial_\mu u \, \partial_\nu u
 \,,
\ee
which on $ {\mcN^-}$ becomes
\bel{26X13.5}
 \overline g^{xy } \overline{\partial_x u} \,   \partial_y (\overline{\partial_x u})  =
 -\frac 12 \partial_x g^{xx}|_{\mcN^-}  (\overline{\partial_x u})^2
% \,.
\ee
(note that $g^{xx}$ vanishes on $\mcN^-$, but there is a priori no
reason why $ \partial_x g^{xx}|_{\mcN^-} $ should vanish as well).
From this it is straightforward to obtain
\bel{26X13.6}
   {\partial_x u}|_{\mcN^-}
   \in
   \cap_{0\leq j\leq \ell-1}C^j([0,b_0];H^{\ell-j-1}(Y))
 \,.
\ee
Summarising:
\beaa
 &
  \overline \psi|_{\mcN^+} \equiv (\partial_\mu u) |_{\mcN^+} \in  \cap_{0\leq j\leq \ell-1}C^j([0,a_0];H^{\ell-j-1}(Y))
  \,,
  &
\\
 &
  \overline \psi|_{\mcN^-}   \in  \cap_{0\leq j\leq \ell-1}C^j([0,b_0];H^{\ell-j-1}(Y))
  \,.
  &
\eeaa

We continue with the energy inequality. Let $h=h_{\alpha\beta} dx^
\alpha dx^ \beta$ be any smooth Riemannian metric on
$\Omega_{a_0,b_0}\times Y$. The $L^2$-energy-density vector
associated with \eq{25X13.12} can be defined as
\bel{26X13.7}
 E^\mu := h(\psi, A^\mu \psi) = - \psi^\mu h(\psi,\psi) =  -h^{\alpha\beta} \partial_\alpha u  \, \partial_\beta u \, \nabla^\mu u
  \,.
\ee
Similarly to Section~\ref{Sec:enid}, the energy inequality with $k=0$
is obtained by integrating the divergence of $e^{-\lambda y } E^\mu
$ over a suitable set, say $\Omega_{a,b,\sigma}$, with $0\le a \le
a_0$, $0\le b \le b_0$, with
\bel{26X13.8}
 \Omega_{a,b ,\sigma}:=\{ 0\le y \le b
 \,,
 \,
 0\le \bar x \le a -\sigma y
 \}
 \,,
\ee
where $\bar x$ will be defined shortly, and where $0<\sigma<
a_0/2b_0 $ is a small constant which will also be determined
shortly.

Indeed, for further purposes we will need to have good control of
the causal character of the level sets of $x$.  This is achieved by
modifying $x$ so that,  after suitable redefinitions, $\partial_x
g^{xx}|_{\mcN^-}=0$. For this, let us  pass to a new coordinate
system
$$
 \bar x = \chi(x,y,x^A) x
 \,,
 \
 \bar y = y
 \,,
 \
 \bar x^A = x^A
 \quad
 \Longrightarrow
  \quad
  \partial_x = \partial_x(x\chi) \partial_{\bar x}
  \,,
$$
with a function $\chi$ which is determined as follows: We have
\beaa
 g^{\bar x \bar x} & = & g^{\mu\nu} \big (
    x \frac {\partial \chi}{\partial x^\mu} +  \frac {\partial x}{\partial x^\mu} \chi)(
    x \frac {\partial \chi}{\partial x^\nu} +  \frac {\partial x}{\partial x^\nu} \chi)
\\
 & = &     x^2 g^{\mu\nu} \frac {\partial \chi}{\partial x^\mu}   \frac {\partial \chi}{\partial x^\nu}
    + 2 x g^{\mu x }
      \frac {\partial \chi}{\partial x^\mu}  \chi
    +g^{xx }   \chi^2
    \,,
\eeaa
leading to
\beaa
 \partial_{\bar x} g^{\bar x \bar x} \big|_{  x = 0}
 & = &   \frac 1 \chi \partial_{  x} g^{\bar x \bar x} \big|_{  x = 0}
     =  2   g^{y x }\big|_{\bar x = 0}
       \partial_y \chi
    +\partial_x g^{xx }\big|_{  x = 0}   \chi
    \,.
\eeaa
This will vanish if we set
\bean
 \chi (0,y,x^A) &=& \exp\bigg(-\frac 12 \int_0^y \frac{\partial_x g^{xx }}{g^{xy}}\bigg|_{  x = 0}(s,x^A) dx
 \bigg)\times  \frac{\partial \overline u}{\partial x}(0,0,x^A)
\\
 &&
  \in  \cap_{0\leq j\leq \ell-1}C^j([0,b_0];H^{\ell-j-1}(Y))
 \,.
\eeal{6XI13.1}
We let $\chi(x,y,x^A)$ be any extension of $\chi(0,y,x^A)$ which is
smooth in all its arguments for $x>0$; the existence of such
extensions is standard. We pass  to the new coordinate system, and
change the notation $(\bar x,\bar y,\bar x^A)$ back to  $(x,y,x^A)$
for the new coordinates. The factor $ \frac{\partial \overline
u}{\partial x}(0,0,x^A)$ in \eq{6XI13.1} has been chosen to obtain
\bel{6XI13.1a}
 \partial _{  x}u(x=0,y=0,x^A)=1
 \,.
\ee
In the new coordinates, from \eq{26X13.5} we find
\bel{6XI13.2}
 \partial _{  x}u(x=0,y,x^A)=1
 \,.
\ee

Now, $\partial \Omega_{a,b,\sigma}$ takes the form
\bean
 \partial \Omega_{a,b ,\sigma}
  & = &
   \mcN^- \cup \mcN^+ \cup \underbrace{\big(\{y\in [0,b ]\,, x= a -\sigma y\}\times Y\big) }_{=:II_\sigma}
\\
 &&
  \cup    \underbrace{\big(\{y=b  \,, 0\le x \le a -\sigma b \}\times Y\big)}_{=:I_\sigma}
 \,,
\eeal{26X13.9}
see Figure~\ref{figure1}.
\begin{figure}[th]
\begin{center} {
\psfrag{a0}{\huge $a_0$} \psfrag{b0}{\huge $b_0$} \psfrag{c0}{\huge
$a_0-\sigma b_0$} \psfrag{x}{\huge$x$} \psfrag{y}{\huge$y$}
\psfrag{nm}{\huge $\mcN^-$} \psfrag{np}{\huge $\mcN^+$}
\psfrag{0}{\huge $0$} \psfrag{1s}{\huge $II_{\sigma}$}
\psfrag{2s}{\huge $I_{\sigma}$} \psfrag{omega}{\huge
$\Omega_{a_0,b_0,\sigma}$}
\resizebox{3in}{!}{\includegraphics[scale=1]{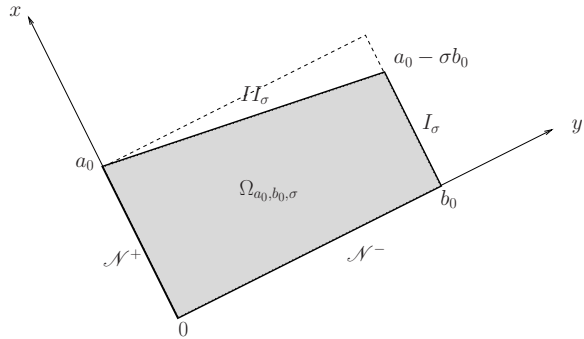}}
\caption{ The set $\Omega_{a_0,b_0,\sigma}$.\label{figure1}}
}
\end{center}
\end{figure}
Before  analyzing the boundary terms arising, recall that we wish to
obtain estimates on various norms of the field. This will be
achieved by repeating the inductive scheme of
Section~\ref{sec:iterati-scheme}, but now using  Sobolev spaces
associated with the level sets of $y$ instead of $H^k(Y)$. For this
we let $\overline  u_i$ be a sequence of smooth functions on
$\mcN^+$ converging to $\overline  u$, with $\overline  u_i=0$ on
$\mcN^-$, each $u_i$ solving a linear equation as done in
Section~\ref{sec:iterati-scheme} with coefficients determined by
$u_{i-1}$. Let $c_1$ and $C_1$ be any positive constants   such that
\bel{27X13.1}
 \sup_{i \in \N  }
 \sup_{   \mcN } \big( |\partial u_i| + |\partial_x \partial u_i| + |\partial^2_x \partial u_i| \big) \le C_1
 \,,
 \quad
 \inf_{i \in \N }
 \inf_ {  \mcN }   \partial_x  u_i\ge  c_1 > 0
 \,.
\ee
Note that a pair of such constants can be determined purely in terms
of the initial data for $u$ on $\mcN$.

To the definition of the sequence $0<a_i$, given just before
\eq{bid9m}, we add the requirement that $a_i \le 1$, and that
\beal{27X13.1}
 &
 \inf_{\Omega_{ a_i,b_0,\sigma_i} } \partial_x  u_i\ge  \frac 12 c_1
 \,,
 &
\\
&
 \sup_{\Omega_{ a_i,b_0,\sigma_i} }\big(  |  \partial u_i|
  +
   |\partial_x  \partial u_i|
  +
   |\partial_x^2  \partial u_i|
    \big)
    \le  C_1 +1
 \,.
 &
\eeal{27X13.1b}

Consider the $L^2$-energy identity on $ \Omega_{a_i,b_0,\sigma_i} $
associated with the equation satisfied by $u_i$,
\bel{11V14.1}
 \int _{\partial\Omega_{a_i,b_0,\sigma}} e^{-\lambda y} E^\mu n_\mu
  =
  \int _{\Omega_{a_i,b_0,\sigma}} \nabla_\mu ( e^{-\lambda y} E^\mu)
  \,,
\ee
with $E^\mu$ given by
\bel{26X13.7x}
 E^\mu := h(\psi_{i}, A^\mu(\psi_{i-1}) \psi_i) = - \psi_{i-1}^\mu h(\psi_i,\psi_i) =  -h^{\alpha\beta} \partial_\alpha u_{i }  \, \partial_\beta u _i\, \nabla^\mu u_{i-1}
  \,.
\ee
On  $\mcN^-=\{x=0\}$ the conormal $n_\mu dx^\mu$ satisfies $n_y=
n_A=0$, so by \eq{26X13.2a} the boundary integrand vanishes
$$
  h(\psi_i,\psi_i)n_\mu \nabla ^\mu u_{i-1}     =
   h(\psi_i,\psi_i) n_x \underbrace{\nabla ^x u_{i-1} }_{=0} =
  0
  \ \mbox{on $\mcN^-$}
 \,.
$$

On  $\mcN^+$ the conormal $n_\mu dx^\mu$ satisfies $n_x= n_A=0$, so
by \eq{26X13.2a} the boundary integrand satisfies
\beaa
 h(\psi_i,\psi_i)\nabla ^\mu u_{i-1} \, n_\mu |_{\mcN^+} & = &
  h(\psi_i,\psi_i) \nabla ^y u_{i-1} \, n_y |_{\mcN^+}=
  h(\psi_i,\psi_i) g^{xy}\partial_x u_{i-1} \, n_y |_{\mcN^+}
\\
  & \sim &
   h(\psi_i,\psi_i)
 \,,
\eeaa
where  ``$f\sim g$''   means that the functions $f$ and $g$ are
bounded by positive constant multiples of each other.

On  $II_{\sigma}$ the conormal $n=n_\mu dx^\mu$ is proportional to $
dx+\sigma dy$. Differentiability of the metric implies that there
exists a constant $C_2$ such that, for $x>0$,
$$
 |g_{\mu\nu}(x,\cdot) - g_{\mu\nu}(0,\cdot)|\le C_2 x
 \,.
$$

By definition of $a_i$, on $\Omega_{a_i,b_0,0}$  it holds that
\bel{12XI13.1}
 |\partial_x^2 \partial u_i |\le (1+C_1)
 \,.
\ee
Since $u_i$ vanishes on $\mcN^-$ so do $\partial_y u_i$ and
$\partial_A u_i$. Further, from \eq{6XI13.2}, $\partial_x\partial_y
u_i$ and $\partial_x\partial_A u_i$ vanish on $\mcN^-$ as well and
we obtain
\bel{12XI13.2}
 |\partial_A u_i| + |\partial_y u_i| \le (1+C_1) x^2
 \,.
\ee
This gives, writing the conormal $n_\mu$ to the level sets of
$II_{\sigma_i}$ as $n_{\mu}dx^\mu = n_x(dx+\sigma_i dy)$,
 with $n_x\ge \delta >0$ and $0<\sigma_i<1$, for any $i$,
\bean
 g^{\mu\nu}  n_\nu \partial_\mu u_i
  & = &
   n_x
 ( g^{\mu x}  + \sigma_i
 g^{\mu y})\partial_\mu u_i
\\
  & = &  \underbrace{n_x}_{\ge \delta}
 \big(
 ( \underbrace{g^{x x}}_{O(x^2)\ge -C C_3 x^2}  + \sigma_i
 \underbrace{g^{x y}}_{\ge c})\underbrace{\partial_x u_i }_{1+O(x) \ge 1 -(1+C_1)x \ge 1/2}
 \nonumber
\\
 &&   +
 \underbrace{( g^{y x}  + \sigma_i
 g^{y y})\partial_y u_i  +
 ( g^{ A x}  + \sigma_i
 g^{A y})\partial_A u_i  }_{\ge -C(1+C_1) x^2}
 \big)
 \nonumber
\\
 &\ge &
  \frac {\delta} 2 \big(\underbrace{(c\sigma_i-\underbrace{CC_3 x^2}_{\le C C_3 a_i x})}_{\ge c \sigma_i /2}
  -  \underbrace{2 C  (1+C_1) x^2}_{\le 2 C (1+C_1) x a_i}
  \big)
  \ge 0
 \,,
  \phantom{xxx}
\eeal{6XI13.3}
for
\bel{6XI13.5}
 x\le \min\big(\frac 1 {2(1+C_1)}, \frac{c  }{2C C_3 }\times \frac{  \sigma_i}{  a_i}
 , \frac{c  }{4C (1+C_1) }\times \frac{  \sigma_i}{  a_i}
  \big)
 \,,
\ee
where $C_1$ is as in \eq{12XI13.1} and
\bel{6X13.4}
 C_3 = \sup | \partial_x^2 g^{\mu\nu}|
 \,.
\ee
Choosing
\bel{7XI13.1}
 \sigma_i = \frac {a_i}{2b_0}
\ee
leads to an $i$-independent bound in \eq{6XI13.5}.

The $II_{\sigma_i}$ boundary integral  gives now a contribution to
the energy identity which we  simply discard, replacing equality by
an inequality.

Note that with this choice we have
\bel{27X13.5}
 [0,\frac 12 a_i]\times [0,b_0]\times Y \subset \Omega_{a_i,b_0,\sigma_i}
  \subset [0,  a_i]\times [0,b_0]\times Y
 \,.
\ee
%.

On $I_{\sigma_i}$ the conormal $n_\mu dx^\mu$ takes the form $n_y dy
$, and from \eq{26X13.2} the boundary integrand takes the form
\beaa
  e^{-\lambda b_0} h(\psi_i,\psi_i)\nabla ^\mu u_{i-1} \, n_\mu |_{I_{\sigma_i}} & = &
  e^{-\lambda b_0} h(\psi_i,\psi_i)\big( g^{xy} \partial_x u_{i-1}  \, n_y  +O(x)\big)|_{y=b_0}
\\
  & \sim &
   h(\psi_i,\psi_i)
 \,.
\eeaa

As a result we obtain
\bel{27X13.8}
 \forall \ 0\le b \le b_0
 \quad
 \mcE_{0,\lambda}[\psi_i,b] \le  C \mcE_{0,\lambda}[\psi_i,0]+ \int_{\Omega_{a_i,b,\sigma_i}} \nabla_\mu (e^{-\lambda y} E^\mu)
 \,,
\ee
similarly for higher-order energy inequalities, with
\bean
 \mcE_{k,
 \lambda}[\psi_i,b]&=&
e^{-\lambda b}
 \sum_{0\le
 j+\ell\le k}
%} \\ \nonumber &&
 \int_{[0,a_i-\sigma_i b]\times Y} |
 \znabla_{q_{r_1}}\ldots \znabla_{q_{{r_j}}}  \partial_x^\ell \psi_i|^2
  dx
 d\mu_Y
\\
 & =&
 e^{-\lambda b} \sum_{0\le \ell \le k}
 \int_0^{a_i-\sigma_i b}
  \|\partial_x^\ell\psi(x,b)\|^2_{H^{k-\ell}(Y)} dx
 \,.
\eeal{bid827X13}

A simpler version of the arguments of
Section~\ref{sec:iterati-scheme} gives the result.

The estimates for $v$ are essentially standard, as we only need to
solve for a short-time in the evolving direction. Should one want to
use an iterative argument as in Section~\ref{sec:iterati-scheme}, we
note that given $\mu_*>0$ as in the statement of the theorem we can
impose an $i$-independent upper bound on the $a_i$'s so that the
boundary
$$
 \{x\in [0,a_*]\,, 0 \le y \le b_0-\mu_*x\} \times Y
$$
gives a non-negative contribution to the energy-identity, and hence
is harmless when considering energy estimates.
\end{proof}

\section{The wave-equation in doubly-null coordinates}
 \label{s25X13.2}

We are ready now to pass to the PDE problem.  Let $W$ be a vector
bundle over $\mcM$. We will be seeking a section  $h$ of   $W$,
defined on a neighborhood of $\mcN^-$  and of differentiability
class  at least $C^2$ there, such that the following hold:
\beqar \Box_g h&=& H(h, \nabla h, \cdot) \quad \mbox{on
$I^+(\mcN^+\cup \mcN^-)$,}
  \label{Eqonde}
\\
  h&=& h^+ \quad \text{on} \quad \mcN^+
  \,,
  \arrlabel{IV+}
\\
  h&=& h^- \quad \text{on} \quad \mcN^-\label{IV-}
  \,.
\eeqar
with prescribed fields $h^\pm$, for some map $H$, allowed to depend
upon the coordinates. For simplicity we assume $H$ to be smooth in
all its arguments, though the results here apply to maps of finite,
sufficiently large, order of differentiability in $h$ and $\nabla
h$, and of Sobolev differentiability in the coordinates: the
resulting thresholds can be easily read from the conditions set
forth in Section~\ref{Sec:enid}.

Let $(u,v,x^A)$ be a coordinate system as in Section~\ref{s25X13.1},
and   let  $\omega, $ and $\ell$ be the vector fields defined there,
with
\bel{decomp}
 g(\omega, \omega)= g(\ell, \ell)= 0
 \,,
 \quad g(\omega, \ell) = -2
 \,.
\ee

As already pointed out, $\ell$ and $\omega$ are determined up to one
multiplicative strictly positive factor,
\bel{14XI13.2}
 \ell\mapsto \alpha \ell\,,
 \quad
  \omega\mapsto \alpha^{-1} \omega
  \,,
  \quad
  \alpha=\alpha(u,v,x^A)>0
  \,.
\ee

Now, every vector orthogonal to $\ell$ is tangent to the level sets
of $v$. Similarly, a vector orthogonal to $\omega$ is tangent to the
level sets of $u$. Hence vectors orthogonal to both have no $u$- and
$v$-components in the coordinate system above. We can thus write
$\left(Vect\{\omega, \ell\}\right)^{\bot}=Vect\{e_B; \; B= 1,
\ldots, n-1\}$, where the $e_B$'s form an ON-basis of $TY$. Thus
$$
 g(e_A,e_B)=\delta^A_B
 \,,
 \quad
  \mbox{and}
  \quad
 e_A = e_A{}^B\partial_B  \
 \Longleftrightarrow
 \
 e_A{}^u= 0 = e_A{}^v
 \,.
$$

For further purposes, we note that the $e_A$'s are determined up to
an $O(n-1)$ rotation:
\bel{14XI13.1} e_A\mapsto \omega_A{}^B e_B \,, \quad
 \omega_A{}^B = \omega_A{}^B(u,v,x^C) \in O(n-1)
 \,;
\ee
this freedom can be used to impose constraints on the projection on
$Vect\{e_B; \; B= 1, \ldots, n-1\}$ of $\nabla_\ell e_A$ or
$\nabla_\omega e_A$.

The inverse metric in terms of this frame reads
$$
 g^{\sharp}= -{1\over2}(\ell\otimes\omega + \omega\otimes\ell)  + \sum_B e_B\otimes e_B
 \,,
$$
so that the wave operator   takes the form
$$
 -{1\over2}\nabla_\omega\nabla_\ell   -{1\over2}\nabla_\ell\nabla_\omega  + \sum _C\nabla_{e_C}\nabla_{e_C}  + \ldots
 \,,
$$
where $\ldots$ denotes first- and zero-derivative terms arising from
the precise nature of the field $h$. This can be rewritten as
$$
 -\nabla_\omega\nabla_\ell     +  \sum _C\nabla_{e_C}\nabla_{e_C}  -\frac 12 [\nabla_\ell, \nabla_\omega] + \ldots \,,
$$
or
$$
  -\nabla_\ell\nabla_\omega     +  \sum _C\nabla_{e_C}\nabla_{e_C} - \frac 12 [\nabla_\omega, \nabla_\ell] +\ldots
  %\,.
$$
(where the commutator terms can be absorbed in ``$\ldots$" in any
case). Setting
\bel{phi-psi}\varphi_0=\psi_0=h, \; \varphi_A=\psi_A= e_A(h), \;
\varphi_+=\omega(h), \; \psi_-= \ell(h) \,,
  \ee
leads to the following set of equations:
\bean \ell(\varphi_0)&=&\psi_0
 \,,
 \nonumber \\
\ell(\vp_+)-\sum_C {e_C}(\psi_C)&=& H_{\varphi_+}
 %+[\nabla_\ell, \nabla_\omega]h
 \,,
 \label{Eqonde1}\\
\ell(\vp_C) -e_C(\psi_-) &=&H_{\varphi_C}
 %-g^{BC}[\nabla_\omega, \nabla_{e_C}]h
 \,,
 \nonumber\\
\omega(\psi_-)-\sum_C e_C(\vp_C)&=& H_{\psi_-}
 %H+[\nabla_\omega,\nabla_\ell]h
 \,,
  \label{Eqonde2}
\\
    \omega(\psi_C) -e_C(\vp_+) &=& H_{\psi_C}
 \,,
 \nonumber
\\
 \omega(\psi_0)&=& \vp_0
\,,
\eea
where $H_{\varphi_+}$, etc., contains $H$ and all remaining terms
that do not involve second derivatives of $h$.

This is a first-order system of PDEs in the unknown $f=
\begin{pmatrix}\vp\\ \psi \end{pmatrix}$,  with $ \varphi=
\begin{pmatrix}\vp_0\\ \vp_+\\\vp_A \end{pmatrix}$  and $ \psi=
\begin{pmatrix}\psi_0\\ \psi_-\\ \psi_A \end{pmatrix}$. Let us check
that it is symmetric hyperbolic, of the form considered in
Section~\ref{Sec:enid}. We have
$$
 A^\mu\nabla_\mu f = \newg(f)
  \,,
$$
equivalently
\bel{eqonde-sys} \begin{pmatrix}
 A^\mu_{\vp\vp}& A^\mu_{\vp\psi} \\  A^\mu_{\psi\vp}& A^\mu_{\psi\psi}
\end{pmatrix} \nabla_\mu
\begin{pmatrix}
 \vp\\ \psi
\end{pmatrix}= \begin{pmatrix}\newg_\vp\\\newg_\psi\end{pmatrix}
 \,,
  \ee
with
\bea
 &
 A^u_{\vp\vp}=\ell^u\cdot \Id\,,
 \quad
 A^u_{\vp\psi}=A^u_{\psi\vp}=A^u_{\psi\psi}= 0  \,,
 &
\\
 &
 A^v_{\psi\psi}=\omega^v\cdot \Id\,,
 \quad
 A^v_{\vp\psi}=A^v_{\psi\vp}=A^v_{\vp\vp}= 0  \,,
 &
\\
 &
 A^B_{\vp\psi} =A^B_{\psi\vp} =\;\;\;\; -\begin{pmatrix}0&0&0&\ldots&0 \\0&0&\delta^{B}_{1}&\ldots&\delta^{B}_{n-1} \\0&\delta^{B}_{1}&0&\ldots&0 \\ \vdots& \vdots&\vdots& & \vdots\\0&\delta^{B}_{n-1}&0&\ldots&0  \end{pmatrix}
  \,,
&
\\
 &
 A^B_{\vp\vp} = 0\,,
  \quad
    A^B_{\psi\psi}= \omega^B\cdot \Id
 \,,
&
\\
 &
  \newg_\vp(\vp,\psi)= \begin{pmatrix}  \psi_0\\ H_{\vp_+}\\H_{\vp_C}\end{pmatrix}\quad \text{and}\quad
 \newg_\psi(\vp,\psi)= \begin{pmatrix}  H_{\psi_-}\\H_{\psi_C}\\\vp_0\end{pmatrix}\,.
 &
\eea

\section{The existence theorem}
 \label{ss28X13.1}

We denote by $\overline {\phi}^-$ the restriction of a map $\phi$ to
$\mcN^-$ and by $\overline {\phi}^+$ the restriction of $\phi$ to
$\mcN^+$.

In order to apply the results  of the previous sections  to the
Cauchy problem  (\ref{IV+})  we need to show, given smooth data $
h^+$ on $\mcN^+$  and $  h^-$ on $\mcN^-$, how to determine the
initial data for $f$ on a suitable subset of  $\mcN^+\cap \mcN^-$,
and that these  fields are in the right spaces. We recall that
$$
  T\mcN^+= Vect\{\ell,\; e_1,\ldots,\; e_{n-1}\}
   \ \mbox{and} \
 T\mcN^-= Vect\{\omega,\;  e_1,\ldots,\; e_{n-1}\}
 \,,
$$
%,
which implies that:
$$
\overline{ \omega ( h)}^-=  \omega (h^-)  \,, \
 \overline{ \ell (h)}^+=  \ell (h^+) \,,
 \
\overline{ {e_B}( h)}^\pm=  {e_B}(  {h}^\pm)\,.
$$
The remaining restrictions $\overline{ \ell (h)}^-$ and $\overline{
\omega (h)}^+$ will be determined using the wave equation: Indeed,
considering the restriction of (\ref{Eqonde1}) to $\mcN^+$ and the
restriction of (\ref{Eqonde2}) to $\mcN^-$ leads to the following,
in general non-linear, \emph{transport equations} for $\overline{
\ell(h)}^-$ and $\overline{ \omega (h)}^+:$
\begin{eqnarray}\label{Eq1}
- \omega(\overline{ \ell (h)}^-) +\overline{g^{BC}}^- \nabla_{e_B}
\nabla_{e_C}h^-&=&
 H_{\psi^-}(h^-, \partial h^-, \overline{ \ell (h)}^-, \cdot)
 \,,
\\
\overline{ \ell(h)}^-\Big|_{\mcN^+\cap\mcN^-} &=&
\ell(h^+)\Big|_{\mcN^+\cap\mcN^-} \,, \nonumber
\end{eqnarray}
and
\begin{eqnarray}
\label{Eq2} - \ell(\overline{ \omega (h)}^+) +\overline{g^{BC}}^-
\nabla_{e_B} \nabla_{e_C}h^+
 &=&
  H_{
 \varphi^+}(h^+, \partial h^+, \overline{ \omega (h)}^+, \cdot)
  \,,
\\
\;\overline{ \omega (h)}^+\Big|_{\mcN^+\cap\mcN^-} &=&
\omega(h^-)\Big|_{\mcN^+\cap\mcN^-} \nonumber
 \,.
\end{eqnarray}
These are ODEs along the integral curves of the vector fields
$\omega$ and $\ell$.

For every generator, say $\Gamma$, of $\mcN^-$ let $\Gamma_0$ be the
maximal interval of existence of the solution  of the transport
equation \eq{Eq2}. Thus the set
$$
  \mcN^-_0=\cup_\Gamma \Gamma_0 \subset\mcN^-
$$
is the largest subset of $\mcN^-$ on which the solution of the
transport equation, with the required data on $\mcN^-\cap \mcN^+$,
exists. By lower semi-continuity of the existence time of solutions
of ODES the set $\mcN^-_0$ is an open subset  of $\mcN^-$.

The set $\mcN_0^+$ is defined analogously.

Applying the construction of Section~\ref{ss29X13.1} to
$\mcN^-_0\cup \mcN^+_0$ instead of $\mcN^-\cup \mcN^+$, we obtain a
double-null coordinate system $(u,v,x^A)$ near $\mcN^+_0\cup
\mcN^-_0$ in which the function  $v$ runs from  $0$ to $\infty$
along all generators of $\mcN^-_0$, and the function $u$ runs from
$0$ to $\infty$ along all generators of $\mcN^+_0$.
Theorem~\ref{T23X13.2} and Remark~\ref{R29X13.11} apply, leading to:

\begin{Theorem}
 \label{T29X13.1}
 Let $\ell\ge \frac{n+11}{2}$.
Consider the Cauchy problem   \eq{IV+}  for a semilinear system of
wave equations, with $H=H(h,\nabla h,\cdot)$  of $C^\ell$
differentiability class
 in all arguments. Without loss of generality we can parameterize $\mcN^\pm$ by $[0,\infty)\times Y$, with the level sets of the first coordinate transverse to the generators of $\mcN^\pm$. Given the initial data
\bel{8V14.1}
 h^\pm\in \cap_{j=0}^\ell C^{j}([0,\infty), H^{\ell-j}(Y))
\ee
denote by
$$
 \mcN_0= \mcN_0^+\cup \mcN_0^- \subset \mcN^+\cup \mcN^-
$$
the maximal domain of existence on $\mcN^-\cup \mcN^+$ of the
transport equations \eq{Eq1}-\eq{Eq2}. There exists a neighborhood
$\mcV$ of $\mcN_0$ and a unique solution $h$ defined there with the
following properties: Reparameterising the generators of
$\mcN^\pm_0$ if necessary, we can  obtain $\mcN^\pm_0\approx
[0,\infty)\times Y$. Then for every $i\in \N$ there exist
$a_i,b_i>0$ so that the set (see Figure~\ref{fig4})
$$
 \mcV_i:=\bigg(\underbrace{\big([0,a_i]\times [0,i]\big)\cup \big([0,i]\times [0,b_i]\big)}_{=:\mcU_i}\bigg)\times Y
$$
is included in $\mcV$, and we  have
\bean \lefteqn{h\in
 L^{\infty}\Big(\mcU_i; H^{\ell-2}(Y)\Big)\cap W^{1,\infty}\Big(\mcU_i; H^{\ell-3 }(Y)\Big)
  }
  &&
\\
&&
  \cap _{0< 3j \le \ell - \frac{n+11}2}
    W^{j+1,\infty}\Big(\mcU_i; H^{\ell-2-3j}(Y)\Big)
%\\
%&&
     \subset
     C^{\ell_1-1,1}(\mcU_i\times Y)
    \,,
     \phantom{xxxx}
\eeal{10III14.t2}
with the last inclusion holding provided that $\ell > \frac{n+17}
2$, with $\ell_1\ge 1$ being the largest integer such that
$\ell-3\ell_1>\frac{n+11} 2$.
%
%Further, for $\ell>6+(n+11)/2$, $h$ belongs to
%%
%\bean
%&
%C\Big(\mcU_*; H^{\ell-2}(Y)\Big)\cap C^{1}\Big(\mcU_*; H^{\ell-3 }(Y)\Big) \cap _{0< 3i \le \ell - \frac{n+11}2}
%    C^{i+1 }\Big(\mcU_*; H^{\ell-2-3i}(Y)\Big)
%&
%\\
%&
%     \subset
%     C^{\ell_1 }(\mcU_*\times Y)
%    \,.
%&
%\eeal{10III14.t3}
%%%
The solution depends continuously on initial data,
 and is smooth if the initial data are.
\end{Theorem}
\begin{figure}[t]
\begin{center} {
\psfrag{i}{\huge $i$} \psfrag{i1}{ \huge$i$+1} \psfrag{bi}{\huge
$b_i$} \psfrag{ai}{\huge $a_i$} \psfrag{ai1}{\huge$a_{i+1}$}
\psfrag{bi1}{\huge$b_{i+1}$} \psfrag{nm}{\huge $\mcN^-$}
\psfrag{np}{\huge $\mcN^+$} \psfrag{V}{\huge $\mcV$}
\resizebox{3in}{!}{\includegraphics{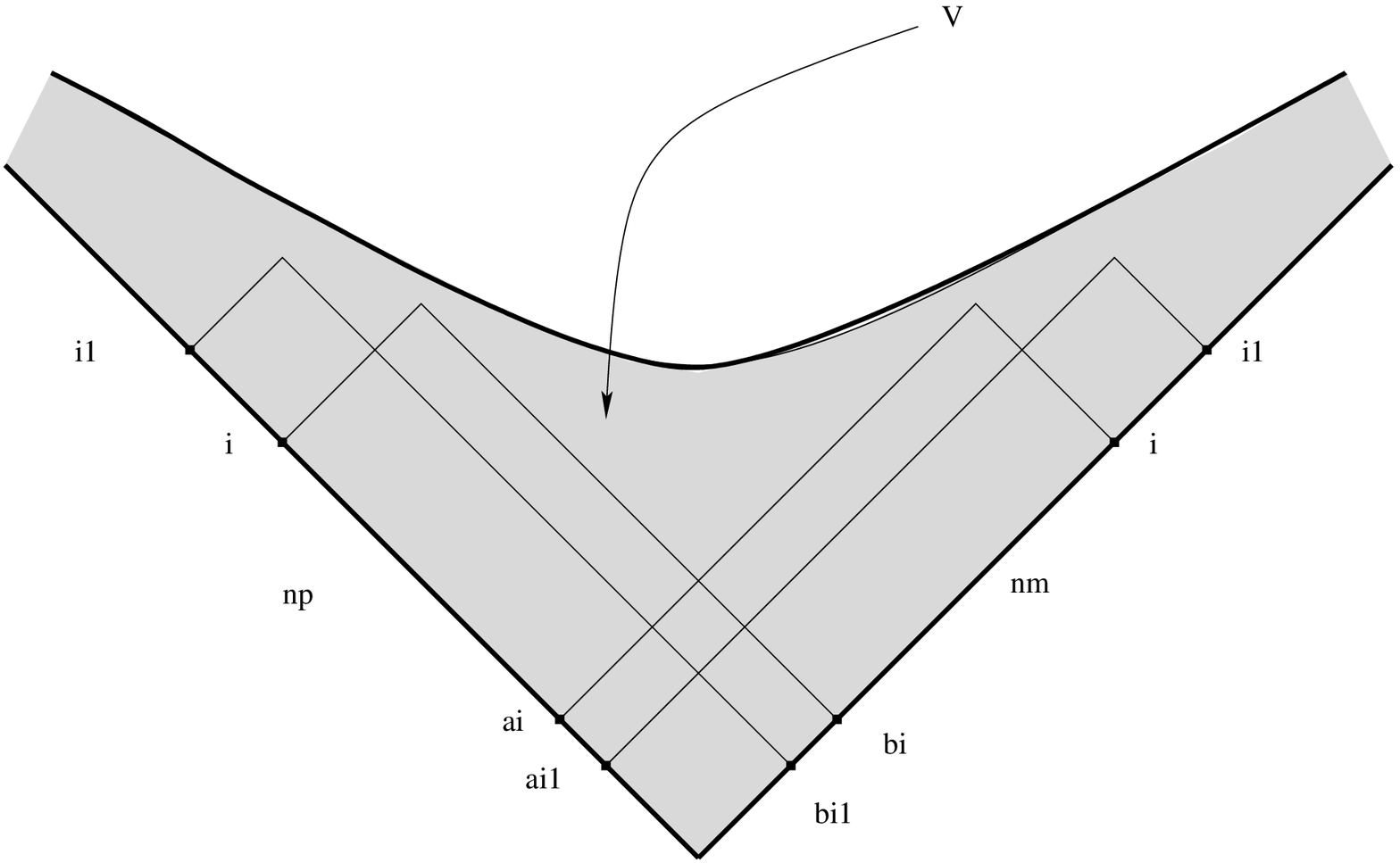}}} \caption{The
neighborhood $\mcV$ of $\mcN$.}\label{fig4}
\end{center}
\end{figure}

\begin{Remark}{\rm
\label{R8V14.11} Condition \eq{8V14.1} will hold for $h^\pm\in
C^\ell([0,\infty)\times Y)$.
}\qed\end{Remark}

\begin{Remark}{\rm
\label{R8V14.11x} An obvious analogue of Remark~\ref{R8V14.1}
concerning further regularity of $h$ applies.
}\qed\end{Remark}

\chapter{Einstein equations}

In this section we will show that our existence theorems above can
be used (in a somewhat indirect manner) to establish neighborhood
theorems for both Einstein equations with suitable sources and the
Friedrich conformal vacuum Einstein equations.

One could try to analyse whether the harmonic coordinate reduction
of Einstein equations leads to equations with a non-linearity
structure to which Theorem~\ref{T23X13.2} applies. Here a problem
arises, because our iteration scheme requires a doubly-null
decomposition of the principal symbol of the wave equation, which is
the wave operator. This requires going to harmonic coordinates, but
those lead to a loss of derivative of the coefficients. It is
conceivable that this can be overcome, but it appears simpler to
work directly in a formalism where the doubly-null decomposition of
the equation is built-in from the outset, namely the
Newman-Penrose-Friedrich-Christodoulou-Klainerman-Nicol\`o
equations. We will show that this decomposition fits indeed in our
set-up.

We use the conventions and notations of~\cite{ChLengardprep}. For
the convenience of the reader we include in
Appendix~\ref{Aformalism} a shortened version of that section
in~\cite{ChLengardprep} which introduces the relevant formalism.

\section{The Einstein vacuum equations}
 \label{ss13XI13.1}

We start with the vacuum Einstein equations, which we write as a set
of equations for a tetrad $e_q=e_q{}^\mu\partial_\mu$, for the
related connection coefficients defined as
\bel{13XI13.5}
 \nabla_i e_j = \Gamma_i{}^k{}_j e_k
 \,,
\ee and for the tetrad components $d^i{}_{jk\ell}$ of the Weyl
tensor. We assume that the scalar products $g_{ij}:=g(e_i,e_j)$ are
point-independent, with the matrix $g_{ij}$ having Lorentzian
signature. We require that $\nabla$ is $g$-compatible, which is
equivalent to
\bel{13XI13.6}
 \Gamma_{ijk}= -\Gamma_{ikj}
 \,,
 \
 \mbox{where $\Gamma_{ijk}:= g_{j\ell} \Gamma_{i}{}^\ell{}_{k}$}
 \,.
\ee

Consider the set of equations due to Friedrich
(see~\cite{FriedrichCargese} and references therein)
\begin{deqarr}
&&[e_{p},e_{q}] = (\Gamma_{p}\,^{l}\,_{q} -
\Gamma_{q}\,^{l}\,_{p})\,e_{l}
 \,,\label{E1x}
\\
&&e_{p}(\Gamma_{q}\,^{i}\,_{j}) - e_{q}(\Gamma_{p}\,^{i}\,_{j}) -
2\,\Gamma_{k}\,^{i}\,_{j}\,\Gamma_{[p}\,^{k}\,_{q]} +
2\,\Gamma_{[p}\,^{i}\,_{|k|} \Gamma_{q]}\,^{k}\,_{j} \nn\\&&\quad =
d^{i}\,_{jpq} + \delta^{ i}_{[p}R_{q]j} -  g_{j[p} R_{q]}{}^{i}
 + \frac R 3 g_{j[p} \delta_{q]} ^{i}
  \,,
 \label{conf6x}
\\
  \label{E3x}
 &&\DD_{i} d^{i}\,_{jkl} = J_{jkl}
 \,.
  \arrlabel{Hconf1Ex}
\end{deqarr}
The first equation \eq{E1x} says that $\Gamma$ has no torsion.
Recall that we assumed that the $\Gamma_{ijk}$'s are anti-symmetric
in the last two-indices; together with \eq{E1x} this implies that
$\Gamma$ is the Levi-Civita connection of $g$. We will assume that
$d_{ijkl}$ has the symmetries of the Weyl tensor, then the
left-hand-side of \eq{conf6x} is simply the definition of the
curvature tensor of the connection $\Gamma$, with $d_{ijkl}$ being
the Weyl tensor, $R_{ij}$ being the Ricci tensor and $R$ the Ricci
scalar.

 In vacuum ($R_{ij}\equiv \lambda g_{ij}$ for some constant $\lambda$),  \eq{E3x}
with $J_{jkl}\equiv 0$ follows from the Bianchi identities for the
curvature tensor.

As shown by Friedrich~\cite[Theorem~1]{F2}
(compare~\cite{TimConformal}), every solution of \eq{Hconf1Ex} with
$J_{jkl}\equiv 0$ satisfying suitable constraint equations on the
initial data surface is a solution of the vacuum Einstein equations.

In Section~\ref{ss15XI13.1} below we will consider a class of
non-vacuum Einstein equations, in which case we will complement the
above with equations for further fields satisfying wave equations,
and then $R_{ij}$ and $J_{jkl}$ in \eq{Hconf1Ex} will be viewed as
prescribed functions of the remaining fields, their first
derivatives, the tetrad, the Christoffel coefficients, and the
$d_{ijkl}$'s, as determined from the energy-momentum tensor of the
matter fields.

We would like to apply Theorem~\ref{T23X13.2} to the problem at
hand. The first step  is to show that we can bring a subset of
\eq{Hconf1Ex}  to the form needed there. This will be done using the
frame formalism of Christodoulou and Klainerman, as described in
Appendix~\ref{Aformalism}.

Before pursuing, we will need to reduce the gauge-freedom available.
For this we need to understand what conditions can be imposed on
coordinates and frames without losing generality.

Given a metric $g$, we have seen in Section~\ref{s14XI13.1} how to
construct a coordinate system $(u,v,x^A)$ and vector fields $e_i$,
$i=1,\ldots,4$, with $e_3$ proportional to the vector field $\ell$
constructed there,
 and $e_4$ proportional to $\omega$ there,
so that the metric takes the form \eq{metric} below, with
\bel{14XI13.11}
 e_3 = \partial_u
 \,,
 \quad
 e_4 = e_4{}^v\partial_v + e_4{}^A\partial_A
 \,.
\ee
With this choice of tetrads $e_i = e_i{}^\mu\partial_\mu$,
\eq{E1x} becomes  an evolution equation for the tetrad coefficients
$e_i{}^\mu$
\bel{14XI13.15}
 [e_3,e_i] = \partial_u e_i{}^\mu \partial_\mu = (\Gamma_{3}\,^{l}\,_{i} -
\Gamma_{i}\,^{l}\,_{3})\,e_{l}
 \,.
\ee
By construction, the $\partial_u e_a$'s have no $u$ and $v$
components, which gives the identities
\bel{14XI13.21}
 0 = (\Gamma_{3}\,^{3}\,_{a} -
\Gamma_{a}\,^{3}\,_{3})
 =(\Gamma_{3}\,^{4}\,_{a} -
\underbrace{\Gamma_{a}\,^{4}\,_{3}}_{=0})
 \,.
\ee
Equivalently, in the notation of Appendix~\ref{Aformalism},
\bel{14XI13.22}
  \eta_a = \zeta_a
  \,,
  \quad
  \underline \xi_a=0
  \,.
\ee
Similarly, $\partial_u e_4$ has no $u$ component, which implies
\bel{14XI13.25}
 0 = \underbrace{\Gamma_{3}\,^{3}\,_{4}}_{=0} -
 \Gamma_{4}\,^{3}\,_{3}
 \quad
 \Longleftrightarrow
 \quad
  \underline \upsilon = 0
 \,.
\ee

Next, the vector fields $e_a$, $a=1,2$ are determined up to
rotations in the planes $\text{Vect}\{e_1,e_2\}$, and we can get rid
of this freedom by imposing
\bel{14XI13.12}
 \Gamma_3{}^a{}_b = 0
 \,.
\ee

By construction, the integral curves of the vector fields $e_3$ and
$e_4$ are null geodesics, though not necessarily
affinely-parameterised:
\bel{epc} \nabla_{e_3}e_3\sim e_3\,,
 \quad
\nabla_{e_4}e_4\sim e_4\,. \ee
In this gauge, using the notation of Appendix~\ref{Aformalism}
(see~\eq{struct1f}),
\bel{uxig} \Gamma_3{}^a{}_3 = 0=\Gamma_4{}^a{}_4 \quad
 \Longleftrightarrow
 \quad
   \xib^a=0= \xi^a\,.
\ee

The vanishing of the rotation coefficients just listed allows us to
get rid of the second term in some of the combinations
$$
 e_{3}(\Gamma_{q}\,^{i}\,_{j}) - e_{q}(\Gamma_{3}\,^{i}\,_{j})
$$
appearing in \eq{conf6x}. In this way, we can algebraically
determine
$$
 \partial_u \Gamma_{q}\,^{a}\,_{3}
 \
 \mbox{and}
 \
 \partial_u \Gamma_{q}\,^{a}\,_{b}
 %
% \,,
% \
% \mbox{and}
% \
% \partial_u \Gamma_{q}\,^{4}\,_{a}
$$
in terms of the remaining fields appearing in \eq{conf6x}. Similarly
we have
$$
 e_{4}(\Gamma_{q}\,^{a}\,_{4}) - e_{q}(\Gamma_{4}\,^{a}\,_{4}) =
 e_{4}(\Gamma_{q}\,^{a}\,_{4})
 \,,
 \quad
 e_{4}(\Gamma_{q}\,^{3}\,_{3}) - e_{q}(\Gamma_{4}\,^{3}\,_{3}) =
 e_{4}(\Gamma_{q}\,^{3}\,_{3})
 \,,
$$
which gives equations for $ e_{4}(\Gamma_{q}\,^{a}\,_{4})$ and  $
e_{4}(\Gamma_{q}\,^{3}\,_{3})$.

Keeping in mind \eq{14XI13.22} and the symmetries of the
$\Gamma_i{}^j{}_k$'s, all the non-vanishing connection coefficients
satisfy ODEs along the integral curves of $e_3=\partial_u$ or of
$e_4$.

The analysis of the divergence equation \eq{E3x} in
Appendix~\ref{Aformalism} leads in vacuum to  the following two
collections of fields,
\beal{phid}
\varphi&=&(e_i,\Gamma_i{}^a{}_b,\Gamma_i{}^a{}_3,\alp,\oubeta,\ro,\si,\obeta)\,,\\
\psi&=&(\Gamma_i{}^a{}_4,\Gamma_i{}^3{}_3,\ubeta,\osi,\oro,\beta,\ualp)
 \,,
\eeal{psid}
with the gauge conditions just given,
\bean
 &
 e_3{}^u=1
 \,,
 \quad
 0=e_3{}^v=e_3{}^A=e_4{}^u=e_a{}^u=e_a{}^v
  \,,
  &
\\
& \Gamma_3{}^3{}_a= \Gamma_a{}^3{}_3
  \,,
  \quad
 0= \Gamma_3{}^a{}_3=  \Gamma_4{}^i{}_4
 = \Gamma_3{}^a{}_b
 %= \Gamma_4{}^3{}_3
    \,,
   &
\eeal{15XI13.1}
to which Theorem~\ref{T23X13.2}, p.~\pageref{T23X13.2} and
Remark~\ref{R29X13.11}, p.~\pageref{R29X13.11} apply. This will be
used to establish our main result for the vacuum Einstein equations.
However, before stating the theorem, an overview of some initial
value problems for the vacuum Einstein equations is in order.

As discussed in detail in~\cite{ChPaetz}, the characteristic initial
data for the vacuum Einstein equations on each of the hypersurfaces
$\mcN^\pm$ consist of a symmetric tensor field $\tilde g$ with
signature $(0,+,\ldots,+)$, so that the integral curves of the
kernel of $\tilde g$ describe the generators of $\mcN^\pm$. To the
tensor field $\tilde g$ one needs to add a connection $\kappa$ on
the bundle of tangents to the generators. In a coordinate system
$(r,x^A)$ on $\mcN$ so that $\partial_r$ is tangent to the
generators we have $\nabla_{\partial_r} \partial_r = \kappa
\partial_r$. The fields $\tilde g $ and $\kappa$ are not arbitrary,
but are subject to a constraint, the Raychaudhuri equation. If we
write
\bel{24V14.1}
 \tilde g= \og_{AB}(r,x^C) dx^A dx^B
\ee
and, in dimension $n+1$, we set
\bel{5VI14.1}
 \tau = \frac 12 \og^{AB} \partial_r \og_{AB}
 \,,
 \quad
 \sigma_{AB}= \frac 12 \partial_r \og_{AB} - \frac 12 \tau \og_{AB}
 \,,
\ee
then, in vacuum, the Raychaudhuri constraint equation reads
\bel{18V14.1}
    \partial_r\tau - \kappa \tau + |\sigma|^2 + \frac{\tau^2}{n-1} = 0
  \,.
  \ee

Here it is appropriate to mention the alternative approach of
Rendall~\cite{RendallCIVP}, where one prescribes the conformal class
of $\tilde g$ and one solves \eq{18V14.1} for the conformal factor,
after adding the requirement  that $\kappa$ vanishes identically.
Thus, in Rendall's scheme the starting point is an initial data
symmetric tensor field $\gamma=\gamma_{AB}(r,x^C)dx^A dx^B$ which is
assumed to form a one-parameter family of Riemannian metrics
$r\mapsto \gamma(r,x^A)$ on the level sets of $r$,  all
assumed to be diffeomorphic to a fixed $(n-1)$-dimensional manifold
$Y$. The conformal factor $\Omega$ relating $\coneg $ and the
initial data $\gamma$, $\overline g_{AB} = \Omega^2 \gamma_{AB}$ can
be written as
\begin{equation}
 \Omega = \varphi \left( \frac{\det s}{\det \gamma}\right)^{1/(2n-2)}
 \,.
 \label{definition_Omega}
\end{equation}
where $s=s_{AB}(x^C)dx^A dx^B$ is any $r$-independent convenient
auxiliary metric on  the surfaces $r=\const$. Note that the field
$\sigma_{AB}$ defined in \eq{5VI14.1} is independent of $\varphi$,
thus is defined uniquely by the representative $\gamma$ of the
conformal class of $\coneg$. One has
\begin{equation}
 \tau =(n-1)\partial_r\log\varphi
 \,,
\label{relation_tau_phi}
\end{equation}
which allows one to rewrite \eq{18V14.1} as a second-order
\textit{linear} ODE:
\begin{equation}
 \partial^2_{r}\varphi -\kappa\partial_r\varphi + \frac{|\sigma|^2}{n-1}\varphi =0
 \,.
 \label{constraint_phi}
\end{equation}
In this case, after solving \eqref{constraint_phi}, one \emph{has
to replace the initial hypersurfaces by its subset on which
$\varphi>0$.}

Recall next that, again in the approach of Rendall
(compare~\cite{CCM2}), the remaining metric functions on $\mcN$ are
obtained by solving linear ODE's along the generators of $\mcN$. One
could then worry that the requirement that the resulting tensor has
Lorentzian signature might lead to the need of passing to a further
subset of $\mcN$. This is indeed the case in the original
formulation of \cite{RendallCIVP}, the problem goes away when
handled appropriately, as it can be reformulated in such a way that
the remaining metric functions are freely
prescribable~\cite{ChPaetz}.

We finally note that the characteristic data on each of $\mcN^\pm$
have to be complemented by certain data on $\mcN^+\cap \mcN^-$, the
precise description of which is irrelevant here; the reader is
referred to \cite{RendallCIVP,ChPaetz,CCM2} for details.

To continue, it is useful to summarize some known
results about Cauchy problems for the Einstein equations:

\begin{Theorem}[Rendall]
   \label{T18V14.1}
   Given smooth characteristic vacuum initial data on $\mcN=\mcN^+\cup \mcN^-$, complemented by suitable data on $Y:=\mcN^+\cap \mcN^-$,  there exists a unique, up to isometry, vacuum metric defined in a future neighborhood of $\mcN^+\cap \mcN^-$, as shown in Figure~\ref{fig5}.
\end{Theorem}
\begin{figure}[th]
\begin{center} {
\psfrag{n+}{\huge $\mcN^+$} \psfrag{n-}{\huge $\mcN^-$}
\psfrag{V}{\huge $V_0$}
\resizebox{2.2in}{!}{\includegraphics{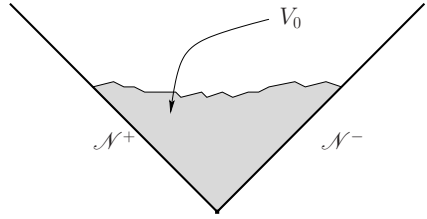}}}
\caption{The guaranteed domain of existence of the solution in
Rendall's Theorem.}\label{fig5}
\end{center}
\end{figure}

The following result is standard:

\begin{Theorem}
   \label{T18V14.2}
   Given  smooth vacuum initial data on a spacelike hypersurface $\Sigma$ with non-empty boundary $\partial \Sigma$ there exists   a unique, up to isometry, vacuum metric defined in a future neighborhood of $\Sigma$, bounded near $\partial \Sigma$ by smooth null ``ingoing'' hypersurfaces orthogonal to $\partial \Sigma$, as shown in Figure~\ref{fig5a}.
\end{Theorem}
\begin{figure}[th]
\begin{center} {
\psfrag{sigma}{\huge $\Sigma$}
%\psfrag{n+}{\huge $\mcN^+$}
%\psfrag{n-}{\huge $\mcN^-$}
%\psfrag{V}{\huge $V_0$}
%
\resizebox{1.5in}{!}{\includegraphics{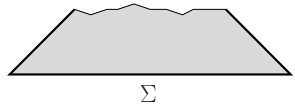}}} \caption{The
guaranteed future domain of existence of the solution with initial
data on a hypersurface with boundary.}\label{fig5a}
\end{center}
\end{figure}

One has of course a similar domain of existence to the past of
$\Sigma$, but this is irrelevant for our purposes.

From Theorems~\ref{T18V14.1} and \ref{T18V14.2} one easily obtains
existence of solutions of the mixed Cauchy problems illustrated in
Figures~\ref{fig7} and \ref{fig6}.
\begin{figure}[ht]
\begin{center} {
%\psfrag{n+}{\huge $\mcN^+$}
\psfrag{n-}{\huge $\mcN^-$} \psfrag{V0}{\huge $V_0$}
\psfrag{sig}{\huge $\Sigma$}
\resizebox{1.8in}{!}{\includegraphics{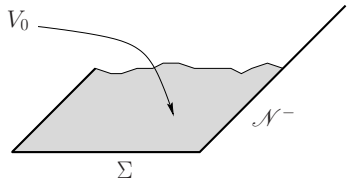}}} \caption{The
guaranteed domain of existence of solutions of a mixed Cauchy
problem with a ``left'' boundary and a characteristic initial data
hypersurface emanating normally from the ``right''
boundary.}\label{fig7}
\end{center}
\end{figure}
\begin{figure}[th]
\begin{center} {
\psfrag{n+}{\huge $\mcN^+$} \psfrag{n-}{\huge $\mcN^-$}
\psfrag{V0}{\huge $V_0$} \psfrag{sig}{\huge $\Sigma$}
\resizebox{2.1in}{!}{\includegraphics{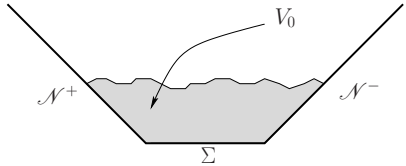}}} \caption{The
guaranteed domain of existence of solutions of a mixed Cauchy
problem with characteristic initial data hypersurfaces emanating
normally from the boundaries of a spacelike
hypersurface.}\label{fig6}
\end{center}
\end{figure}

We are ready to pass now to our main result, which for simplicity we
state for smooth metrics. The interested reader can   chase the
losses of differentiability which arise in various steps of the
proof to obtain the corresponding theorem with initial data of
finite Sobolev differentiability; compare~\cite{Luk}:

\begin{Theorem}
\label{TEin1} For any set of smooth characteristic initial data for
the vacuum Einstein equations on two transversally intersecting null
hypersurfaces $\mcN:=\mcN^+\cup \mcN^-$ there exists a smooth vacuum
metric defined in a  future  neighborhood $\mcU$ of $\mcN$. The
solution is unique up to diffeomorphism when $\mcU$ is appropriately
chosen and when appropriate initial data on $\mcN^+\cap \mcN^-$ are
given.
\end{Theorem}

\begin{proof} As shown in Section~\ref{ss29X13.1}, without loss of
generality we can parameterise each of $\mcN^\pm$ as
$[0,\infty)\times Y$. Symmetry under interchange of $u$ and $v$,
together with the argument presented in Remark~\ref{R29X13.11} shows
that it suffices to establish that given $ b_0>0$
%and smooth initial data defined on $[0,b_0]\times Y\subset \mcN^-_0$
there exists $a_*>0$ and a solution of the vacuum Einstein equations
defined in a doubly-null coordinate system covering the set
$[0,a_*]\times [0,b_0]\times Y$.

Theorem~\ref{T23X13.2} shows indeed that there exists such a
constant $a_*$ and a set of fields \eq{phid}-\eq{psid} solving the
equations described above  with the initial data determined from the
general relativistic initial data by a standard procedure. The
theorem would immediately follow if one knew that every resulting
set of fields \eq{phid}-\eq{psid} provides a solution of the
Einstein equations. While we believe that this is the case, such a
direct proof would require a considerable amount of work.
Fortunately one can proceed in a less work-intensive manner,
adapting the idea of Luk~\cite{Luk} to use the function $u+v$ as a
tool to ``build-up'' the solution:

Let $\mcU$ be any maximal domain of existence of a solution of the
vacuum Einstein equations assuming the given initial data. (Note
that the     question, whether a \emph{unique} such maximal domain
exists is irrelevant for our purposes.) As explained in
Section~\ref{s25X13.1}, there exists a neighborhood $\mcV_0$ of
$\mcN$ in $\mcU$ on which we can introduce a coordinate system
$(u,v,x^A)$ comprising a pair of null coordinates $u$ and $v$. On
$\mcV_0$ define
\bel{18V14.11}
 t:=u+v
 \,,
\ee
then $ \nabla t$ is timelike, and hence the level sets of $t$ are
spacelike.

Define
\bean
 t_*&:= & \sup\big\{ t\ |\   \mbox{the coordinates $u$ and $v$ cover the set}
\\
 &&
   \mbox{ $([0,a_*]\times [0,b_0])\cap \{u+v< t\}$}\big\}
 \,.
\eeal{18V14.12}
It follows from Theorem~\ref{T18V14.1} that $t_*>0$.

On the set
\bel{19V.1}
   \mbox{ $\big(([0,a_*]\times [0,b_0])\cap \{u+v< t_*\}\big)\times Y$}
 \,.
\ee
we have a solution of the vacuum Einstein equations, and therefore
corresponding fields $(\varphi,\psi)$ as in \eq{phid}-\eq{psid}
calculated from the vacuum metric, with $\mathring \beta=\beta$,
$\mathring {\underline \beta}=\underline\beta$, $\mathring \sigma =
\sigma$, and $\mathring \rho =\rho$. Let us denote those fields by
$(\varphi_E,\psi_E)$. But on this set we also have a smooth solution
$(\varphi,\psi)$ of the equations described in
Appendix~\ref{Aformalism}, with initial data calculated form the
solution of the Einstein equation.
Since both fields satisfy the same system of equations and have identical initial data, uniqueness gives  %
$$
 (\varphi,\psi)
 = (\varphi_E,\psi_E)
 \,.
$$
Suppose that $t_*< a_*$, as shown in Figure~\ref{fig8}.
\begin{figure}[t]
\begin{center} {
\psfrag{n+}{ \Huge $\mathsmaller{\mcN^+}$} \psfrag{n-}{\huge
$\mcN^-$} \psfrag{ts}{\huge $t_*$} \psfrag{as}{\huge $a_*$}
\psfrag{a0}{\huge $a_0$} \psfrag{b0}{\huge $b_0$}
%\psfrag{V0}{\huge $V_0$}
%\psfrag{sig}{\huge $\Sigma$}
%
\resizebox{3in}{!}{\includegraphics{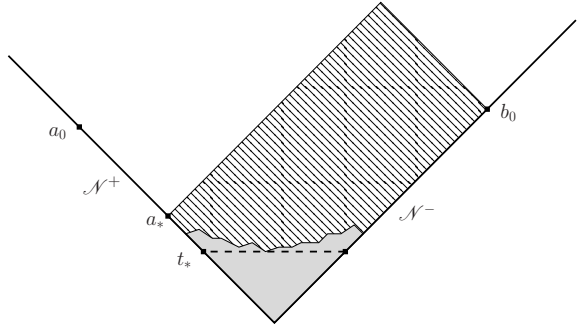}}}
\caption{The case $a_*>t_*$.}\label{fig8}
\end{center}
\end{figure}
Since  $ (\varphi,\psi)$ extend smoothly to the boundary $t=t_*$, so
do
  $ (\varphi_E,\psi_E)$. The pair $ (\varphi_E,\psi_E)$ at $t=t_*$ can be used to determine smooth Cauchy data for the vacuum Einstein equations for a Cauchy problem as shown in Figure~\ref{fig6}. The solution of this Cauchy problem allows us to extend the solution beyond $t=t_*$, contradicting the fact that $t_*$ was maximal.

The hypothesis that $a_*\le t_*\le b_0$ leads to a contradiction by
an identical argument, using instead the Cauchy problem illustrated
in Figure~\ref{fig7}, see Figures~\ref{fig9} and \ref{fig10}.
\begin{figure}[t]
\begin{center} {
\psfrag{n+}{\huge $\mcN^+$} \psfrag{n-}{\huge $\mcN^-$}
\psfrag{asts}{\huge $a_*=t_*$}
%\psfrag{as}{\huge $a_*$}
\psfrag{a0}{\huge $a_0$} \psfrag{b0}{\huge $b_0$}
%\psfrag{V0}{\huge $V_0$}
%\psfrag{sig}{\huge $\Sigma$}
%
\resizebox{3in}{!}{\includegraphics{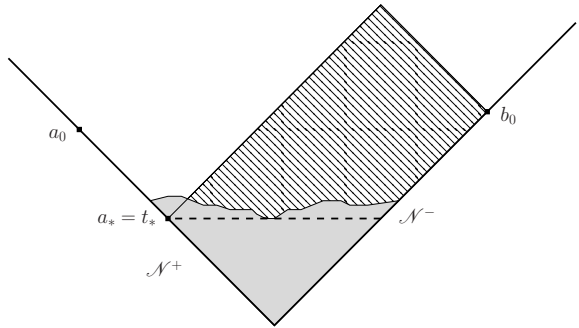}}} \caption{The case
$a_*= t_*$.}\label{fig9}
\end{center}
\end{figure}
\begin{figure}[t]
\begin{center} {
\psfrag{n+}{\huge $\mcN^+$} \psfrag{n-}{\huge $\mcN^-$}
\psfrag{ts}{\huge $t_*$} \psfrag{as}{\huge $a_*$} \psfrag{a0}{\huge
$a_0$} \psfrag{b0}{\huge $b_0$}
%\psfrag{V0}{\huge $V_0$}
%\psfrag{sig}{\huge $\Sigma$}
%
\resizebox{2,5in}{!}{\includegraphics{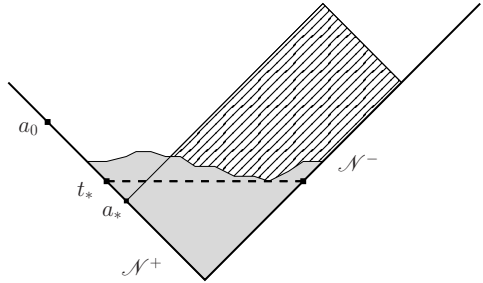}}} \caption{The case
$a_*<t_*<b_0$.}\label{fig10}
\end{center}
\end{figure}

The hypothesis that $b_0\le t_*< a_*+ b_0$ (see Figure~\ref{fig11})
\begin{figure}[ht]
\begin{center} {
\psfrag{n+}{\huge $\mcN^+$} \psfrag{n-}{\huge $\mcN^-$}
\psfrag{b0ts}{\huge $b_0=t_*$} \psfrag{as}{\huge $a_*$}
\psfrag{a0}{\huge $a_0$}
%\psfrag{V0}{\huge $V_0$}
%
\resizebox{3in}{!}{\includegraphics{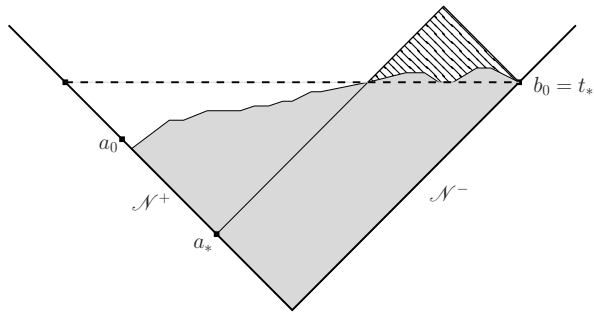}}} \caption{The case
$a_*< t_*=b_0$.}\label{fig11}
\end{center}
\end{figure}
leads to a contradiction by an identical argument, using
Theorem~\ref{T18V14.2}, compare Figure~\ref{fig5a}.

Hence $t_*=a_*+b_0$, and the result is established.
\end{proof}

\section{Einstein equations with sources satisfying wave equations}
 \label{ss15XI13.1}

The analysis of the previous section generalises immediately to
Einstein equations with matter fields satisfying wave equations,
such as the Einstein-scalar field system or the
Einstein-Yang-Mills-Higgs equations. More generally, consider a
system of equations of the form
\bel{19V14.1}
 R_{\mu\nu} -\frac R 2 g_{\mu\nu} = T_{\mu\nu}
 \,,
 \quad
 T_{\mu\nu}=T_{\mu\nu}(\Phi,\partial\Phi,g,\partial g)
 \,,
\ee
with $\nabla_\mu T^\mu{}_\nu = 0 $ whenever the matter fields $\Phi$
satisfy a set of wave-equations of the form
\bel{19V14.2}
 \Box_g \Phi = F(\Phi,\partial \Phi,g,\partial g)
 \,.
\ee
As explained in Section~\ref{s25X13.1}, one can obtain a doubly-null
system of equations from \eq{19V14.2}. The Einstein equations
\eq{19V14.1} are treated as in the vacuum case, with non-zero source
terms $J_{ijk}$ in the Bianchi equations determined by the matter
fields. This leads to an obvious equivalent of Theorem~\ref{TEin1},
the reader should have no difficulties formulating a precise
statement.

\section{Friedrich's conformal equations}
 \label{ss15XI13.2}

Let
$\tilde{g}$ be the physical space-time metric (not to be confused
with the initial data tensor field of \eq{24V14.1}), let $\Omega$ be
a function and let $g = \Omega^2 \tilde{g}$ be the unphysical
conformally rescaled counterpart of $\tilde{g}$. (To make easier
reference to~\cite{F2,F1,Friedrich:aDS,Friedrich:i0}, throughout
this section the symbol $g$ denotes the {\em unphysical} metric.)
Consider any frame field $e_k = e^{\mu}\,_k\,\partial_{x^{\mu}}$
such that the $g(e_i, e_k) \equiv g_{ik}$'s are constants, with $i$,
$k$, etc. running from zero to three. Using the Einstein vacuum
field equations, Friedrich~\cite{F2,F1} has derived a set of
equations for the fields %\ptc{watch out a renewcommand}
 \renewcommand{\DD}{\nabla}
\[
e^{\mu}\,_k,\,\,\,\,\, \Gamma_i\,^j\,_k,\,\,\,\,\, d^{i}\,_{jkl} =
\Omega^{-1}\,C^{i}\,_{jkl},\,\,\,\,\, L_{ij} = \frac{1}{2}\,R_{ij} -
\frac{1}{12}\,R \,g_{ij},
\]
\[
\Omega,\,\,\,\,\, s = \frac{1}{4}\DD_{i}\DD^{i}\Omega +
\frac{1}{24}\,R\,\Omega,
\]
where $\Gamma_i\,^j\,_k$ denotes the Levi-Civita connection
coefficients in the frame $e_k$, $\nabla_i e_k = \Gamma_i\,^j\,_k
e_j$,
 while $C^{i}\,_{jkl}$, $R_{ij}$,
and $R$ stand, respectively, for the Weyl tensor, the Ricci tensor,
and the Ricci scalar of $g$. Friedrich's ``conformal field
equations'' read
\begin{deqarr}
&&[e_{p},e_{q}] = (\Gamma_{p}\,^{l}\,_{q} -
\Gamma_{q}\,^{l}\,_{p})\,e_{l}\,,
 \label{18XI13.1a}
\\
&&e_{p}(\Gamma_{q}\,^{i}\,_{j}) - e_{q}(\Gamma_{p}\,^{i}\,_{j}) -
2\,\Gamma_{k}\,^{i}\,_{j}\,\Gamma_{[p}\,^{k}\,_{q]} +
2\,\Gamma_{[p}\,^{i}\,_{|k|} \Gamma_{q]}\,^{k}\,_{j}
\nn\\
&&\quad = 2\,g^{i}\,_{[p}\, L_{q]j} - 2\,g^{ik}\,g_{j[p}\,L_{q]k} +
\Omega \, d^{i}\,_{jpq}\,,
 \label{18XI13.1b}
\\
&&\DD_{i} d^{i}\,_{jkl} = 0\,,
 \label{18XI13.1c}
\\
&&\DD_{i} L_{jk} - \DD_{j} L_{ik} = \DD_{l} \Omega \,
d^{l}\,_{kij}\,,
\\
&&\DD_{i}\, \DD_{j} \Omega = - \Omega \, L_{ij} + s g_{ij}\,,
\\
&&\DD_{i} s = - L_{ij} \DD^{j} \Omega\,,
\\
&&6 \Omega\,s - 3\,\DD_{j}\Omega\,\DD^{j}\Omega =0
 \,.
\label{18XI13.1g} \arrlabel{18XI13.1}
\end{deqarr}
The first equation expresses the fact that the Levi-Civita
connection is torsion free; the second is the definition of the
Riemann tensor; the third is the Bianchi identity assuming that
$\tilde{g}$ is Ricci flat. The remaining equations are obtained by
algebraic manipulations from the vacuum Einstein equations, using
the conformal transformation laws for the various objects at hand.
In regions where $\Omega > 0$ the system is equivalent to the vacuum
Einstein equations~\cite{F1,F2}.

 %\ptc{an idea commented out}
%
% perhaps one can set things up so that $L_{34}=0$ as a gauge condition; rewrite the equations assuming this; the equation for $L$ then would be of the right form if $\partial _A L _{3B}$ and $\partial_A L_{4B}$ are part of the dependent fields. To obtain this one needs to differentiate some of the equations with respect to $\partial_A$; but then make sure that one can do it so that one does not get further undesirable new fields $\partial _A \cdot$

We have seen in Section~\ref{ss13XI13.1} how to bring
\eq{18XI13.1a}-\eq{18XI13.1c} to a form to which
Theorem~\ref{T23X13.2} applies. It remains to provide equations for
the fields $L_{ij}$, $s$ and $\Omega$. For this we can use a subset
of the wave equations derived in~\cite{TimConformal}:
\begin{eqnarray}
 \Box _{ g} L_{i j}&=&  4 L_{i k } L_{j}{}^{k } -  g_{i j}| L|^2
  - 2\Omega d_{i m  j}{}^{\ell }  L_{\ell }{}^{m }
 + \frac{1}{6}\nabla_{i }\nabla_{j} R
  \label{cwe1}
  \,,
  \\
  \Box_g  s  &=& \Omega| L|^2 -\frac{1}{6}\nabla_{k } R\,\nabla^{k }\Omega  - \frac{1}{6} s  R
  \label{cwe2}
  \,,
  \\
  \Box_{ g}\Omega &=& 4 s-\frac{1}{6} \Omega  R
  \label{cwe3}
%  \,,
%  \\
%  \Box^{(H)}_g d_{i jm \ell }
%  &=& \Omega d_{i jk }{}^{\alpha}d_{m \ell \alpha}{}^{k }
%   - 4\Omega d_{m k [i } {}^{\alpha}d_{j]\alpha\ell }{}^{k }
%  + \frac{1}{2}R d_{i jm \ell }
%  \label{cwe4}
%\,,
% \\
%  R^{(H)}_{\muj}[g] &=& 2L_{\muj} + \frac{1}{6} R g_{\muj}
  \label{cwe5}
  \,,
\end{eqnarray}
with the conformal gauge $R=0$. In order to control the first
derivatives of the Christoffel symbols that appear in $\Box _{ g}
L_{i j}$ we add to the above set of equations the set of equations
obtained by differentiating \eq{18XI13.1a}-\eq{18XI13.1c} with
respect to all coordinates. This collection of fields will be
referred to as \emph{Friedrich's fields}.

The wave equations \eq{cwe1}-\eq{cwe3} are rewritten as a
doubly-null system as in Section~\ref{s14XI13.1}, noting that the
inverse metric $g^{\mu\nu}=g^{ij} e^\mu{}_i  e^{\nu}{}_j$ is
directly in a doubly-null form by construction. This leads to a
system of equations to which Theorem~\ref{T23X13.2} applies provided
that the initial data have the properties required there.

For this, we will assume that the characteristic initial data on
two transversally intersecting null hypersurfaces $\mcN:=\mcN^+\cup
\mcN^-$ are smoothly conformally extendable across a boundary at
infinity. The reader is referred to~\cite{PaetzScri,ChPaetz3} for a
detailed description of this class of initial data.

Given such initial data, we can use Theorem~\ref{T18V14.1} to solve
the Einstein equations to the future of $\mcN$. The solution can be
used to provide the initial data for Friedrich's collection of
fields just described on $\mcN$. We can then extend the resulting
initial data to a hypersurface which extends beyond the conformal
boundary at infinity. Theorem~\ref{T23X13.2} guarantees the
existence of a uniform neighborhood of the extended hypersurface and
a smooth solution of the Friedrich fields there. An argument
identical to the one in the proof of Theorem~\ref{T18V14.1} shows
that the solution of the Einstein equations exists on a uniform
neighborhood of $\mcN$ in the region where $\Omega>0$. This leads
to:

\begin{Theorem}
\label{TEin2} For any set of characteristic initial data for the
vacuum Einstein equations on two transversally intersecting null
hypersurfaces $\mcN:=\mcN^+\cup \mcN^-$ which are smoothly
conformally extendable across a boundary at infinity there exists a
smooth vacuum metric defined in a future neighborhood $\mcU$ of
$\mcN$ such that the resulting space-time has a  smooth non-empty
conformal boundary at null infinity .
%the maximal domain of existence
%%
%$$
% \mcN\supset \mcN_0:=\mcN^+_0\cup \mcN^-_0
%$$
%%
%of solutions of the Raychaudhuri constraint equation.
The solution is unique up to diffeomorphism when $\mcU$ is
appropriately chosen and when appropriate initial data on
$\mcN^+\cap \mcN^-$ are given.
\end{Theorem}

An identical theorem applies to initial data given on a null cone.
When the initial data are sufficiently near to the Minkowskian ones,
all causal geodesics will be future complete in the resulting vacuum
space-time, with the null geodesics acquiring an end point on a
conformal boundary at null infinity.

\appendix

\chapter[A]{Doubly-null decompositions of the vacuum Einstein equations}\label{Aformalism}\label{C2}

The material in this appendix follows closely the presentation in
\cite{ChLengardprep}.

\section{Connection coefficients in a doubly null frame}
\label{SWcdnf}

Consider any field of vectors $e_i$, $i=1,\ldots,4$, such that
\be
(g_{ij}):=(g(e_i,e_j))=\left(\begin{array}{ccc}  \delta^a_b & 0 & 0 \\
0& 0
& -2 \\ 0 & -2 & 0
 \\ \end{array} \right)\,,
 \label{metric}
\ee
where indices $i,j$ {\em etc.\/} run from 1 to 4, while indices $a,b$
{\em etc.\/} run from 1 to 2. One therefore has $$ (g^{ij}):=
g(\theta^i,\theta^j)=\left(\begin{array}{ccc} \delta^a_b & 0 & 0 \\
0 &0 & -1/2
\\  0& -1/2 & 0  \end{array} \right)\,, $$
where
$\theta^i$ is a basis of $T^*\M$ dual to $e_i$.
If $\alpha_i$, $i=1,\cdots,4$, is a usual Lorentzian orthonormal basis
of $T\M$,
$$g(\alpha_i,\alpha_j)=\eta_{ij}=\mathrm{diag}(+1,+1,+1,-1)\,,
$$
then a basis $e_i$ as above can be constructed by setting
$$e_a=\alpha_a\,, \quad e_3=\alpha_3+\alpha_4\,, \quad
e_4=\alpha_4-\alpha_3\,.$$
Let ${\mathrm{Vol}}_g$ be the Lorentzian volume element of $g$, with the
associated completely anti-symmetric tensor $\epsilon_{ijkl}$:
$$
{\mathrm{Vol}}_g = \beta^1\land \beta^2 \land \beta^3 \land \beta^4 =
\frac{1}{4!} \epsilon_{ijkl} \;\beta^i\land \beta^j \land
\beta^k \land \beta^l\,,
$$
where $\beta^i$ is a basis dual to $\alpha_j$.
We have $\theta^3=(\beta^3+\beta^4)/2$, $\theta^4=
(\beta^4-\beta^3)/2$, $\beta^3= \theta^3-\theta^4$, $\beta^4=
\theta^3+\theta^4$, hence
$$\Vol = 2 \theta^1\land\theta^2\land\theta^3\land\theta^4=
\frac{1}{4!} \epsilon_{ijkl} \;\theta^i\land \theta^j \land
\theta^k \land \theta^l\,.$$ It
follows that in the basis $e_i$ the entries of the $\epsilon$ tensor
are zeros, twos, and their negatives:
\be
\label{epstens}
\epsilon_{1234}=2\,.
\ee
We let%\ptc{ macro ``ts'' so we can change it later if needed}
$$
 \ts=\mathrm{Vect}(\{e_1,e_2\})\,,
$$
where Vect$(X)$ denotes the
vector space spanned by the elements of the set $X$.

\newcommand{\hDlocal  }{D}%
\newcommand{\hGlocal }{\Gamma}%
For any
connection $\hDlocal  $ we define the connection coefficients $\hGlocal _i{}^j{}_k$
by the formula $$\hGlocal _i{}^j{}_k:=\theta^j(\hDlocal  _{e_i}e_k)\,,$$ so that
$$
 \hDlocal  _{e_i}e_k = \hGlocal _i{}^j{}_k e_j
 \,.
$$
The connection $\hDlocal  $ has no torsion if and only if
$$
\hDlocal  _{e_i}e_k -\hDlocal  _{e_k}e_i = [e_i,e_k]\,,
$$
and it is  metric compatible if and only if
\begin{equation}\label{metnotcomp}
\hDlocal  _{i}g_{jk}\equiv (\hDlocal  _{e_i}g)(e_j,e_k) =
-\hGlocal _{ijk}-\hGlocal _{ikj}=0\,.
\end{equation}
Here and elsewhere, $$\hGlocal _{ijk}:=g_{jm }\hGlocal _i{}^m{}_k\,.$$

The {\em null second fundamental forms} of a codimension two
submanifold $S$
 are the two symmetric tensors on $S$ defined
as\footnote{Those objects are only defined up to an overall
multiplicative function, related to the possibility of rescaling
the null vector fields $e_3$ and $e_4$; some definite choices of
this scale will be made later.} \be \chi(X,Y)=g(D_Xe_4,Y)\,,\qquad
\chib(X,Y)=g(D_Xe_3,Y)\,, \eeq where $D$ is the Levi-Civita
connection of $(\M,g)$, while $X,Y$ are tangent to $S$.  The {\em
torsion} of $S$ is a 1-form on $S$, defined for vector fields $X$
tangent to $S$ by the formula
\be
\ze(X)=-\frac{1}{2}g(D_Xe_3,e_4)=\frac{1}{2}g(D_Xe_4,e_3)
 \,.
\eeq
In the definitions above it is also assumed that $e_3$ and $e_4$ are
normal to $S$, so that $\ts$ coincides, over $S$, with the
distribution $TS$ of the planes tangent to $S$.  (Throughout the
indices are raised and lowered with the metric $g$.)

Following%
\footnote{We are grateful to Klainerman and Nicol\`{o} for
making their tex files available to us.}
Klainerman and Nicol\`{o}, we use the following labeling of the remaining Newman-Penrose
coefficients associated with the frame fields $e_i$:
\renewcommand{\hxi}{\xi}%
\renewcommand{\hxib}{\underline{\xi}}%
\renewcommand{\heta}{\eta}%
\renewcommand{\hetab}{\underline{\eta}}%
\renewcommand{\home}{\omega}%
\renewcommand{\homb}{\underline{\omega}}%
\renewcommand{\hups}{\upsilon}%
\renewcommand{\hupsb}{\underline{\upsilon}}%
\begin{deqarr}
\hxi_{a}&=&\frac{1}{2}g(\dd_{e_4}e_4,e_{a})\,,\label{NP}\\
\hxib_{a}&=&\frac{1}{2}g(\dd_{e_3}e_3,e_{a})\,,\arrlabel{NP.}\\
\heta_{a}&=&-\frac{1}{2}g(\dd_{e_3}e_{a},e_4)
=\frac{1}{2}g(\dd_{e_3}e_4,e_{a})\,,\\
\hetab_{a}&=&-\frac{1}{2}g(\dd_{e_4}e_{a},e_3)
=\frac{1}{2}g(\dd_{e_4}e_3,e_{a})\,,\\
2\home&=&-\frac{1}{2}g(\dd_{e_4}e_3,e_4)
\,, \\
2\homb&=&-\frac{1}{2}g(\dd_{e_3}e_4,e_3)
\,,\\
2\hups&=& - \frac{1}{2}g(\dd_{e_3}e_3,e_4)\,,\\
 2\hupsb&=&- \frac{1}{2}g(\dd_{e_4}e_4,e_3)\,.\label{NPlast}
\end{deqarr}
(The principle that determines
which symbols are underlined, and which are not, should be clear from
\Eq{struct1} below: all the terms
\emph{at the right hand side} of that equation have a coefficient in
front of $e_4$ which is underlined.)  The above definitions, together
with the properties of the  connection coefficients $\hGlocal _{ijk}$,
imply the following:
\renewcommand{\hchi}{\chi}%
\renewcommand{\hchib}{\underline{\chi}}%
\renewcommand{\hze}{\zeta}%
\renewcommand{\hzeb}{\underline{\zeta}}%
\renewcommand{\hxi}{\xi}%
\renewcommand{\hxib}{\underline{\xi}}%
\begin{deqarr}
\hchi_{ab}& = & \hGlocal _{ab4}= -\hGlocal _{a4b} = 2 \hGlocal _a{}^3{}_b = - 2 \hGlocal _{ab}{}^3
\,,\nnn\label{gammas0}\\
\hchib_{ab}& = & \hGlocal _{ab3}= -\hGlocal _{a3b} = 2 \hGlocal _a{}^4{}_b = - 2 \hGlocal _{ab}{}^4
\,,\label{gammasb}\\
\hze_a  & = & \hGlocal _{a}{^3}{}_3 = -\frac 12 \hGlocal _{a43} =  \hGlocal _{a4}{}^4
\,,\nnn\\
\hzeb_a  & = & \hGlocal _{a}{^4}{}_4 = -\frac 12 \hGlocal _{a34} = - \hGlocal _{a3}{}^3
\,,\nnn\\
\hxi_a  & = &  \hGlocal _4{}^3{}_{a} = - \hGlocal _{4a}{}^3 = \frac 12 \hGlocal _{4a4}= -\frac 12 \hGlocal _{44a}
\,,\nnn\\
\hxib_a  & = &  \hGlocal _3{}^4{}_{a} = - \hGlocal _{3a}{}^4 = \frac 12 \hGlocal _{3a3}= -\frac 12 \hGlocal _{33a}
\,,\nnn\\
\heta_a  & = & \hGlocal _{3}{^3}{}_a = -\frac 12 \hGlocal _{34a} =  \frac
12\hGlocal _{3a4} = -\hGlocal _{3a}{}^3
\,,\nnn\\
\hetab_a  & = & \hGlocal _{4}{^4}{}_a = -\frac 12 \hGlocal _{43a} =  \frac
12\hGlocal _{4a3} = -\hGlocal _{4a}{}^4
\,,\nnn\\
2\home  & = & \hGlocal _{4}{^3}{}_3 = -\frac 12 \hGlocal _{443} =  \hGlocal _{44}{}^4
\,,\nnn\\
2\homb  & = & \hGlocal _{3}{^4}{}_4 = -\frac 12 \hGlocal _{334} =  \hGlocal _{33}{}^3
\,,\nnn\\
2\hups  & = & \hGlocal _{3}{^3}{}_3 = -\frac 12 \hGlocal _{343} =  \hGlocal _{34}{}^4
\,,\nnn\\
2\hupsb  & = & \hGlocal _{4}{^4}{}_4 = -\frac 12 \hGlocal _{434} =  \hGlocal _{43}{}^3
\,. \arrlabel{gammas}
\end{deqarr}
This leads to
\renewcommand{\hnabb}{\nabb}%
\renewcommand{\hdddd}{\ddd}%
\renewcommand{\hchib}{\chib}%
\renewcommand{\hdivv}{\divv}%
\renewcommand{\hJ}{J}%
\begin{deqarr}
\hDlocal  _{a}e_b&=&\hnabb_{a}e_b+\frac{1}{2}\hchi_{ab}e_3+\frac{1}{2}
        \hchib_{ab}e_4\nnn\,,\\
\hDlocal  _{3}e_a&=&\hdddd_3 e_a+\heta_{a}e_3+\hxib_{a}e_4\nnn\,,\\
\hDlocal  _{4}e_a&=&\hdddd_4 e_a+\hetab_{a}e_4+\hxi_{a}e_3\,,\nnn\\
\hDlocal  _{a}e_3&=&\hchib_{a}{}^{b}e_b+\hze_{a}e_3\nnn\,,\\
\hDlocal  _{a}e_4&=&\hchi_{a}{}^{b}e_b+\hzeb_{a}e_4\nnn\,,\\
\hDlocal  _{3}e_3&=&2\xib^ae_a+2\hups e_3\,,\label{struct1f}\\
\hDlocal  _{4}e_4&=&2 \hxi^ae_a +2\hupsb e_4\,,\nnn\\
\hDlocal  _{4}e_3&=&2\hetab^{b}e_{b}+2\home e_3\,,\nnn\\
\hDlocal  _{3}e_4&=&2\heta^{b}e_{b}+2\homb e_4\,.
\arrlabel{struct1}
\end{deqarr}
Here and elsewhere, $\hnabb_{a}e_b$, $\hdddd_3 e_a$ and $\hdddd_4 e_a$ are defined as the orthogonal projection of the left-hand side of the
corresponding equation to $\ts$.
We stress that no
simplifying assumptions have been made concerning the nature of the
vector fields $e_a$, except for the orthonormality relations
\eq{metric}.

\section{The double-null decomposition of Weyl-type  tensors}
 \label{ss9V14.1}

Let $\wdd^i{}_{jkl}$  be any tensor field with the symmetries of
the Weyl tensor, \be\label{weylsyms} \wdd_{ijkl}= \wdd_{klij}\,,
\quad \wdd_{ijkl}= -\wdd_{jikl}\,, \quad g^{jk}\wdd_{ijkl}=0\,,
\quad  \wdd_{i[jkl]}= 0\,;\ee we decompose $\wdd^i{}_{jkl}$ into
its null components, relative to the null pair $\{e_3,e_4\}$, as
follows:
\begin{deqarr}
&&\ua(\wdd)(X,Y)=\wdd(X,e_3,Y,e_3)\,,\quad \alp(\wdd)(X,Y)=\wdd(X,e_4,Y,e_4)\,,\nnn\\
&&\ub(\wdd)(X)=\frac{1}{2}\wdd(X,e_3,e_3,e_4)\,,\quad
\beta(\wdd)(X)=\frac{1}{2}\wdd(X,e_4,e_3,e_4)\,,\nnn\\
&&\ro(\wdd)=\frac{1}{4}\wdd(e_3,e_4,e_3,e_4)\,,\quad
\si(\wdd):=\ro(\sthd \wdd)=\frac{1}{4}{}\sthd\wdd(e_3,e_4,e_3,e_4)\,,
 \phantom{xxxxxxx}
\eql{decomp}
\end{deqarr}
where $X,Y$ are arbitrary vector fields orthogonal to $e_3$ and $e_4$,
while  $\sthd$ denotes the space-time Hodge dual
with respect to the first two indices of $\wdd_{ijkl}$:
\bel{13XI13.2}
 \sthd\wdd_{ijkl} = \frac 12 \epsilon_{ij}{}^{mn}\wdd_{mnkl}
 \,.
\ee
The fields $\alpha$ and $\ua$ are symmetric and
traceless. From \Eq{decomp} one finds
\begin{deqarr}
\wdd_{a3b3}=\ua_{ab}\,, & \wdd_{a4b4}= \alp_{ab}\,, \nnn\\
\wdd_{a334}=2\ub_{a}\,, & \wdd_{a434}= 2\beta_{a}\,, \nnn\\
\wdd_{3434}=4\rho_{}\,, & \wdd_{ab34}= 2\sigma\epsilon_{ab}\,, \nnn\\
\wdd_{abc3}=\epsilon_{ab}{}\;{}^\star{}\ub_{c}\,, & \wdd_{abc4}=
-\epsilon_{ab}\;{}^\star{}\beta_{c}\,, \nnn\\
\wdd^a{}_{3b4}=-\rho\delta^a_{b} +\sigma\epsilon^a{}_{b}\,, &
\wdd_{abcd}= -\rho \epsilon_{ab}\epsilon_{cd}
 \,,
 \arrlabel{identi}
\end{deqarr}
where
\be\label{epsdef}\epsilon_{12}=-\epsilon_{21}=1\,,\qquad
\epsilon_{11}=\epsilon_{22}=0\,.\ee
Further, $^\star$ denotes the Hodge dual on $\ts$ with respect to the
metric induced by $g$ on $\ts$:
\be
\label{twohodge}
{}^\star{}\beta_a=\epsilon_a{}^b \beta_b\,.
\ee

\section{The double-null decomposition of the Bianchi equations}
\label{dndBe}

Recall the second Bianchi identity for the Levi-Civita connection $D$,
\bel{12XI13.10}
 D_i R_{jk\ell m} +
 D_j R_{ ki\ell m} +
 D_k R_{ij\ell m} =
 0
 \,.
\ee
Contracting $i$ with $m$ one obtains
\bel{12XI13.11}
 D_i R_{jk\ell}{}^{ i} +
 D_j R_{ k \ell  } -
 D_k R_{j\ell } =
 0
 \,.
\ee
Inserting into this equation  the expression for the Riemann tensor in terms of the Weyl and Ricci tensors,
%
%\bea
% R_{\mu\nu\sigma}{}^{\kappa}[ g] = W_{\mu\nu\sigma}{}^{\kappa} + 2\left(g_{\sigma[\mu} L_{\nu]}{}^{\kappa}  - \delta_{[\mu}{}^{\kappa}L_{\nu]\sigma} \right)
% \label{conf6}
%\;
%\end{eqnarray}
%
%
\bea
 R_{jk\ell}{}^{i}[ g] = W_{jk\ell}{}^{i} + 2\left(g_{\ell[j} L_{k]}{}^{i}  - \delta_{[j}^{i}L_{k]\ell} \right)
  \,,
 \label{conf6}
\;
\end{eqnarray}
where
\begin{equation}
  L_{ij} := \frac{1}{2}R_{ij} - \frac{1}{12} R g_{ij}
  % \frac{1}{d-2}R_{\mu\nu} - \frac{1}{2(d-1)(d-2)} R g_{\mu\nu}
 \,,
\end{equation}
we obtain
\begin{eqnarray}
\label{Weyleq}
D_i W^i{}_{jk\ell} &= &J_{jk\ell}\,,
%\\
%\hDlocal  _{[i}\wdd_{jk]lm} &=& \hJ_{ijklm}\,,\label{Biancheq2}
\end{eqnarray}
where
\bel{12XI13.14}
J_{jkl}
 =
 D_{[j} R_{ k] \ell  }
  - \frac 1{6} g_{\ell [k}D_{j]}R
  \,.
\ee
%  %and $\hJ_{ijklm}$ are
Here, and elsewhere, square brackets around a
set of $\ell$ indices denote antisymmetrization  with a multiplicative
factor $1/\ell!$.

Recall that the dual $\sthd W^i{}_{jk\ell}$ of $ W^i{}_{jk\ell}$ is defined as
$$
 \sthd W_{ijk\ell}:=  \frac 12 \epsilon_{ijmn}W^{mn}{}_{k\ell}
 \,.
$$
The well-known identity
$$
 \epsilon_{ijmn}W^{mn}{}_{k\ell} = \epsilon_{k\ell mn}W^{mn}{}_{ij}
 \,,
$$
together with \eq{Weyleq} leads to
\begin{eqnarray}
\label{dualWeyleq}
D_i {} \sthd W^i{}_{jk\ell} &= &\sthd J_{jk\ell}\,,
%\\
%\hDlocal  _{[i}\wdd_{jk]lm} &=& \hJ_{ijklm}\,,\label{Biancheq2}
\end{eqnarray}
where
\bel{13XI13.1}
\sthd J_{jmn}
 :=
 \frac 12 \epsilon_{mn}{}^{k\ell} J_{jk\ell}=
   \frac 12 \epsilon_{mn}{}^{k\ell}\big(D_{[j} R_{ k] \ell  }
  - \frac 16 g_{\ell [k}D_{j]}R\big)
  \,.
\ee

Equations~\eq{Weyleq} and \eq{dualWeyleq} are often referred to as the \emph{Bianchi equations}.

We use, as in  Section~\ref{ss9V14.1}, the symbol $d_{ijkl}$ for the Weyl tensor $W_{ijkl}$. In vacuum \eq{Weyleq} becomes
\begin{eqnarray}
\label{Weyleq2}
\hDlocal  _i(g^{im}\wdd{}_{mjkl}) = g^{im}\hDlocal  _i\wdd{}_{mjkl}  =0\,.
\end{eqnarray}
\Eq{Weyleq2} with $k=3$ and $k=4$ gives
\begin{deqarr}
\arrlabel{Weyleq2.3}
\hDlocal  _3\wdd{}_{43kl} &= & 2 h^{ab}\hDlocal  _a\wdd{}_{b3kl}  -2 \hJ_{3kl}\,,\nnn\\
 \hDlocal  _4\wdd{}_{34kl} & =& 2 h^{ab}\hDlocal  _a\wdd{}_{b4kl}   -2 \hJ_{4kl}\,,
\end{deqarr}
which will give equations for $\beta, \ub,\sigma$ and $\rho$;
we use the symbol $h$ to denote the metric induced on $\ts$ by
$g$: for all $X,Y\in T\M$,
\begin{eqnarray}
\label{hmet}
h(X,Y)=g(X,Y) +\frac 12  g(e_3,X)g(e_4,Y) +\frac 12  g(e_4,X)g(e_3,Y)\,.
\end{eqnarray}
The equations for $\alpha_{ab}$ and $\ua_{ab}$ can be obtained
from
\be\label{boleq}\hDlocal  _i d^i{}_{ab4} = \hJ_{ab4}\,. \ee For any
tensor field $T_{ab}$ we denote by $\overline{T_{ab}}$ the
symmetric traceless part of $T_{ab}$, and by $\tr T$ its trace. As already pointed out,
we
set
\bea\label{be3D}
\hdddd_3 \beta_a
 & := &
   e_3(\beta_a) - \hGlocal _3{}^b{}_a \beta_b\,,
\\
   \label{a3eq}
   \hdddd_3 \alp_{ab}  & := &
    e_3(\alp_{ab}) - \hGlocal _3{}^c{}_a \alp_{cb}-
\hGlocal _3{}^c{}_b \alp_{ac} \,.
\eea
Following Christodoulou and
Klainerman~\cite{Ch-Kl}, we use the notation $\heta \hot \beta $
for \emph{twice} the trace-free symmetric tensor product of
vectors,
\be\label{tstp}(X\hot Y)^{ab} = %\frac 12 \left\{
X^aY^b + X^bY^a - g^{ab} X_c Y^c
%\right\}
\,,\ee similarly for covectors. We let $\hnabb$ be the orthogonal
projection on
  $\ts$ of the relevant covariant derivatives in directions tangent to
  $\ts$, \emph{e.g.}
\be\label{oprcd}\hnabb_{a}e_b = \hGlocal _a{}^c{}_b e_c\,.
\ee

Tedious but otherwise straightforward, calculations allow one to
obtain the equations satisfied by the tensor field $d$, listed out
as \Eq{Bianchiid} below. A useful symmetry principle, which allows
to reduce the number of calculations by half, is to note that
under the interchange of $e_3$ with $e_4 $ the underlined rotation
coefficients \eq{gammas} are exchanged with the non-underlined
ones. On the other hand, the null components of the tensor $\wdd$
transform as follows:
\begin{eqnarray}
\alpha \leftrightarrow \ua\,, &\qquad \rho \leftrightarrow \rho\,,
\nn\\
\beta \leftrightarrow -\ub\,, &\qquad \sigma \leftrightarrow -\sigma\,.
\label{transfrules}
\end{eqnarray}
 A convenient identity in the relevant manipulations  is
\begin{equation}
\hnabb_{c}\epsilon_{ab}= -2f_c\epsilon_{ab}=-(\hze_c+\hzeb_c)\epsilon_{ab}\,,
\label{epsident}
\end{equation}
as well as
\begin{equation}
\hnabb_{c}\epsilon_{a}{}^b= 0\,. \label{epsident2}
\end{equation}

%\ptc{$\rho_3$ checked; $\rho_4$ entered; $\sigma_4$ checked;
%$\sigma_3$ entered and seems to look good; $\ubeta_4$ entered and
%checked; $\beta_4$ entered and
%checked; $\ubeta_3$ entered and
%checked; $\alpha$'s entered and checked; attention au bar qui denote
%deja une symetrisation!}
%\eject
The dynamical equations obtained by the
doubly-null decomposition of Equation \eq{Weyleq}
read\footnote{Equations~\eq{Bianchiid} are essentially a subset
of the Newman-Penrose equations written out in a tensor formalism.
The equations in~\cite{Ch-Kl} or in~\cite{KlainermanNicoloBook} can be
obtained from \eq{Bianchiid} by specialisation, and
straightforward changes of notation. We have corrected some inessential
misprints in the equations in~\cite{KlainermanNicoloBook}.}
\begin{deqarr}
\hdddd_4\ualp  & =& - \frac 12 \tr\hchi \ualp-\hnabb\hot \ubeta +
(2\home-2\hupsb)\ualp
 %-\frac 12 a(\hchi)\dual\ualp
 \nn
\\
 && -3(\overline{\hchib}\ro - \dual\overline{\hchib}\si)
-(4\hetab -\hze)\hot \ubeta+ 2\overline{\hJ({\cdot,\cdot,e_3})}
\,,\nnn\\
 \hdddd_3 \ubeta  & = & - 2\tr \hchib {}\ubeta  -\hdivv \ualp + 2\hups \ubeta -
\ualp\cdot(\heta-2\hze) %\nn \\&&
%+ 2 a(\hchib) \dual\ubeta
 + 3(-\hxib\ro + \dual\hxib\si)
 \nn
\\
 &&
-\hJ({e_3,\cdot,e_3})\,,\nnn\\
\hdddd_4  \ubeta& = &-\tr\hchi \ubeta-\hnabb \ro +\dual\hnabb
  \si+2\KlNhat{\hchib}\cdot\beta+2\home \ubeta
  +3(-\hetab\ro+\dual\hetab\si)\nn\\
  &&
  +(\hze+\hzeb)\ro - (\dual\hze+\dual\hzeb)\si -
  \hxi\cdot\ualpha
%  -a(\hchi)%_{ab} \varepsilon^{ab}\dual\ubeta
 +
 \hJ({e_4,e_3,\cdot})
 \,,
  \nnn
\\
 \dd_3\ro
 & =&
 -\frac{3}{2}\tr \hchib\ro-\hdivv
\ub-\frac{1}{2}\KlNhat{\hchi}\cdot\ua
  +(2\hze + \hzeb -2\heta)\cdot\ub\nn
\\
 &&
  +2\hxib\cdot \beta
 %+ \frac 32 a(\hchib) \sigma
 + 4
(\hups + \homb)\rho + \frac 12 \hJ_{334}\,,
\nnn\\
\dd_4\ro
& =& -\frac{3}{2}\tr \hchi\ro
+\hdivv\beta-\frac{1}{2}\KlNhat{\hchib}\cdot\alpha
  -(2\hzeb + \hze -2\hetab)\cdot\beta\nn
\\ && -2\hxi\cdot \ubeta
 %- \frac 32 a(\hchi) \sigma
 + 4
(\hupsb + \home)\rho + \frac 12 \hJ_{443}\,,
\nnn
\\
\dd_3\si & = &-\frac{3}{2}\tr \hchib{}\si-\hdivv \dual\ubeta +2
(\homb+\hups)\sigma -\frac{1}{2}{\hchit}\cdot\dual\ualpha
-2\hxib\cdot\dual\beta \nn \\ && +(\hzeb + 2\hze-2\heta
)\cdot\dual\ubeta
% - \frac 32 \rho
%\epsilon^{ab} \hchib_{ab} a(\hchib)
 - \frac 12 a(\hJ{(e_3,\cdot,\cdot)}) %\epsilon^{ab}\hJ_{3ab}
\,,\nnn
\\
\dd_4\si & = &-\frac{3}{2}\tr \hchi{}\si-\hdivv
\dual\beta
+2 (\home+\hupsb)\sigma
+\frac{1}{2}{\hchibt}\cdot\dual\alpha
-2\hxi\cdot\dual\ub
\nn \\ && +(\hze + 2\hzeb-2\hetab )\cdot\dual\beta
% + \frac 32 \rho
%\epsilon^{ab} \hchi_{ab} a(\hchi)
 - \frac 12 \epsilon^{ab}\hJ_{4ab}\,,\nnn\\
%\frac 12 a(\hJ{(e_4,\cdot,\cdot)})\,,\nnn\\
%
 \hdddd_3 \beta & = &-\tr\hchib \beta +\hnabb \ro +\dual\hnabb
  \si+2\KlNhat{\hchi}\cdot\ub+2\homb \beta
  +3(\heta\ro+\dual\heta\si)\nn\\
  &&
  -(\hze+\hzeb)\ro - (\dual\hze+\dual\hzeb)\si +
  \hxib\cdot\alp
  %-a(\hchib) \dual\beta
  -
  \hJ({e_3,e_4,\cdot)}%\theta^a
\,,\nnn
\\
 \hdddd_4 \beta  & = & - 2\tr \hchi \beta
+\hdivv \alp + 2\hupsb \beta +
\alp\cdot(\hetab-2\hzeb) %\nn \\&&
%+ 2 a(\hchi) \dual\beta
+ 3(\hxi\ro + \dual\hxi\si)
\nn\\ &&
-\hJ({e_4,\cdot,e_4})\,,\nnn
\\
\hdddd_3\alp  & =& - \frac 12 \tr\hchib \alp
+\hnabb\hot \beta +
(2\homb-2\hups)\alp
%-\frac 12 a(\hchib)\dual\alp
\nn
\\
&& -3(\overline{\hchi}\ro + \dual\overline{\hchi}\si) +
(4\heta -\hzeb)\hot \beta+ 2\overline{\hJ({\cdot,\cdot,e_4})}
\,.
\arrlabel{Bianchiid}
\end{deqarr}
For the convenience of the reader we give a summary
of notations used: The operators $\hdddd_4$ and $\hdddd_3$ are
defined as the orthogonal projections on $\ts$ of the $\hDlocal  $-covariant
derivatives along the null directions $e_3$ and $e_4$, e.g.:
$$ \hdddd_3 e_a
= \hGlocal _3{}^b{}_a e_b\,,
\qquad \hdddd_4 e_a =
\hGlocal _4{}^b{}_a e_b\,.
$$
In particular
$$
\hdddd_3\rho = \dd_3\ro= e_3(\rho)\,,\qquad \hdddd_3\sigma = \dd_3\sigma= e_3(\sigma)\,,
$$ {\em etc.}, with $\hdddd_3 \beta$ and $\hdddd_3 \alp_{ab}$ written
out explicitly in \Eq{be3D} and \Eq{a3eq}. Next, the $\nabb_a$'s
are differential operators in directions tangent to $\ts$ defined as
the orthogonal projection on $\ts$ of the relevant covariant
derivatives in directions tangent to $\ts$, \emph{cf.}
\Eq{oprcd}. We use the symbol  $\divv$  to denote the
``$\ts$-divergence'' operator: if $X=X^a e_a$ and $Y=
Y^{ab}e_a\otimes e_b$ then
$$\divv X =\hnabb_a X^a\,,
 \quad
 \divv  Y  = (\hnabb_a Y^{ab}) e_b
  \,.
$$
We have also set
$$
\hchit{}^{ab}= \hchi^{ba}\,.
$$
Next, a bar over a valence-two tensor denotes its {\em symmetric
traceless part},
{\emph{e.g.}}
$$
\KlNhat{\hchib}_{ab} = \frac 12 \left\{\hchib_{ab} + \hchib_{ba}-  g^{cd}\hchib_{cd}g_{ab}\right\}\,,
$$
while, for any two-index tensor $\chi_{ab}$,
$$
 a(\chi) = \varepsilon^{ab}\chi_{ab}
 \,.
$$
To avoid
ambiguities, we emphasize that in Equations~\eq{Bianchiid} the free
slot in $\hJ$, whenever occurring, refers to vectors in $\ts$, in
particular
\beaa
 a(\hJ{(e_4,\cdot,\cdot)}):= \epsilon^{ab}\hJ_{4ab}
\,, \quad a(\hJ{(e_3,\cdot,\cdot)}):= \epsilon^{ab}\hJ_{3ab}\,.
%  \nnn\\
\eeaa
Finally the symbol $\hot$ has been defined in \Eq{tstp}.

\section{Bianchi equations and symmetric hyperbolic systems}

Let us pass now to a specific null reformulation of the equations at hand. Let $\alpha$, $\beta$, {\em etc}, be the null components of
$\wdd$, and for reasons which will become apparent below introduce
\beqar
&\obe:=\beta\,,\quad \oube:=\ubeta\,,&
\\ & \osi:= \sigma\,,\quad \orho:=\rho\,.&
\arrlabel{ovar}\eeqar
A convenient doubly-null form of  \Eq{Bianchiid} is obtained, in vacuum, by   rewriting \eq{Bianchiid} using \eq{ovar}
as follows
\footnote{\label{footfreedom}There is a certain
amount of freedom which undifferentiated terms at the right should be decorated
with ``o'''s, which is irrelevant for our purposes in this work.}
\begin{deqarr}
 \hdddd_4\ualp + \frac 12 \tr\hchi \ualp
 & =& -\nabb\hot \ubeta   +
(2\home-2\hupsb)\ualp -3(\overline{\hchib}\ro - \dual\overline{\hchib}\si)
 \nn
\\
 &&
-(4\hetab -\hze)\hot \ubeta
+ 2\overline{\hJ({\cdot,\cdot,e_3})}
\,,\nnn
\\
 \hdddd_3 \ubeta + 2\tr \hchib {}\ubeta
 & = &  -\hdivv \ualp + 2\hups \ubeta -
\ualp\cdot(\heta-2\hze) %\nn \\&&
%+ 2 a(\hchib) \dual\ubeta
 + 3(-\hxib\ro + \dual\hxib\si)
 \nn
\\
 &&
-\hJ({e_3,\cdot,e_3})\,,\nnn
\arrlabel{Bianchiid2.1}
\\
\nydeqno
 \hdddd_4 \oubeta+\tr\hchi \oubeta  & =  &-\hnabb \oro +\dual\hnabb
  \osi+2\KlNhat{\hchib}\cdot\beta+2\home \oubeta
  +3(-\hetab\oro+\dual\hetab\osi)\nn\\
  &&
  +(\hze+\hzeb)\oro - (\dual\hze+\dual\hzeb)\osi -
  \hxi\cdot\ualpha
%  -a(\hchi)%_{ab} \varepsilon^{ab}\dual\ubeta
 +
 \hJ({e_4,e_3,\cdot})
 \,,
  \nnn
\\
D_3\osi+\frac{3}{2}\tr \hchib{}\osi & = &-\hdivv \dual\oubeta +2
(\homb+\hups)\osi -\frac{1}{2}{\hchit}\cdot\dual\ualpha
-2\hxib\cdot\dual\beta \nn \\ && +(\hzeb + 2\hze-2\heta
)\cdot\dual\oubeta
% - \frac 32 \rho
%\epsilon^{ab} \hchib_{ab} a(\hchib)
 - \frac 12 a(\hJ{(e_3,\cdot,\cdot)}) %\epsilon^{ab}\hJ_{3ab}
\,,\nnn\\
D_3\oro+\frac{3}{2}\tr \hchib\oro
& =&
 -\hdivv
\oubeta-\frac{1}{2}\KlNhat{\hchi}\cdot\ua
  +(2\hze + \hzeb -2\heta)\cdot\oubeta\nn
\\
 &&
  +2\hxib\cdot \beta
 %+ \frac 32 a(\hchib) \sigma
 + 4
(\hups + \homb)\orho + \frac 12 \hJ_{334}\,,
\nnn
 \arrlabel{Bianchiid2.2}
 \\
D_4\ro+\frac{3}{2}\tr \hchi\ro \nydeqno
& =&
 \hdivv\beta-\frac{1}{2}\KlNhat{\hchib}\cdot\alpha
  -(2\hzeb + \hze -2\hetab)\cdot\beta\nn
\\ && -2\hxi\cdot \ubeta
 %- \frac 32 a(\hchi) \sigma
 + 4
(\hupsb + \home)\rho + \frac 12 \hJ_{443}\,,
\nnn
\\D_4\si+\frac{3}{2}\tr \hchi{}\si & = &-\hdivv
\dual\beta
+2 (\home+\hupsb)\sigma
+\frac{1}{2}{\hchibt}\cdot\dual\alpha
-2\hxi\cdot\dual\ub
\nn \\ && +(\hze + 2\hzeb-2\hetab )\cdot\dual\beta
% + \frac 32 \rho
%\epsilon^{ab} \hchi_{ab} a(\hchi)
 - \frac 12 \epsilon^{ab}\hJ_{4ab}\,,\nnn
\\
 \hdddd_3 \beta+\tr\hchib \beta  & = &\hnabb \ro +\dual\hnabb
  \si+2\KlNhat{\hchi}\cdot\ub+2\homb \beta
  +3(\heta\ro+\dual\heta\si)\nn\\
  &&
  -(\hze+\hzeb)\ro - (\dual\hze+\dual\hzeb)\si +
  \hxib\cdot\alp
  %-a(\hchib) \dual\beta
  -
  \hJ({e_3,e_4,\cdot)}%\theta^a
\,,
\phantom{xxx}\nnn
 \arrlabel{Bianchiid2.3}
\\
 \hdddd_4 \obeta + 2\tr \hchi \obeta  \nydeqno
 & = &  \hdivv \alp + 2\hupsb \obeta +
\alp\cdot(\hetab-2\hzeb) %\nn \\&&
%+ 2 a(\hchi) \dual\beta
+ 3(\hxi\ro + \dual\hxi\si)
\nn\\ &&
-\hJ({e_4,\cdot,e_4})\,,\nnn\\
\hdddd_3\alp + \frac 12 \tr\hchib \alp
 & =&  \hnabb\hot \obeta +
(2\homb-2\hups)\alp
%-\frac 12 a(\hchib)\dual\alp
 -3(\overline{\hchi}\ro + \dual\overline{\hchi}\si)
\nn
\\
&&
 +
(4\heta -\hzeb)\hot \beta+ 2\overline{\hJ({\cdot,\cdot,e_4})}
\,. \label{Bianchiid2.4}
\arrlabel{Bianchiid2.4}
\end{deqarr}
We have kept the source terms $\hJ$ for future reference; however, in vacuum, which is of interest here, we have $\hJ\equiv 0$.

Let us show that the principal part of
each of the systems \eq{Bianchiid2.1}-\eq{Bianchiid2.4}
is
symmetric hyperbolic, and of the form required in our analysis, when the scalar products are appropriately
chosen.

\begin{enumerate}
\item{\bf The $(\ualpha,\ubeta)$ equations \eq{Bianchiid2.1}:} We have $\ualpha_{12}=\ualpha_{21}$,
$\ualpha_{11}=-\ualpha_{22} $ hence the pair $(\ualpha,\ubeta)$
can be parameterized by
$f=(\ualpha_{11},\ualpha_{12},\ubeta_{1},\ubeta_{2})$.
\Eq{Bianchiid2.1} can be rewritten as
\newcommand{\oB}{\mathring{B}}
\begin{equation}
\label{symhyp} A^\mu\partial_\mu f + A f = F\,,
\end{equation}
with
\begin{equation}
\label{symhypB1} A^\mu\partial_\mu  = \left(\begin{array}{rrrr}
e_4 & 0 &e_1 & -e_2 \cr 0 &e_4 &e_2 & e_1\cr e_1 &e_2 &e_3 & 0\cr
-e_2 &e_1 & 0&e_3
      \end{array}
\right) \,,
\end{equation}
which is obviously symmetric with respect to the scalar product
\begin{deqarr}
\langle f,f\rangle &=&
\ualpha_{11}^2+\ualpha_{12}^2+\ubeta_1^2+\ubeta_2^2
\\ &=& \frac 12 h^{ac}h^{bd}\ualpha_{ab}\ualpha_{cd} +
h^{ab}\ubeta_a\ubeta_b\,. \label{absp}\end{deqarr}
\item
{\bf The $(\obetab,(\osigma,\orho))$ equations \eq{Bianchiid2.2}:}
The analysis of \eq{Bianchiid2.2} is obtained by obvious renamings
and permutations from that of \eq{Bianchiid2.3}, leading to a system
with identical principal part.

\item{\bf The $((\rho,\sigma),\beta)$ equations \eq{Bianchiid2.3}:} We
set $f=((\rho,\sigma),\beta)=(\rho,\sigma,\beta_1,\beta_2)$.
\Eq{Bianchiid2.3} can be rewritten in the form \eq{symhyp}
with%\checked
\begin{equation}
\label{symhypB2} A^\mu\partial_\mu  = \left(\begin{array}{rrrr}
e_4 & 0 &-e_1 & -e_2 \cr 0 &e_4 &-e_2 & e_1\cr -e_1 &-e_2 &e_3 &
0\cr -e_2 &e_1 & 0&e_3
      \end{array}
\right) \,,
\end{equation}
which is obviously symmetric with respect to the scalar product
\begin{eqnarray*}
\langle f,f\rangle &=& \rho^2 + \sigma^2 + \beta_1^2+ \beta_2^2
\\ &=&  \rho^2+\sigma^2 +
h^{ab}\beta_a\beta_b\,.
\end{eqnarray*}%
\item{\bf The $(\obeta,\alpha)$ equations \eq{Bianchiid2.4}:}
The analysis of \eq{Bianchiid2.4} is obtained by obvious renamings
and permutations from that of \eq{Bianchiid2.1}, done above.
\end{enumerate}

\bibliographystyle{amsplain}
\bibliography{%
../../references/reffile,%
../../references/newbiblio,%
../../references/newbiblio2,%
../../references/bibl,%
../../references/howard,%
../../references/bartnik,%
../../references/myGR,%
../../references/newbib,%
../../references/Energy,%
../../references/netbiblio,%
../../references/PDE}

\end{document}